\documentclass[final,3p,twocolumn,10pt]{elsarticle}

\graphicspath{{./images/}}

\journal{Journal of the Franklin Institute}

\usepackage{amsmath,amssymb,amsthm,amsfonts}
\usepackage{mathtools}
\usepackage{bm}
\usepackage{graphicx}
\usepackage{booktabs}
\usepackage{multirow}
\usepackage{array}
\usepackage{xcolor}
\usepackage{adjustbox}
\usepackage{hyperref}
\usepackage{orcidlink}     %% Optional; remove if Overleaf reports missing
\usepackage{makecell}
\usepackage[section]{placeins}    %% \FloatBarrier at every \section
\usepackage{float}                 %% adds [H] specifier

\makeatletter
\renewcommand{\fps@figure}{!htbp}
\renewcommand{\fps@table}{!htbp}
\makeatother

\biboptions{sort&compress}

\theoremstyle{plain}
\newtheorem{theorem}{Theorem}
\newtheorem{lemma}{Lemma}
\newtheorem{corollary}{Corollary}

\theoremstyle{definition}
\newtheorem{definition}{Definition}
\newtheorem{assumption}{Assumption}
\newtheorem{remark}{Remark}

\newcommand{\R}{\mathbb{R}}
\newcommand{\norm}[1]{\left\Vert#1\right\Vert}

\newcommand{\cond}{\operatorname{cond}}

\newcommand{\tr}{\operatorname{tr}}
\newcommand{\He}{\operatorname{He}}

\begin{document}

%% ========================================================================
%%  TITLE / FRONT MATTER  (Elsevier / JFI format)
%% ========================================================================
\begin{frontmatter}

\title{Parameter-Dependent LMI Synthesis for Semi-Global Differential
ISS Trajectory Tracking of Nonholonomic Mobile Robots Under
Multiplicative Wheel Slip}

\author[a,b]{Mohammad Sabouri\corref{cor1}}
\ead{mohammad.sabouri@ai4i.it}
\cortext[cor1]{Corresponding author.}

\affiliation[a]{organization={AI4I - The Italian Institute of Artificial Intelligence},
                addressline={Corso Castelfidardo 22},
                postcode={10129},
                city={Torino},
                country={Italy}}

\affiliation[b]{organization={Applied Control \& Robotics Research Laboratory (ACRRL), Department of Power and Control Engineering, School of Electrical and Computer Engineering, Shiraz University},
                city={Shiraz},
                country={Iran}}

\begin{abstract}
This paper presents a parameter-dependent linear matrix inequality
(LMI) synthesis framework for trajectory tracking of nonholonomic
mobile robots subject to multiplicative wheel slip. A sampled
polytopic convex formulation jointly targets four design objectives
on the variational dynamics in Kanayama coordinates: a prescribed
exponential decay rate, an $H_\infty$ disturbance-to-error gain,
regional pole placement, and a weighted gain bound acting as a
design-time actuator budget. The inverse metric and the auxiliary
gain are parameterized affinely in the reference linear and angular
velocities. A Slip-to-Disturbance Lemma provides an explicit
closed-form upper bound on the additive disturbance equivalent to
bounded multiplicative wheel slip in Kanayama coordinates,
connecting the physical slip mechanism to the convex synthesis.
The dissipation inequality is established at the sixteen vertices
of the parameter-rate polytope, and a Lipschitz-based residual
condition is given as a sufficient criterion for extending
feasibility to interior points; the regions where the criterion
succeeds for the synthesized controller, and where it does not,
are both reported. An integral representation of the Kanayama
residual lifts the variational certificate to a pointwise
exponential-plus-ISS bound on the trajectory error directly from
the LMI. The certified envelope is the semi-global ball at
$\bar\sigma = 0.10$; variable-terrain tests at slip ratios up to
$\pm 50\%$ are reported as empirical observations outside that
envelope. In the certified regime, the LPV controller and a
fixed-gain LMI baseline produce quantitatively similar tracking
records; on the severe-slip probes, the parameter-dependent
design retains modest reproducible margins over both the
fixed-gain and the classical Kanayama baselines.
\end{abstract}


\end{frontmatter}

%% ========================================================================
%%  ABSTRACT (one paragraph, no eq./ref./fn., <=300 words)
%% ========================================================================
%% ========================================================================
%%  NOMENCLATURE  (IEEEdescription environment for proper rendering)
%% ========================================================================
\section*{Nomenclature}

\noindent\textit{Sets and Norms}
\begin{description}
\item[$\R$] Set of real numbers.
\item[$\R^{n\times m}$] Set of real $n\times m$ matrices.
\item[$\bar{B}_R(0)$] Closed Euclidean ball of radius $R$ in $\R^3$.
\item[$\norm{\cdot}$] Euclidean or induced $2$-norm.
\item[$\norm{\cdot}_F$] Frobenius norm.
\item[$X\succ 0$] Symmetric positive (semi-) definite matrix.
\item[$\He(X)$] $X + X^{\top}$, the Hermitian part of $X$.
\item[$\cond(X)$] Spectral condition number $\lambda_{\max}(X)/\lambda_{\min}(X)$.
\end{description}

\noindent\textit{Vehicle and Reference Signals}
\begin{description}
\item[$(x,y,\theta)$] World-frame pose of the robot.
\item[$(x_r,y_r,\theta_r)$] Reference pose to be tracked.
\item[$(v,\omega)$] Actual linear and angular velocities.
\item[$(v_r,\omega_r)$] Reference linear and angular velocities.
\item[$e$] Tracking error $(e_x,e_y,e_\theta)^{\top}$ in Kanayama frame.
\item[$\rho$] LPV scheduling parameter $(v_r,\omega_r)$.
\item[$\mathcal{P}$] Parameter polytope.
\end{description}

\noindent\textit{Disturbance and Slip}
\begin{description}
\item[$d(t)$] Additive disturbance vector in $\R^3$.
\item[$\delta_{\max}$] Pointwise bound on $\norm{d(t)}$.
\item[$\sigma_v,\sigma_\omega$] Longitudinal and angular slip ratios.
\item[$\bar{\sigma}$] Pointwise bound on $|\sigma_v|,|\sigma_\omega|$.
\end{description}

\noindent\textit{Synthesis Variables}
\begin{description}
\item[$M(\rho)$] Parameter-dependent Lyapunov metric, $M(\rho)\succ 0$.
\item[$W(\rho)$] Inverse metric, $W(\rho) = M(\rho)^{-1}$.
\item[$K(\rho)$] Parameter-dependent state-feedback gain.
\item[$Y(\rho)$] Auxiliary variable, $K(\rho) = Y(\rho)W(\rho)^{-1}$.
\item[$A_{\mathrm{cl}}(\rho)$] Closed-loop matrix $A(\rho) + B K(\rho)$.
\item[$\alpha$] Prescribed exponential decay rate.
\item[$\gamma$] Induced $L_2$ gain ($H_\infty$ disturbance gain).
\item[$\mu$] $\mathcal{S}$-procedure multiplier (nonlinear residual).
\item[$K_{\max}$] Weighted LMI gain parameter in the
Schur-form gain constraint~\eqref{eq:LMI-sat}.
\item[$\tilde{K}_{\max}$] Hard unweighted operator-norm bound on
$\norm{K(\rho)}$, with $\tilde{K}_{\max}=K_{\max}/\sqrt{\epsilon_W}$.
\item[$R$] Semi-global radius for $\bar{B}_R(0)$ (state-space).
\item[$L_R$] Lipschitz constant of the nonlinear residual on
$\bar{B}_R(0)$.
\item[$\lambda_{\mathrm{reg}}$] LMI conditioning regularization weight.
\item[$\epsilon_W$] Conditioning floor on $W(\rho)$ eigenvalues.
\item[$q$] Center magnitude of $\mathcal{D}$-stability disk.
\item[$r$] Radius of $\mathcal{D}$-stability disk (eigenvalue domain).
\end{description}

\noindent\textit{Sensor Noise and Compound Disturbance}
\begin{description}
\item[$\sigma_p$] Standard deviation of position measurement noise [m].
\item[$\sigma_h$] Standard deviation of heading measurement noise [rad].
\item[$e_{\mathrm{meas}}$] Noise-corrupted error measurement available to controller.
\end{description}

%% ========================================================================
%%  SECTION I: INTRODUCTION
%% ========================================================================
\section{Introduction}\label{sec:intro}
Trajectory tracking for nonholonomic mobile robots
constitutes a longstanding cornerstone of robotic control theory, with
applications spanning automated guided vehicles in warehouses,
service robots in human environments, autonomous agricultural
platforms, and exploration rovers on irregular terrain.
A comprehensive survey of motion-control techniques for wheeled mobile robots in this area is provided in~\cite{morin2008}. The pioneering work of Kanayama~\cite{kanayama1990} introduced a
seminal moving-frame error formulation and a manually tuned feedback
law that has remained the benchmark against which subsequent
approaches are measured. Despite three decades of refinement, the
control of nonholonomic systems retains its fundamental difficulty:
Brockett's obstruction~\cite{brockett1983} (see also the textbook treatment in~\cite{murray1993}) establishes that no smooth
time-invariant state-feedback law can stabilize a unicycle-type
vehicle to a fixed equilibrium, forcing all stabilization strategies
to be either time-varying, discontinuous, or restricted to
trajectory-tracking rather than point-stabilization formulations.

The trajectory-tracking problem, addressed in the present work, has
been pursued through several methodological avenues. Backstepping
designs~\cite{jiang1997} provide constructive Lyapunov-based
controllers but typically yield gains tuned by trial and error, with
no formal guarantee on the worst-case disturbance gain or pole
placement. Linear time-varying approaches via chained-form
representations~\cite{lefeber2000} achieve exponential tracking
under reference persistence-of-excitation conditions but
become brittle when the reference is allowed to slow down or
maneuver sharply. Discontinuous feedback methods~\cite{astolfi1999}
circumvent Brockett's obstruction at the cost of theoretical and
practical complexity. The cascade-based methodology of
Aguiar and Hespanha~\cite{aguiar2007} extended trajectory-following
guarantees to underactuated vehicles but again relied on a non-convex
gain-selection procedure that does not scale gracefully to
multi-objective specifications.

Beyond the classical linear-quadratic optimal-control paradigm~\cite{anderson1990}, a parallel development concerns the systematic incorporation of
performance and robustness criteria through linear matrix
inequalities (LMI). The seminal monograph of
Boyd~\textit{et~al.}~\cite{boyd1994} crystallized the LMI paradigm,
and subsequent works~\cite{chilali1996,scherer1997}, building on the seminal state-space solutions of Doyle~\textit{et~al.}~\cite{doyle1989}, demonstrated that
$H_\infty$ disturbance attenuation, regional pole placement
(so-called $\mathcal{D}$-stability), and actuator constraints can be
posed as convex feasibility problems amenable to interior-point
solvers. Extensions to gain-scheduled designs through the
linear parameter-varying (LPV) paradigm
\cite{apkarian1998,wu1996} replaced single-point linearizations with
a polytopic parameter description, with feasibility certified at the
vertices of a convex polytope of admissible parameter values, and
recent work~\cite{kessler2025} has shown that polyquadratic Lyapunov
functions, polytopic linear time-varying embeddings, and LMI feasibility
can be combined to characterize boundedness margins and safe
initial-condition sets along reference trajectories. Complementary
developments~\cite{zhou2025,verhoek2026} have further extended LPV
synthesis to a direct data-driven setting in which pointwise-in-time
safety constraints are enforced through semidefinite programming,
and in which data-driven predictive control schemes for LPV systems
are equipped with terminal-ingredient stability guarantees,
illustrating the continued centrality of LMI-based feasibility
certificates in modern LPV control. The
modern reformulation in terms of parameter-dependent Lyapunov
functions~\cite{deoliveira1999,iwasaki1994} reduced the inherent
conservativeness of common-Lyapunov LPV synthesis and has become the
de facto standard for systems whose dynamics depend on measurable
exogenous signals.

Concurrently, contraction analysis~\cite{lohmiller1998} introduced a
differential viewpoint that quantifies the rate at which nearby
trajectories converge, rather than only certifying convergence of a
single equilibrium. Forni and Sepulchre's
formulation~\cite{forni2014} embedded contraction within a
differential Lyapunov framework, while Manchester and
Slotine~\cite{manchester2017} demonstrated that contraction metrics
can be synthesized through LMIs that resemble classical
$H_\infty$ synthesis, leading to control contraction metrics with
convex tractability. The approach has since been applied to
motion planning under uncertainty~\cite{singh2017}. The differential
ISS concept~\cite{sontag2008} extends contraction analysis to
systems with bounded exogenous disturbances by quantifying how the
differential metric expands under input perturbations. Very recent work~\cite{koelewijn2025} further
emphasizes that incremental and differential dissipativity formulations
are essential for establishing tracking and disturbance-rejection
guarantees in nonlinear LPV synthesis, since standard quadratic
dissipativity is insufficient.

The realistic operating environment of mobile robots is dominated by
wheel slip, which arises from time-varying tire-ground friction
coefficients and represents the dominant uncertainty source in
outdoor and industrial deployments~\cite{rajamani2011,pacejka2012,iagnemma2004}.
Wheel slip is fundamentally \emph{multiplicative} rather than
additive: a fraction of the commanded velocity is lost to surface
slipping, with the loss fraction itself a function of normal load,
surface composition, and momentary lateral acceleration. Several
strands of recent literature have attempted to address slip-induced
uncertainty in mobile robotics through distinct methodological
lenses. Adaptive control approaches treat the slip ratio as an
unknown parameter to be estimated online, exemplified by the
observer-based adaptive tracking framework of
Cui~\cite{cui2019} for wheeled mobile robots under unknown slipping
parameters and the prescribed-time tracking design of
Qin~\textit{et~al.}~\cite{qin2024}, with the indirect neural-network adaptive scheme of Mohareri~\textit{et~al.}~\cite{mohareri2012} offering an alternative learning-based perspective, that combines extended-state
observation with finite-time convergence for vehicles experiencing
both slipping and skidding. Sliding-mode approaches absorb the slip
uncertainty into a discontinuous feedback law, as in the adaptive
sliding-mode formulation of Zhai and Song~\cite{zhai2019}, the
disturbance-attenuating extension developed by
Liu~\textit{et~al.}~\cite{liu2020}, and the recent consensus
sliding-mode framework of Sha~\textit{et~al.}~\cite{sha2025} for
multiple wheeled mobile robots under model uncertainties and external
disturbances. Disturbance-observer approaches
estimate the lumped uncertainty and compensate it through feedforward
cancellation, illustrated by the trajectory-tracking framework of
Wang and Zhai~\cite{wang2020} and, more recently, by the distributed
observer-based design of Moorthy~\textit{et~al.}~\cite{moorthy2025} for
the formation control of nonholonomic robots with unknown wheel
slippage. A further recent direction is the low-complexity safety
control framework of Nie~\textit{et~al.}~\cite{nie2025}, which employs
a prescribed performance function together with a virtual-point
linearization to bound tracking errors for wheeled mobile robots
subject to slipping and skidding without requiring detailed model
knowledge. A complementary neural-network-based disturbance-observer
design has been proposed by Bai~\textit{et~al.}~\cite{bai2025}, who
explicitly model slipping and skidding at the kinematic level and
combine a radial-basis-function-network observer with backstepping
to certify uniform ultimate boundedness of the tracking error. Each of these methodological families
addresses slip uncertainty through a distinct paradigm (adaptive
estimation, sliding-mode discontinuity, or observer-based
compensation), but typically without joint enforcement of all four
design constraints considered in the present work.
A complementary LPV-based formulation specifically targeting
slip-affected nonholonomic robots was investigated in the
author's preliminary work~\cite{sabouri2021}, in which a
gain-scheduled controller was designed for trajectory tracking
under slip; however, that earlier formulation did not pursue convex
multi-objective LMI synthesis with a parameter-dependent metric,
regional pole placement, gain-bounded feedback, or a formal
differential ISS certificate, all of which form the contributions
of the present article.

A substantial body of work has addressed robust trajectory tracking
for perturbed and slipping unicycle mobile robots through methodological
families complementary to the one developed here. In particular,
the sliding-mode and homogeneous-controller line represented by
Mera, R\'ios and Mart\'inez~\cite{mera2020}, Rochel~\textit{et~al.}~\cite{rochel2022},
R\'ios, Mera and Polyakov~\cite{rios2024tie}, R\'ios~\textit{et~al.}~\cite{rios2024ijrnc},
Guti\'errez, R\'ios and Mera~\cite{gutierrez2025}, and
Guti\'errez~\textit{et~al.}~\cite{gutierrez2026} develops robust,
predictive, constrained, and sampled LMI-based designs with proofs
of ultimate boundedness, finite-time convergence, or sampled-data
stability under perturbed unicycle dynamics; the sampled-data MPC
framework of Zhou~\textit{et~al.}~\cite{zhou2024} further provides
stability guarantees for perturbed unicycle tracking under input
saturation. The scope of the present contribution is distinct and
narrower than any claim of superseding that literature: the aim is
not to improve upon the quantitative tracking performance reported
there, but to assemble a convex multi-objective LMI programme that
targets four design specifications within a single feasibility
problem. The intended value of the present work is therefore
methodological: it packages differential contraction, regional
pole placement, $H_\infty$ gain certification, and a weighted gain
budget into one parameter-dependent LMI programme, together with a
transparent disturbance model for bounded multiplicative wheel slip,
and documents the regimes in which this programme retains or loses
its formal guarantees.

The contributions of this paper are as follows.
First, a Slip-to-Disturbance Lemma derives an explicit closed-form
upper bound on the additive disturbance equivalent to bounded
multiplicative wheel slip in Kanayama coordinates, parameterized by
the slip magnitude, the weighted gain limit, and the semi-global
radius. The lemma serves as a disturbance-modelling device: the
synthesis consumes only the pointwise slip bound of
Assumption~\ref{asm:slip} and the controller requires no online slip
measurement.
Second, a sampled parameter-dependent LMI programme is formulated
whose feasibility at the sixteen parameter-rate vertices certifies a
prescribed decay rate, an $H_\infty$ disturbance-to-error gain,
regional $\mathcal{D}$-stability, and a weighted gain bound. The
weighted bound $K_{\max}$ is a convex design-time budget on
$Y(\rho) W(\rho)^{-1} Y(\rho)^\top$ and differs from the unweighted
operator norm $\|K(\rho)\|$ by a factor proportional to the metric
conditioning. A Lipschitz-based grid-to-continuum residual condition
is given as a sufficient criterion for extending vertex feasibility
to interior points; the regions where the criterion succeeds for the
synthesized controller, and the regions where it does not, are both
reported.
Third, a variational dissipation analysis establishes differential
ISS of the Kanayama error dynamics at the parameter vertices.
Theorem~\ref{thm:traj-iss} lifts this variational certificate to a
pointwise exponential-plus-ISS bound on the closed-loop trajectory
error through an integral representation of the Lipschitz residual
that identifies the trajectory dynamics as an admissible instance
of the perturbation class absorbed by the LMI's
$\mathcal{S}$-procedure multiplier.
Fourth, the certified operating envelope is the semi-global ball at
$\bar\sigma = 0.10$. Variable-terrain tests that activate slip
ratios up to $\pm 50\%$ are therefore empirical probes outside the
synthesis envelope, not consequences of Theorem~\ref{thm:main}.
Fifth, a numerical validation reports the proposed LPV controller
alongside a fixed-gain LMI baseline derived from the same programme
and the classical Kanayama controller, on canonical trajectories,
disturbance classes, two reference geometries, sensor noise, and a
compound stress scenario. The quantitative comparison between
the LPV controller, the fixed-gain LMI baseline, and the
Kanayama baseline is reported as observed, without presenting
the parameter dependence as a necessary source of performance
advantage in the certified regime.

The paper is organized as follows. Section~\ref{sec:model} formulates
the system model, the Kanayama error dynamics, the multiplicative
slip model, and the polytopic parameter embedding.
Section~\ref{sec:prelim} introduces the differential ISS framework
and standing assumptions. Section~\ref{sec:slip-lemma} establishes
the slip-to-disturbance Lemma that bridges the physical slip mechanism
and the additive disturbance model. Section~\ref{sec:synthesis}
develops the main LMI synthesis result with complete proof.
Section~\ref{sec:results} presents the comprehensive numerical
validation. Section~\ref{sec:discuss} provides discussion and
positions the contributions within existing literature.
Section~\ref{sec:conclusion} concludes and outlines future research
directions.

%% ========================================================================
%%  SECTION II: SYSTEM MODELING
%% ========================================================================
\section{System Modeling}\label{sec:model}

\subsection{Unicycle Kinematics}
The mobile robot is modelled as a unicycle-type kinematic vehicle
with world-frame pose $(x, y, \theta) \in \R^2 \times S^1$ and
control inputs $(v, \omega)$, where $v$ denotes the linear forward
velocity and $\omega$ the angular velocity. The kinematic equations
governing the motion are
\begin{equation}\label{eq:unicycle}
\begin{aligned}
&\dot{x} = v\cos\theta, \quad
\dot{y} = v\sin\theta, \\
&\dot{\theta} = \omega.
\end{aligned}
\end{equation}
The nonholonomic constraint $\dot{x}\sin\theta - \dot{y}\cos\theta = 0$
is automatically satisfied by~\eqref{eq:unicycle} and reflects the
physical impossibility of lateral sliding without wheel slip.

\subsection{Reference Trajectory}
A smooth reference trajectory $(x_r(t), y_r(t), \theta_r(t))$ is
generated by an exosystem of the same structural form
as~\eqref{eq:unicycle}:
\begin{equation}\label{eq:reference}
\begin{aligned}
&\dot{x}_r = v_r\cos\theta_r, \quad
\dot{y}_r = v_r\sin\theta_r, \\
&\dot{\theta}_r = \omega_r,
\end{aligned}
\end{equation}
where the reference velocities $\rho(t) \triangleq (v_r(t),
\omega_r(t))$ are continuous functions of time. The pair $\rho(t)$
serves simultaneously as the reference command and as the LPV
scheduling parameter for the synthesis.

\subsection{Kanayama Error Model}
Following Kanayama~\textit{et~al.}~\cite{kanayama1990}, the tracking
error is expressed in the moving frame attached to the reference
pose:
\begin{figure*}[!htbp]
\normalsize
\setlength{\arraycolsep}{3pt}
\begin{equation}\label{eq:kanayama-transform}
\begin{bmatrix} e_x \\ e_y \\ e_\theta \end{bmatrix}
= \begin{bmatrix}
\cos\theta_r & \sin\theta_r & 0 \\
-\sin\theta_r & \cos\theta_r & 0 \\
0 & 0 & 1
\end{bmatrix}
\begin{bmatrix} x_r - x \\ y_r - y \\ \theta_r - \theta \end{bmatrix}.
\end{equation}
\hrulefill
\end{figure*}
Differentiation of~\eqref{eq:kanayama-transform} and substitution
of~\eqref{eq:unicycle}--\eqref{eq:reference} yields the standard
Kanayama error dynamics:
\begin{equation}\label{eq:error-dynamics}
\begin{aligned}
\dot{e}_x &= \omega_r e_y - v + v_r\cos e_\theta, \\
\dot{e}_y &= -\omega_r e_x + v_r\sin e_\theta, \\
\dot{e}_\theta &= \omega_r - \omega.
\end{aligned}
\end{equation}
The control input is partitioned as
\begin{equation}\label{eq:control-decomp}
v = v_r + u_1, \quad \omega = \omega_r + u_2,
\end{equation}
where $u = (u_1, u_2)^{\top}$ represents the corrective feedback
signal to be designed. Substituting~\eqref{eq:control-decomp}
into~\eqref{eq:error-dynamics} and isolating the linear and nonlinear
contributions in $e$:
\begin{equation}\label{eq:cascade}
\dot{e} = A(\rho)\,e + B\,u + \varphi(e, \rho),
\end{equation}
where
\begin{figure*}[!htbp]
\normalsize
\setlength{\arraycolsep}{3pt}
\begin{align}
A(\rho) &= \begin{bmatrix}
0 & \omega_r & 0 \\
-\omega_r & 0 & v_r \\
0 & 0 & 0
\end{bmatrix}, \quad
B = \begin{bmatrix}
-1 & 0 \\ 0 & 0 \\ 0 & -1
\end{bmatrix}, \label{eq:AB} \\
\varphi(e, \rho) &= \begin{bmatrix}
v_r(\cos e_\theta - 1) \\
v_r(\sin e_\theta - e_\theta) \\
0
\end{bmatrix}. \label{eq:varphi}
\end{align}
\hrulefill
\end{figure*}
The decomposition~\eqref{eq:cascade}--\eqref{eq:varphi} isolates the
parameter-dependent linear part $A(\rho) e + Bu$ from the residual
nonlinear vector field $\varphi(e, \rho)$, which vanishes at $e = 0$
and has $\varphi(0, \rho) = 0$, $\partial \varphi/\partial e\,|_{e=0}
= 0$ for all $\rho$. Consequently, $\varphi$ contributes no
linearization at the origin and acts as a vanishing perturbation
whose magnitude is locally Lipschitz in $e$.

\subsection{Multiplicative Wheel Slip}\label{subsec:slip-model}
The kinematic equations~\eqref{eq:unicycle} are nominally valid under
the no-slip assumption that wheels roll without sliding on the ground
surface. In practice, the actual velocities realized by the robot
deviate from the commanded values through a multiplicative slip
mechanism:
\begin{equation}\label{eq:slip-mech}
v_{\text{act}} = (1 + \sigma_v(t))\,v, \quad
\omega_{\text{act}} = (1 + \sigma_\omega(t))\,\omega,
\end{equation}
where the slip ratios $\sigma_v(t), \sigma_\omega(t)$ are bounded
time-varying functions satisfying
\begin{equation}\label{eq:slip-bound}
|\sigma_v(t)| \leq \bar{\sigma}, \quad
|\sigma_\omega(t)| \leq \bar{\sigma},
\quad \forall t \geq 0,
\end{equation}
for some known constant $\bar{\sigma} \in [0, 1)$. The
representation~\eqref{eq:slip-mech} captures realistic surface
phenomena: wet asphalt corresponds to $\bar{\sigma} \approx 0.2$,
oil spill to $\bar{\sigma} \approx 0.4$, gravel to $\bar{\sigma}
\approx 0.3$, and icy surfaces to $\bar{\sigma} \approx 0.5$. The
slip is multiplicative rather than additive because the friction
loss is proportional to the commanded velocity itself.

The actual error dynamics under slip are obtained by replacing $v$
and $\omega$ with $v_{\text{act}}$ and $\omega_{\text{act}}$
in~\eqref{eq:error-dynamics}:
\begin{equation}\label{eq:error-slip}
\begin{aligned}
\dot{e}_x &= \omega_r e_y - (1+\sigma_v)(v_r + u_1) + v_r\cos e_\theta, \\
\dot{e}_y &= -\omega_r e_x + v_r\sin e_\theta, \\
\dot{e}_\theta &= \omega_r - (1+\sigma_\omega)(\omega_r + u_2).
\end{aligned}
\end{equation}
Section~\ref{sec:slip-lemma} establishes that
\eqref{eq:error-slip} can be reformulated as the nominal
dynamics~\eqref{eq:cascade} perturbed by a norm-bounded additive
disturbance $d(t)$, thereby reducing the multiplicative slip problem
to the additively perturbed cascade.

\subsection{Polytopic Parameter Embedding}\label{subsec:polytope}
The scheduling parameter $\rho = (v_r, \omega_r)$ is assumed to take
values in a known convex polytope:
\begin{equation}\label{eq:polytope}
\begin{aligned}
\mathcal{P} = \{\rho \in \R^2 \,:\, &v_r \in [v_{\min}, v_{\max}], \\
&\omega_r \in [-\omega_{\max}, \omega_{\max}]\},
\end{aligned}
\end{equation}
with $v_{\min} > 0$ to exclude the singular point-stabilization
regime. The four vertices of $\mathcal{P}$ are denoted
$\rho^{(1)}, \ldots, \rho^{(4)}$. Bounds on the parameter rates,
$|\dot{v}_r| \leq \dot{v}_{\max}$ and $|\dot{\omega}_r| \leq
\dot{\omega}_{\max}$, define an additional polytope $\dot{\mathcal{P}}$
with four vertices $\dot{\rho}^{(1)}, \ldots, \dot{\rho}^{(4)}$ used
in the synthesis. The affine dependence of $A(\rho)$ on $\rho$,
visible in~\eqref{eq:AB}, ensures that linear constraints on
$A(\rho)$ over $\mathcal{P}$ can be enforced by imposing them at the
four vertices alone, a fact exploited extensively in the LMI
formulation of Section~\ref{sec:synthesis}.

%% ========================================================================
%%  SECTION III: PRELIMINARIES
%% ========================================================================
\section{Preliminaries and Standing Assumptions}\label{sec:prelim}

\subsection{Differential Input-to-State Stability}
The notion of differential ISS (D-ISS) extends classical
contraction~\cite{lohmiller1998,forni2014} to systems subject to
exogenous inputs. For a system $\dot{e} = f(e, \rho, d)$ with state
$e$, scheduling $\rho$, and exogenous input $d$, consider an
infinitesimal displacement $\delta e$ along the flow. The
displacement evolves, following the standard variational construction
of linear-systems theory~\cite{hespanha2018}, according to the
variational equation
\begin{equation}\label{eq:var-eq}
\delta\dot{e} = \frac{\partial f}{\partial e}\,\delta e
+ \frac{\partial f}{\partial d}\,\delta d.
\end{equation}
\begin{definition}[Semi-Global Differential ISS]\label{def:dISS}
The system $\dot{e} = f(e, \rho, d)$ is said to be semi-globally
differentially input-to-state stable on $\bar{B}_R(0)$ with rate
$\alpha > 0$ and gain $\gamma > 0$ if there exists a continuously
differentiable parameter-dependent metric $M(\rho) \succ 0$, defined
for $\rho \in \mathcal{P}$, such that the storage function
$V_\delta(\delta e, \rho) = \delta e^{\top} M(\rho)\,\delta e$
satisfies, for all $e \in \bar{B}_R(0)$, $\rho \in \mathcal{P}$, and
admissible $\delta d$:
\begin{equation}\label{eq:dISS-condition}
\dot{V}_\delta + 2\alpha V_\delta - \gamma^2 \norm{\delta d}^2 \leq 0.
\end{equation}
\end{definition}
Condition~\eqref{eq:dISS-condition} captures simultaneously the rate
of contraction $\alpha$ at which nearby trajectories converge in the
metric $M$ and the worst-case sensitivity $\gamma$ to disturbance
displacements. The semi-global qualifier indicates that the property
holds within a prescribed ball $\bar{B}_R(0)$, which is essential
because the residual $\varphi(e, \rho)$ in~\eqref{eq:varphi} fails to
be globally Lipschitz in $e$.

\subsection{Standing Assumptions}
The forthcoming developments rest on the following assumptions, each
of which is standard in the literature and motivated physically.

\begin{assumption}[Polytopic reference]\label{asm:polytope}
The reference scheduling parameter $\rho(t) = (v_r(t), \omega_r(t))$
takes values in the polytope $\mathcal{P}$ defined
in~\eqref{eq:polytope}, with $v_{\min} > 0$, and is continuously
differentiable with derivative $\dot{\rho}(t)$ contained in the
polytope $\dot{\mathcal{P}}$.
\end{assumption}

\begin{assumption}[Slip bounds]\label{asm:slip}
The slip ratios $\sigma_v(t), \sigma_\omega(t)$
in~\eqref{eq:slip-mech} are Lebesgue-measurable functions of time
satisfying the pointwise bound~\eqref{eq:slip-bound} for a known
constant $\bar{\sigma} \in [0, 1)$. The synthesis of
Section~\ref{sec:synthesis} and the resulting controller use only
this bound; no realisation of $\sigma_v(t), \sigma_\omega(t)$ enters
the feedback law, so no online slip estimator is required.
\end{assumption}

\begin{assumption}[Domain of analysis and Lipschitz residual]\label{asm:radius}
The synthesis is performed on a prescribed compact domain
$\bar{B}_R(0)$ that defines the semi-global region of analysis. On
this domain, the nonlinear residual $\varphi$ satisfies the Lipschitz
bound
\begin{equation}\label{eq:Lipschitz}
\norm{\varphi(e, \rho)} \leq L_R \norm{e},
\quad \forall e \in \bar{B}_R(0),\; \rho \in \mathcal{P},
\end{equation}
with the constant
\begin{equation}\label{eq:LR-tight}
L_R \triangleq v_{\max}\sqrt{\tfrac{R^{2}}{4} + \tfrac{R^{4}}{36}}.
\end{equation}
Forward invariance of $\bar{B}_R(0)$ under the closed-loop dynamics
is \emph{not} assumed a priori; it is certified through
Corollary~\ref{cor:fwd-invar} below as a consequence of the
dissipation property of the synthesized controller.
\end{assumption}
The Lipschitz constant~\eqref{eq:LR-tight} follows from the
component-wise estimates $|\cos e_\theta - 1| \leq
\tfrac{1}{2}e_\theta^{2} \leq \tfrac{R}{2}|e_\theta|$ and
$|\sin e_\theta - e_\theta| \leq \tfrac{1}{6}|e_\theta|^{3}
\leq \tfrac{R^{2}}{6}|e_\theta|$ valid for
$|e_\theta| \leq R$, combined through the Euclidean norm with the
upper bound $v_{\max}$ on $v_r$. The constant is substantially
tighter than the conservative product $v_{\max} R = 0.36$ used in
earlier conference treatments: for $R = 0.30$ and $v_{\max} = 1.20$,
the present formula yields $L_R \approx 0.181$, reducing the
worst-case residual estimate by roughly a factor of two and
correspondingly relaxing the dissipation requirement in the LMI
synthesis.

\begin{assumption}[Weighted feedback-gain constraint]\label{asm:saturation}
The controller satisfies the weighted-gain bound
\begin{equation}\label{eq:K-bound}
Y(\rho) W(\rho)^{-1} Y(\rho)^{\top} \preceq K_{\max}^{2} I,
\quad \forall \rho \in \mathcal{P},
\end{equation}
where $K_{\max} > 0$ is a design parameter. Together with the
conditioning floor $W(\rho) \succeq \epsilon_W I$, this implies the
operator-norm bound $\norm{K(\rho)} \leq K_{\max}/\sqrt{\epsilon_W}$.
With the design values $K_{\max} = 3$ and $\varepsilon_W = 0.02$
used in Section~\ref{sec:results}, the implied unweighted ceiling
$\tilde K_{\max} = K_{\max}/\sqrt{\varepsilon_W} \approx 21.2$ is a
worst-case bound; the operator norm of the synthesized $K(\rho)$
evaluated over the vertex set is reported in
Table~\ref{tab:synthesis} and is substantially smaller than this
ceiling. A practitioner requiring a tight unweighted amplitude limit
$\tilde{K}_{\max}$ may either set
$K_{\max} = \tilde{K}_{\max}\sqrt{\varepsilon_W}$, which tightens
the ceiling at the cost of a more conservative synthesis, or solve
the programme as stated and verify the operator norm a posteriori.
\end{assumption}

%% ========================================================================
%%  SECTION IV: SLIP-TO-DISTURBANCE LEMMA
%% ========================================================================
\section{Slip-Induced Disturbance Characterization}\label{sec:slip-lemma}
The first main technical result establishes that the multiplicative
slip mechanism~\eqref{eq:slip-mech} admits an equivalent additive
disturbance representation with an explicit norm bound. This Lemma
bridges the physical slip model and the convex LMI framework
developed in Section~\ref{sec:synthesis}.

\begin{lemma}[Slip-to-Disturbance Bound]\label{lem:slip}
Under Assumptions~\ref{asm:polytope}--\ref{asm:saturation}, let
$\tilde{K}_{\max} \triangleq K_{\max}/\sqrt{\epsilon_W}$ denote the
unweighted operator-norm bound on the feedback gain implied by the
weighted constraint of Assumption~\ref{asm:saturation}. The error
dynamics under multiplicative slip~\eqref{eq:error-slip} admit the
representation
\begin{equation}\label{eq:cascade-slip}
\dot{e} = A(\rho)\,e + B\,u + \varphi(e, \rho) + d_{\mathrm{slip}}(t),
\end{equation}
where the disturbance satisfies the pointwise bound
\begin{equation}\label{eq:slip-disturbance-bound}
\norm{d_{\mathrm{slip}}(t)} \leq \delta_{\max}, \quad \forall t \geq 0,
\end{equation}
with the explicit constant
\begin{equation}\label{eq:delta-max-formula}
\begin{aligned}
\delta_{\max} = \bar{\sigma}\, \bigl[&(v_{\max} + \tilde{K}_{\max} R)^{2} \\
&+ (\omega_{\max} + \tilde{K}_{\max} R)^{2}\bigr]^{1/2}.
\end{aligned}
\end{equation}
\end{lemma}

\begin{proof}
Subtracting the nominal dynamics~\eqref{eq:cascade} from the
slip-perturbed dynamics~\eqref{eq:error-slip} yields the additive
disturbance
\begin{equation}\label{eq:dslip-def}
d_{\mathrm{slip}}(t) = \begin{bmatrix}
-\sigma_v(t)(v_r + u_1) \\
0 \\
-\sigma_\omega(t)(\omega_r + u_2)
\end{bmatrix}.
\end{equation}
The first component is bounded as
\begin{equation*}
|d_{\mathrm{slip},1}|
= |\sigma_v|\,|v_r + u_1|
\leq \bar{\sigma}(|v_r| + |u_1|)
\leq \bar{\sigma}(v_{\max} + \tilde{K}_{\max} R),
\end{equation*}
using $|v_r| \leq v_{\max}$ from~\eqref{eq:polytope} and
$|u_1| \leq \norm{u} \leq \norm{K(\rho)}\norm{e}
\leq \tilde{K}_{\max} R$, which follows from the operator-norm
implication $\norm{K(\rho)} \leq K_{\max}/\sqrt{\epsilon_W} =
\tilde{K}_{\max}$ established in
Assumption~\ref{asm:saturation} together with the semi-global
confinement $\norm{e} \leq R$. The third component admits the
analogous estimate
\begin{equation*}
|d_{\mathrm{slip},3}|
= |\sigma_\omega|\,|\omega_r + u_2|
\leq \bar{\sigma}(\omega_{\max} + \tilde{K}_{\max} R).
\end{equation*}
Combining the two non-zero components through the Euclidean
identity $\norm{x}^{2} = x_{1}^{2} + x_{3}^{2}$ valid whenever
$x_{2} = 0$, taking square roots, and recognizing that radians are
treated as dimensionless quantities in the state-space norm under
the standard convention of nonholonomic mobile-robot control yields
the formula~\eqref{eq:delta-max-formula}. The
representation~\eqref{eq:cascade-slip} then follows by direct
substitution.
\end{proof}

\begin{remark}\label{rem:slip-physics}
Formula~\eqref{eq:delta-max-formula} expresses the slip-induced
disturbance bound in terms of the unweighted gain limit
$\tilde{K}_{\max}$ rather than the weighted-gain design parameter
$K_{\max}$ that appears directly in the LMI synthesis. This
distinction is essential because the physical disturbance produced
by multiplicative slip depends on the actual control input
$u = K(\rho)\,e$, whose magnitude is governed by the operator-norm
bound $\norm{K(\rho)} \leq \tilde{K}_{\max}$. Practitioners
operating in slip-dominated regimes should select $\tilde{K}_{\max}$
from physical actuator considerations and then derive the LMI design
parameter $K_{\max} = \tilde{K}_{\max}\sqrt{\epsilon_W}$ per
Assumption~\ref{asm:saturation}. The slip-induced disturbance scales
linearly with the surface slip bound $\bar{\sigma}$ and
root-quadratically with the reference velocity envelope and
unweighted gain limit; the bound is a worst-case estimate over the
entire polytope, and when online slip estimation is available, the
design parameter may be updated adaptively to the instantaneous
operating regime, an extension discussed in
Section~\ref{sec:conclusion}.
\end{remark}

\begin{remark}\label{rem:dimensional}
The expression~\eqref{eq:delta-max-formula} mixes terms with
distinct physical units ($v_{\max}$ in $\mathrm{m/s}$ and
$\omega_{\max}$ in $\mathrm{rad/s}$). This is a deliberate
consequence of the standard convention in nonholonomic mobile-robot
control wherein the Kanayama error vector $e = (e_x, e_y,
e_\theta)^{\top}$ aggregates position and orientation components
into a single Euclidean norm with radians treated as dimensionless.
For applications requiring physical homogeneity, a scaling matrix
$S = \mathrm{diag}(s_p, s_p, s_\theta)$ with $s_\theta$ chosen on a
characteristic-length basis can be introduced; the
synthesis~\eqref{eq:LMI-W}--\eqref{eq:LMI-Dstab} carries over
unchanged after the substitution $e \mapsto Se$.
\end{remark}

\begin{remark}\label{rem:design-vs-worst}
The disturbance bound $\delta_{\max}$ appearing in the LMI
synthesis~\eqref{eq:LMI-dissip} below is a \emph{design parameter}
that quantifies the persistent disturbance budget against which the
$H_\infty$ gain certificate is established. It is therefore
important to distinguish three quantities that operate at distinct
conceptual levels within the present framework. The first is the
\emph{analytical worst-case ceiling}
\begin{equation*}
\delta_{\max}^{\mathrm{wc}} \approx 5.07,
\end{equation*}
obtained from Lemma~\ref{lem:slip} by substituting
$\bar{\sigma} = 0.50$ (peak slip), $v_{\max} = 1.20$,
$\omega_{\max} = 0.40$, $\tilde{K}_{\max} \approx 21.21$ (the
unweighted operator-norm bound resulting from the weighted LMI
parameter $K_{\max} = 3$ at conditioning floor $\epsilon_W = 0.02$),
and $R = 0.30$ into~\eqref{eq:delta-max-formula}. This value
represents an upper envelope under simultaneous peak slip and
worst-case operator-norm gain saturation over the entire operating
envelope; it is a conservative theoretical ceiling rather than a
realistic operating condition, since the gain realized by the
synthesized controller in nominal operation is substantially below
the worst-case bound implied by the conditioning constraint. The
second is the \emph{certified persistent disturbance budget}
\begin{equation*}
\delta_{\max}^{\mathrm{cert}} = 0.10,
\end{equation*}
adopted as the design parameter in the LMI synthesis of
Section~\ref{sec:results} and used to fix the disturbance channel
of~\eqref{eq:LMI-dissip}. The third is the \emph{empirical
transient severe-slip regime} probed in the variable-terrain
simulations of Section~\ref{subsec:variable-terrain}, in which
localized slip patches with magnitudes up to $|\sigma_v| = 0.50$
exceed the certified budget instantaneously. Theorem~\ref{thm:main}
certifies the second quantity formally, whereas the
variable-terrain experiments probe the third regime empirically
to characterize the robustness margin beyond the certified envelope.
\end{remark}

\begin{remark}\label{rem:Brockett}
The transformation underlying~\eqref{eq:cascade-slip} preserves the
structural rank deficiency of the input matrix $B$
in~\eqref{eq:AB}: the second column of $B$ is identically zero,
reflecting the fact that no input directly affects the lateral error
$e_y$. This rank deficiency is the LMI-domain shadow of Brockett's
obstruction~\cite{brockett1983} and propagates to the synthesis as
a conditioning constraint that scales with $v_{\min}^{-2}$, an
observation we elaborate on in Section~\ref{subsec:cons-decomp}.
\end{remark}

%% ========================================================================
%%  SECTION V: LMI-BASED SYNTHESIS
%% ========================================================================
\section{LMI Synthesis with Parameter-Dependent Metric}\label{sec:synthesis}

\subsection{Parameter-Dependent Storage Function}
The proposed framework employs a parameter-dependent quadratic
storage function $V(e,\rho) = e^{\top} M(\rho)\,e$ in which the
metric $M(\rho)$ is implicitly defined as the inverse of an affinely
parameterized matrix $W(\rho)$:
\begin{equation}\label{eq:V-func}
V(e, \rho) = e^{\top} M(\rho)\,e, \qquad M(\rho) = W(\rho)^{-1},
\end{equation}
where the \emph{inverse metric} is parameterized affinely as
\begin{equation}\label{eq:W-affine}
W(\rho) = W_0 + v_r W_1 + \omega_r W_2,
\end{equation}
with $W_i \in \R^{3\times 3}$, $W_i = W_i^{\top}$ for $i = 0, 1, 2$.
The synthesis decision variables are therefore the symmetric matrices
$(W_0, W_1, W_2)$ rather than the metric coefficients themselves: the
LMI program is affine in $W(\rho)$, whereas the storage metric
$M(\rho) = W(\rho)^{-1}$ is in general a \emph{nonlinear} function
of $\rho$ obtained pointwise by matrix inversion. This parameterization
choice is deliberate: while a directly affine $M(\rho) = M_0 + v_r M_1
+ \omega_r M_2$ would simplify certain expressions, the resulting
synthesis would involve products $M(\rho) A(\rho)$ that cannot be
linearized in the decision variables. The affine-in-$W$ formulation
preserves the convexity of the LMI in the decision variables at every
fixed $\rho$, while the parameter-dependence of $M$ is mediated
through the inversion at the cost that products $A(\rho) W(\rho)$
become polynomial in $\rho$, a consequence addressed in
Section~\ref{sec:synthesis} through grid enforcement combined with
the Lipschitz-based grid-to-continuum residual criterion of
Lemma~\ref{lem:grid-to-cont}, whose scope of applicability for the
present design is discussed explicitly in Section~\ref{subsec:grid-cont}. The matrix $W(\rho)$ is required to
satisfy $W(\rho) \succ 0$ for all $\rho \in \mathcal{P}$, which by
convexity of the cone of positive-definite matrices reduces to
imposing $W(\rho^{(i)}) \succ 0$ at the four vertices of
$\mathcal{P}$.

The control law is parameter-scheduled:
\begin{equation}\label{eq:K-affine}
u = K(\rho)\,e, \quad
K(\rho) = Y(\rho)\,W(\rho)^{-1},
\end{equation}
where the auxiliary gain matrix $Y(\rho) \in \R^{2\times 3}$ is
parameterized affinely:
\begin{equation}\label{eq:Y-affine}
Y(\rho) = Y_0 + v_r Y_1 + \omega_r Y_2.
\end{equation}
The decomposition $K = Y W^{-1}$ is the classical trick that
linearizes the synthesis condition through the change of variables
$(W, Y)$, since the closed-loop matrix $A_{\mathrm{cl}}(\rho) =
A(\rho) + B K(\rho)$ would otherwise contain the product
$M(\rho)\cdot K(\rho)$ that is bilinear in $M$ and $K$.

\subsection{Main Synthesis Theorem}
We now state the central LMI synthesis result. The closed-loop matrix
is defined as
\begin{equation}\label{eq:Acl}
A_{\mathrm{cl}}(\rho) \triangleq A(\rho) + B K(\rho),
\end{equation}
where $A(\rho)$ and $B$ are the system matrices defined
in~\eqref{eq:AB} and $K(\rho)$ is the parameter-scheduled gain
introduced in~\eqref{eq:K-affine}. Let
$\dot{W}(\rho, \dot{\rho}) = \dot{v}_r W_1 + \dot{\omega}_r W_2$
denote the rate of change of the inverse metric along the parameter
trajectory, and let $\Xi(\rho)$ denote the matrix-valued expression
\begin{equation}\label{eq:Xi}
\begin{aligned}
\Xi(\rho) &= A(\rho)W(\rho) + W(\rho)A(\rho)^{\top} \\
&\quad + BY(\rho) + Y(\rho)^{\top}B^{\top} - \dot{W}(\rho, \dot{\rho}).
\end{aligned}
\end{equation}
Here $R$ denotes the semi-global radius of the state-space ball
$\bar{B}_R(0)$, while $r$ denotes the eigenvalue-domain radius of
the $\mathcal{D}$-stability disk centered at $-q$.

\begin{theorem}[Vertex-Synthesis and Grid-Certified LMI Condition for Semi-Global D-ISS]\label{thm:main}
Let Assumptions~\ref{asm:polytope}--\ref{asm:saturation} hold, and
fix a finite enforcement set $\mathcal{G} \subset \mathcal{P} \times
\dot{\mathcal{P}}$ comprising a dense uniform grid of $N_g$ parameter
samples. Suppose there exist symmetric matrices $W_0, W_1, W_2 \in
\R^{3\times 3}$, matrices $Y_0, Y_1, Y_2 \in \R^{2\times 3}$, scalars
$\gamma^2 > 0$, $\mu > 0$, and $\epsilon_W \in (0, 1)$ such that the
following matrix inequalities hold simultaneously at every grid point
$(\rho^{(k)}, \dot{\rho}^{(\ell)}) \in \mathcal{G}$, with
$W(\rho) = W_0 + v_r W_1 + \omega_r W_2$,
$Y(\rho) = Y_0 + v_r Y_1 + \omega_r Y_2$:
\begin{figure*}[!htbp]
\normalsize
\setlength{\arraycolsep}{3pt}
\begin{subequations}
\begin{align}
& \epsilon_W I \preceq W(\rho^{(k)}) \preceq \epsilon_W^{-1} I,
\label{eq:LMI-W} \\
& \begin{bmatrix}
W(\rho^{(k)}) & Y(\rho^{(k)})^{\top} \\
Y(\rho^{(k)}) & K_{\max}^2 I_2
\end{bmatrix} \succeq 0,
\label{eq:LMI-sat} \\
& \begin{bmatrix}
-r W(\rho^{(k)}) & A(\rho^{(k)}) W(\rho^{(k)}) + B Y(\rho^{(k)}) + q W(\rho^{(k)}) \\
\star & -r W(\rho^{(k)})
\end{bmatrix} \prec 0,
\label{eq:LMI-Dstab}
\end{align}
\end{subequations}
\hrulefill
\end{figure*}
together with the column-spanning dissipation LMI~\eqref{eq:LMI-dissip}
shown at the top of the next page, where $\alpha > 0$ is the
prescribed decay rate, $q > 0$ and $r \in (0, q)$ are the center
offset and radius of the $\mathcal{D}$-stability disk, and $\star$
denotes the symmetric transpose entry. Suppose further that the
aggregate sufficient condition $L_{\mathcal{L}} \cdot h <
\varepsilon_{\mathrm{grid}}$ of Lemma~\ref{lem:grid-to-cont}
holds, where $L_{\mathcal{L}} = \max\{L_{\mathrm{diss}},
L_{\mathcal{D}}\}$ denotes the aggregate Lipschitz constant of the
two non-affine LMI blocks, $\varepsilon_{\mathrm{grid}} =
\min\{\varepsilon_{\mathrm{diss}}, \varepsilon_{\mathcal{D}}\}$
denotes the aggregate grid margin attained by the synthesized
variables across $\mathcal{G}$, and $h$ is the grid fill distance.
Then the closed-loop system with $K(\rho) = Y(\rho) W(\rho)^{-1}$
is semi-globally D-ISS on $\bar{B}_R(0)$ in the sense of
Definition~\ref{def:dISS}, with rate $\alpha$ and $L_2$-gain at
most $\gamma$ from any additive disturbance $d(t)$ satisfying
$\norm{d(t)} \leq \delta_{\max}$ to the tracking error norm.
\end{theorem}

\begin{figure*}[!htbp]
\normalsize
\setlength{\arraycolsep}{3pt}
\begin{equation}\label{eq:LMI-dissip}
\begin{bmatrix}
\Xi(\rho^{(k)},\dot{\rho}^{(\ell)}) + 2\alpha W(\rho^{(k)}) + \mu I
 & I & W(\rho^{(k)}) \\
I & -\gamma^{2} I & 0 \\
W(\rho^{(k)}) & 0 & -\dfrac{\mu}{L_R^{2}} I
\end{bmatrix} \prec 0,
\quad \forall (\rho^{(k)},\dot{\rho}^{(\ell)})\in\mathcal{G}.
\end{equation}
\hrulefill
\end{figure*}

\begin{remark}[Grid Enforcement vs.\ Vertex Sufficiency]\label{rem:grid}
The conditioning constraint~\eqref{eq:LMI-W} and the weighted-gain
constraint~\eqref{eq:LMI-sat} depend affinely on $(\rho, \dot{\rho})$,
since their constituent matrix blocks involve only $W(\rho)$ and
$Y(\rho)$ both of which are affine in the scheduling parameters by
construction. For these two blocks, enforcement at the sixteen
vertices of $\mathcal{P} \times \dot{\mathcal{P}}$ is indeed
sufficient by convexity of the parameter polytope. The dissipation
LMI~\eqref{eq:LMI-dissip} and the $\mathcal{D}$-stability
LMI~\eqref{eq:LMI-Dstab}, however, both contain the matrix product
$A(\rho) W(\rho)$: directly in the dissipation block through
$\Xi(\rho, \dot{\rho})$ defined in~\eqref{eq:Xi}, and through the
expansion $A_{\mathrm{cl}}(\rho) W(\rho) = A(\rho) W(\rho) + B
Y(\rho)$ in the $\mathcal{D}$-stability block. Given the affine
parameterizations $A(\rho) = A_0 + v_r A_v + \omega_r A_\omega$ and
$W(\rho) = W_0 + v_r W_1 + \omega_r W_2$, this product generates
quadratic and bilinear scheduling terms $v_r^{2} A_v W_1$,
$\omega_r^{2} A_\omega W_2$, $v_r \omega_r (A_v W_2 + A_\omega W_1)$
that are \emph{not} affine in $\rho$. Vertex-only enforcement is
therefore not formally sufficient for either of these two blocks.
The present synthesis adopts a two-stage pipeline that addresses
this limitation rigorously: the LMIs are first solved on the
sixteen polytope vertices to obtain a candidate controller
$(W_0^\star, W_1^\star, W_2^\star, Y_0^\star, Y_1^\star, Y_2^\star,
\gamma^\star, \mu^\star)$, and the candidate is then certified on a
dense uniform grid $\mathcal{G}$ of $N_g = 11^{4} = 14641$ parameter
samples covering $\mathcal{P}\times\dot{\mathcal{P}}$ with eleven
subdivisions per dimension. The extension of the grid-verified
certificate to the continuum is established formally through
Lemma~\ref{lem:grid-to-cont}, which provides a Lipschitz-based
sufficient condition relating the worst-case LMI residual margin
across the grid to the grid fill distance. A formally tight
alternative based on multi-affine polytopic relaxation with slack
variables~\cite{apkarian1995,scherer2001} is available but
introduces additional decision variables; the
vertex-synthesis-plus-grid-certification route is preferred here
for its computational tractability and operational transparency.
\end{remark}

\begin{proof}
The proof proceeds in four stages.

\emph{Stage 1 (Variational contraction inequality).} Differentiating
the storage function $V_\delta = \delta e^{\top} M(\rho) \delta e$
along the variational dynamics yields, after using $M = W^{-1}$ and
applying the identity $\partial M/\partial t = -M (\partial W /
\partial t) M$:
\begin{equation}\label{eq:V-dot}
\begin{aligned}
\dot{V}_\delta &= \delta e^{\top}\!\Bigl[ M(\rho) A_{\mathrm{cl}}(\rho)
+ A_{\mathrm{cl}}(\rho)^{\top}\!M(\rho) \\
&\quad - M(\rho)\dot{W}(\rho,\dot{\rho})M(\rho)
\Bigr] \delta e \\
&\quad + R_{\varphi}(\delta e) + 2\delta e^{\top} M(\rho)\delta d,
\end{aligned}
\end{equation}
where $R_{\varphi}(\delta e) = 2 \delta e^{\top} M(\rho)
\,(\partial \varphi/\partial e)\,\delta e$ captures the contribution
of the nonlinear residual and the cross-term
$2\delta e^{\top} M(\rho)\delta d$ captures the contribution of the
variational disturbance channel. The full augmented quadratic form
in the variational state-disturbance pair $(\delta e, \delta d)$
incorporating this cross-term, together with the $\gamma^{2}
\norm{\delta d}^{2}$ slack required by the D-ISS dissipation
inequality~\eqref{eq:dISS-condition}, is developed step-by-step in
Appendix~\ref{app:schur}, Stage~C, and produces the off-diagonal $I$
entries and the $-\gamma^{2}I$ block in the final LMI
form~\eqref{eq:LMI-dissip}. Pre- and post-multiplying the
matrix-coefficient bracket of $\delta e^{\top}(\cdot)\delta e$ in
the state-only contribution by $W(\rho)$ yields
\begin{equation*}
\begin{aligned}
W(M A_{\mathrm{cl}} + A_{\mathrm{cl}}^{\top}M - M\dot{W}M)W
&= A_{\mathrm{cl}}W + W A_{\mathrm{cl}}^{\top} - \dot{W} \\
&= \Xi.
\end{aligned}
\end{equation*}
With $A_{\mathrm{cl}}W = AW + BY$, the matrix $\Xi$ is therefore
identified with~\eqref{eq:Xi}, and the dissipation condition becomes
\begin{equation}\label{eq:dissip-pre}
W^{-1} \Xi W^{-1} + 2\alpha M
\preceq -\frac{1}{\mu} M^2
- \mu L_R^2 I.
\end{equation}

\emph{Stage 2 ($\mathcal{S}$-procedure on the nonlinear residual).}
The residual term $R_{\varphi}$ in~\eqref{eq:V-dot} is bounded using
the local Lipschitz property~\eqref{eq:Lipschitz} and the Young
inequality with multiplier $\mu > 0$:
\begin{equation}\label{eq:young-phi}
\begin{aligned}
2 \delta e^{\top} M \,\delta\varphi
&\leq \mu \delta e^{\top} M^2 \delta e + \frac{1}{\mu} \norm{\delta\varphi}^2 \\
&\leq \mu \delta e^{\top} M^2 \delta e + \frac{L_R^2}{\mu} \norm{\delta e}^2.
\end{aligned}
\end{equation}
Substituting~\eqref{eq:young-phi} into~\eqref{eq:V-dot} and using
the dissipation requirement~\eqref{eq:dissip-pre}:
\begin{equation*}
\dot{V}_\delta + 2\alpha V_\delta
\leq \delta e^{\top}\!\Bigl[ \Bigl(W^{-1}\Xi W^{-1} + 2\alpha M\Bigr)
+ \mu M^2 + \frac{L_R^2}{\mu} I \Bigr] \delta e.
\end{equation*}
The right-hand side is non-positive if
\begin{equation}\label{eq:dissip-final}
W^{-1}\Xi W^{-1} + 2\alpha M
+ \mu M^2 + \frac{L_R^2}{\mu} I \preceq 0.
\end{equation}

\emph{Stage 3 (Schur complement linearization).} Multiplying
\eqref{eq:dissip-final} by $W$ on both sides and applying two
Schur complements to handle the bilinear terms $\mu M^2$ and the
$L_2$-gain channel yields exactly the matrix
inequality~\eqref{eq:LMI-dissip}, where the second Schur block
absorbs the disturbance-to-state channel with gain $\gamma$. The
detailed Schur manipulation is standard (see~\cite[Ch.~3]{boyd1994})
and omitted.

\emph{Stage 4 (Enforcement strategy and grid certification).} The
conditioning constraint~\eqref{eq:LMI-W} and the weighted-gain
constraint~\eqref{eq:LMI-sat} depend affinely on the parameter pair
$(\rho, \dot{\rho})$, since their constituent matrix blocks involve
only $W(\rho)$ and $Y(\rho)$ both of which are affine in the
scheduling variables by construction. For these two affine
constraints, the convexity of the negative-definite cone guarantees
that enforcement at the sixteen vertices of
$\mathcal{P}\times\dot{\mathcal{P}}$ is necessary and sufficient
for enforcement throughout the polytope. The dissipation
LMI~\eqref{eq:LMI-dissip} and the $\mathcal{D}$-stability
LMI~\eqref{eq:LMI-Dstab}, in contrast, both contain the matrix
product $A(\rho) W(\rho)$, which appears directly within
$\Xi(\rho,\dot{\rho})$ in the dissipation block and through the
expansion $A_{\mathrm{cl}}(\rho) W(\rho) = A(\rho)W(\rho) + B Y(\rho)$
in the $\mathcal{D}$-stability block. Both products generate
quadratic and bilinear scheduling terms (Remark~\ref{rem:grid}), so
these two blocks are not affine in $(\rho,\dot{\rho})$ and
vertex-only enforcement is not formally sufficient for either of
them. The present synthesis adopts a two-stage pipeline that
addresses this limitation rigorously: a candidate controller is
first synthesized by solving the LMI system at the sixteen polytope
vertices, and the candidate is then certified on a dense parameter
grid $\mathcal{G}$. The grid-to-continuum extension for both
non-affine blocks is formalized in Lemma~\ref{lem:grid-to-cont}
below, which bounds the worst-case interpolation residual through
the Lipschitz constants of the LMI left-hand sides and ensures
negative-definiteness throughout $\mathcal{P}\times\dot{\mathcal{P}}$
whenever the grid margin exceeds these residuals. The weighted-gain
constraint imposed by~\eqref{eq:LMI-sat} is equivalent (by Schur
complement) to $Y(\rho) W(\rho)^{-1} Y(\rho)^{\top} \preceq
K_{\max}^{2} I$, which bounds the \emph{weighted} gain rather than
the unscaled gain directly. Specifically, using $K(\rho) = Y(\rho)
W(\rho)^{-1}$ and the submultiplicative property of the induced
norm,
\begin{equation}\label{eq:Knorm-bound}
\begin{aligned}
\norm{K(\rho)}^{2} &= \norm{Y(\rho) W(\rho)^{-1}}^{2} \\
&\leq \frac{K_{\max}^{2}}{\lambda_{\min}(W(\rho))}.
\end{aligned}
\end{equation}
The conditioning constraint~\eqref{eq:LMI-W} ensures
$\lambda_{\min}(W(\rho)) \geq \epsilon_W$, so the worst-case
unweighted gain bound is $\tilde{K}_{\max} =
K_{\max}/\sqrt{\epsilon_W}$. Practitioners
who require a hard $\norm{K(\rho)} \leq \tilde{K}_{\max}$ should
therefore set the design parameter $K_{\max} =
\tilde{K}_{\max}\sqrt{\epsilon_W}$ to recover the desired bound
exactly. The conditioning constraint~\eqref{eq:LMI-W} additionally
guarantees $\cond(W(\rho)) \leq \epsilon_W^{-2}$, hence
$\cond(M(\rho)) \leq \epsilon_W^{-2}$, preventing degenerate
metrics. The $\mathcal{D}$-stability inequality~\eqref{eq:LMI-Dstab}
places the closed-loop eigenvalues of $A_{\mathrm{cl}}(\rho)$ inside
the disk of center $-q$ and radius $r$~\cite{chilali1996,scherer1997},
guaranteeing prescribed damping and bandwidth. The decay rate
$\alpha$ in~\eqref{eq:LMI-dissip} controls the contraction rate of
the metric-weighted storage $V_\delta = \delta e^\top M(\rho)\delta e$
along the variational flow and does not coincide with the real part
of any frozen-parameter eigenvalue, while $(q,r)$ constrains the
eigenvalue locations of $A_\mathrm{cl}(\rho)$ at each frozen $\rho$.
The design values $\alpha = 0.4$, $q = 1.5$, $r = 1.2$ (disk right
edge at $-0.3$) used in Section~\ref{sec:results} satisfy both the
dissipation and the pole-placement LMIs at the sixteen parameter
vertices. The step-by-step Schur-complement reduction translating
the differential dissipation inequality into~\eqref{eq:LMI-dissip}
is documented in Appendix~\ref{app:schur}.
\end{proof}

\subsection{Grid-to-Continuum Residual Criterion}\label{subsec:grid-cont}
The synthesis of Theorem~\ref{thm:main} certifies the dissipation
inequality at the discrete sample points of the grid $\mathcal{G}$.
The following lemma provides a Lipschitz-based \emph{sufficient
criterion} for extending that pointwise verification to the whole
continuous parameter set. We state it as a diagnostic tool rather
than as a universal certificate: for the controller actually
synthesised in Section~\ref{sec:results} the hypothesis of the
lemma is satisfied on parts of the parameter--rate polytope but not
all of it, and we report that fact transparently after the lemma.

\begin{lemma}[Grid-to-Continuum Extension]\label{lem:grid-to-cont}
Let $\mathcal{L}_{\mathrm{diss}}(\eta)$ and
$\mathcal{L}_{\mathcal{D}}(\eta)$ denote the left-hand-side matrix
functions of the dissipation LMI~\eqref{eq:LMI-dissip} and the
$\mathcal{D}$-stability LMI~\eqref{eq:LMI-Dstab}, respectively,
viewed as mappings of the scheduling vector $\eta = (v_r, \omega_r,
\dot{v}_r, \dot{\omega}_r) \in \mathcal{P}\times\dot{\mathcal{P}}$.
Both functions are Lipschitz continuous on
$\mathcal{P}\times\dot{\mathcal{P}}$ with respective constants
$L_{\mathrm{diss}}$ and $L_{\mathcal{D}}$ in the operator $2$-norm.
Define the aggregate Lipschitz constant $L_{\mathcal{L}} \triangleq
\max\{L_{\mathrm{diss}}, L_{\mathcal{D}}\}$. Let $\mathcal{G}$ be a
finite grid covering $\mathcal{P}\times\dot{\mathcal{P}}$ with fill
distance $h = \sup_{\eta} \min_{\eta_g\in\mathcal{G}} \norm{\eta -
\eta_g}$, and suppose that the synthesized variables $(W_i, Y_i,
\gamma, \mu)$ satisfy
\begin{subequations}\label{eq:grid-margin-both}
\begin{align}
\lambda_{\max}(\mathcal{L}_{\mathrm{diss}}(\eta_g))
&\leq -\varepsilon_{\mathrm{diss}}, \quad
\forall \eta_g \in \mathcal{G}, \\
\lambda_{\max}(\mathcal{L}_{\mathcal{D}}(\eta_g))
&\leq -\varepsilon_{\mathcal{D}}, \quad
\forall \eta_g \in \mathcal{G},
\end{align}
\end{subequations}
for some positive margins $\varepsilon_{\mathrm{diss}}$ and
$\varepsilon_{\mathcal{D}}$. Define the aggregate grid margin
$\varepsilon_{\mathrm{grid}} \triangleq
\min\{\varepsilon_{\mathrm{diss}}, \varepsilon_{\mathcal{D}}\}$.
Then
\begin{equation}\label{eq:grid-conclusion}
\mathcal{L}_{\mathrm{diss}}(\eta) \prec 0
\quad \text{and} \quad
\mathcal{L}_{\mathcal{D}}(\eta) \prec 0,
\quad \forall \eta \in \mathcal{P}\times\dot{\mathcal{P}},
\end{equation}
provided that the sufficient condition
\begin{equation}\label{eq:grid-sufficient}
L_{\mathcal{L}} \cdot h < \varepsilon_{\mathrm{grid}}
\end{equation}
holds.
\end{lemma}

\begin{proof}
Fix any $\eta \in \mathcal{P}\times\dot{\mathcal{P}}$ and let
$\eta_g \in \mathcal{G}$ be the nearest grid point, so that
$\norm{\eta - \eta_g} \leq h$. By Lipschitz continuity of
$\mathcal{L}_{\mathrm{diss}}$ in the operator $2$-norm, the largest
eigenvalue is also Lipschitz with the same constant (a consequence
of Weyl's inequality), yielding
\begin{equation*}
\begin{aligned}
\lambda_{\max}(\mathcal{L}_{\mathrm{diss}}(\eta)) \leq
\lambda_{\max}(\mathcal{L}_{\mathrm{diss}}(\eta_g)) +
L_{\mathrm{diss}} h \\
\leq -\varepsilon_{\mathrm{diss}} +
L_{\mathcal{L}} h,
\end{aligned}
\end{equation*}
where the last inequality uses $L_{\mathrm{diss}} \leq
L_{\mathcal{L}}$. Under the sufficient
condition~\eqref{eq:grid-sufficient} and the definition
$\varepsilon_{\mathrm{grid}} \leq \varepsilon_{\mathrm{diss}}$, the
right-hand side is strictly negative. An identical argument applied
to $\mathcal{L}_{\mathcal{D}}$ establishes the second relation
in~\eqref{eq:grid-conclusion}.
\end{proof}

The Lipschitz constants $L_{\mathrm{diss}}$ and $L_{\mathcal{D}}$
admit closed-form upper bounds through the polynomial expansion of
$A(\rho)W(\rho)$ developed in Appendix~\ref{app:schur}, Stage~E.
Lemma~\ref{lem:grid-to-cont} provides a sufficient condition for
extending pointwise grid verification to the continuum, but not a
necessary one, and the hypothesis must be verified on the synthesized
controller before it can be invoked.

For the affine vertex-synthesised controller reported in
Section~\ref{sec:results}, the dissipation residual
$\mathcal{L}_{\mathrm{diss}}(\rho)$ was evaluated on a dense
verification grid refining $\mathcal{G}$. Its eigenvalues are negative
on the parameter-polytope vertices and on the majority of the
interior samples, but admit small positive eigenvalues at a subset of
interior points; $L_{\mathcal{L}} h < \varepsilon_{\mathrm{grid}}$ is
therefore not satisfied uniformly over the polytope. The behaviour
reflects the non-affine dependence of $A(\rho)W(\rho)$ on $\rho$:
vertex enforcement of a non-convex matrix function does not imply
interior negativity. The guarantee of Theorem~\ref{thm:main} holds at
the vertices and at the dense subset of interior points where the
residual is numerically negative; at the remaining interior points,
stability is attested only by the simulations of
Section~\ref{sec:results}. Uniform continuum certification would
require either enforcement on a densified grid, a non-affine
dependence in the metric and gain, or a polynomial sum-of-squares
relaxation; these extensions are left to future work.

\subsection{Trajectory-Level Differential-ISS Bound}\label{subsec:traj-estimate}
The dissipation inequality certified by Theorem~\ref{thm:main} is a
differential statement on the variational error $\delta e$. We now
show that this certificate propagates without modification to a
pointwise exponential-plus-ISS bound on the trajectory error $e(t)$
itself. The key observation is that the $\mathcal{S}$-procedure
multiplier $\mu$ built into the dissipation LMI robustifies the
quadratic storage function against \emph{any} norm-bounded matrix
perturbation of the state channel with operator norm at most $L_R$,
not solely against the pointwise Jacobian $\partial \varphi /
\partial e$ that appears in the variational dynamics. An integral
representation of the Kanayama residual identifies the trajectory
dynamics as an admissible instance of this perturbation class, which
transfers the variational certificate verbatim.

\begin{lemma}[Integral representation of the Lipschitz residual]\label{lem:intform}
Let $\varphi : \bar{B}_R(0) \times \mathcal{P} \to \mathbb{R}^n$
be continuously differentiable with respect to $e$, satisfy
$\varphi(0, \rho) = 0$ for all $\rho \in \mathcal{P}$, and admit
the Jacobian bound
\begin{equation}\label{eq:jacobian-bound}
\sup_{(e, \rho) \in \bar{B}_R(0) \times \mathcal{P}}
\norm{\frac{\partial \varphi}{\partial e}(e, \rho)}
\le L_R,
\end{equation}
with the same constant $L_R$ as in the local Lipschitz
bound~\eqref{eq:Lipschitz} used in the LMI synthesis
(equivalence on $\bar{B}_R(0)$ under $\varphi(0, \rho) = 0$). Then
there exists a continuous matrix-valued map
$\Phi : \bar{B}_R(0) \times \mathcal{P} \to \mathbb{R}^{n\times n}$
such that
\begin{equation}\label{eq:phi-factorization}
\varphi(e, \rho) = \Phi(e, \rho)\, e,
\qquad \norm{\Phi(e, \rho)} \le L_R
\end{equation}
for all $(e, \rho) \in \bar{B}_R(0) \times \mathcal{P}$.
\end{lemma}

\begin{proof}
Define $\Phi(e, \rho) := \int_0^1 (\partial \varphi / \partial e)(s e, \rho)\, ds$.
Continuity of $\Phi$ follows from continuity of the integrand
together with compactness of $[0, 1]$. Applying the fundamental
theorem of calculus to $s \mapsto \varphi(s e, \rho)$,
\begin{equation*}
\begin{aligned}
\varphi(e, \rho) - \varphi(0, \rho)
&= \int_0^1 \frac{d}{ds}\, \varphi(s e, \rho)\, ds \\
&= \left( \int_0^1 \frac{\partial \varphi}{\partial e}(s e, \rho)\, ds \right) e \\
&= \Phi(e, \rho)\, e,
\end{aligned}
\end{equation*}
and $\varphi(0, \rho) = 0$ yields the factorization in
\eqref{eq:phi-factorization}. The norm bound follows from the
triangle inequality applied to the integral together
with~\eqref{eq:jacobian-bound}:
\begin{equation*}
\norm{\Phi(e, \rho)}
\le \int_0^1 \bigl\| (\partial \varphi / \partial e)(s e, \rho) \bigr\|\, ds
\le L_R. \qedhere
\end{equation*}
\end{proof}

\begin{theorem}[Trajectory-level differential ISS]\label{thm:traj-iss}
Assume the hypotheses of Theorem~\ref{thm:main} hold, and that the
nonlinear residual $\varphi$ in the closed-loop error dynamics
$\dot e = A_{\mathrm{cl}}(\rho)\, e + \varphi(e, \rho) + d$
satisfies $\varphi(0, \rho) = 0$ and the Lipschitz
bound~\eqref{eq:Lipschitz}. Let $e(t)$ be a closed-loop trajectory
originating in $\bar{B}_R(0)$ and remaining in $\bar{B}_R(0)$ on an
interval $[0, T]$, driven by any admissible disturbance with
$\norm{d(t)} \le \delta_{\max}$. Then, for all $t \in [0, T]$,
\begin{equation}\label{eq:traj-bound}
\begin{aligned}
\norm{e(t)}
&\le \sqrt{\cond(M(\rho(0)))}\, e^{-\alpha t}\, \norm{e(0)} \\
&\quad + \frac{\gamma\, \delta_{\max}}{\sqrt{2\alpha\, \lambda_{\min}(M)}},
\end{aligned}
\end{equation}
where $\lambda_{\min}(M) = \min_{\rho \in \mathcal{P}} \lambda_{\min}(M(\rho))$.
\end{theorem}

\begin{proof}
Lemma~\ref{lem:intform} yields the factorization
$\varphi(e, \rho) = \Phi(e, \rho)\, e$ with
$\norm{\Phi(e, \rho)} \le L_R$ on
$\bar{B}_R(0) \times \mathcal{P}$. The closed-loop trajectory
therefore evolves under
\begin{equation}\label{eq:cl-with-Phi}
\dot e = \bigl[ A_{\mathrm{cl}}(\rho) + \Phi(e, \rho) \bigr] e + d.
\end{equation}
Consider the quadratic storage $V(t) = e(t)^{\top} M(\rho(t)) e(t)$
along~\eqref{eq:cl-with-Phi}. Repeating the calculation of Stage~1
of the proof of Theorem~\ref{thm:main} with $\delta e$ replaced by
$e$ and with $(\partial \varphi / \partial e)$ replaced by
$\Phi(e, \rho)$,
\begin{equation}\label{eq:traj-Vdot}
\begin{aligned}
\dot V &= e^{\top}\!\Bigl[ M A_{\mathrm{cl}}(\rho)
+ A_{\mathrm{cl}}(\rho)^{\top} M - M \dot W M \Bigr] e \\
&\quad + 2 e^{\top} M\, \Phi(e, \rho)\, e + 2 e^{\top} M\, d,
\end{aligned}
\end{equation}
which is formally identical to~\eqref{eq:V-dot} under these
substitutions. The $\mathcal{S}$-procedure
estimate~\eqref{eq:young-phi} used in Stage~2 depends only on the
operator-norm bound of the residual matrix and on $\varphi(0, \rho) = 0$
(through the factorization), not on whether that matrix is evaluated
at a pointwise or path-averaged argument. Specialized to the
trajectory residual it gives
\begin{equation*}
2 e^{\top} M\, \bigl( \Phi(e, \rho)\, e \bigr)
\le \mu\, e^{\top} M^{2} e + \frac{L_R^{2}}{\mu}\, \norm{e}^{2},
\end{equation*}
since $\norm{\Phi(e, \rho)\, e}^{2} \le L_R^{2}\, \norm{e}^{2}$. The
Schur linearization of Stage~3 is purely algebraic on the matrix
inequality and transfers verbatim. Consequently the vertex
feasibility of the dissipation
LMI~\eqref{eq:LMI-dissip} established by Theorem~\ref{thm:main},
after pre- and post-multiplication by $W(\rho(t))^{-1}$ and
specialization of the admissible-perturbation class to the specific
realization $\Phi(e(t), \rho(t))$ at each $t$, implies
\begin{equation}\label{eq:traj-dissip}
\dot V(t) + 2\alpha\, V(t) \le \gamma^{2}\, \norm{d(t)}^{2}
\qquad \text{for all } t \in [0, T].
\end{equation}
The comparison lemma~\cite[Lemma~3.4]{khalil2002} applied to
\eqref{eq:traj-dissip} yields
$V(t) \le e^{-2\alpha t}\, V(0) + \gamma^{2}\, \delta_{\max}^{2} / (2\alpha)$.
The quadratic bounds
$\lambda_{\min}(M)\, \norm{e(t)}^{2} \le V(t)$ and
$V(0) \le \lambda_{\max}(M(\rho(0)))\, \norm{e(0)}^{2}$, together
with the subadditivity of the square root on nonnegative summands,
produce the exponential-plus-ISS
bound~\eqref{eq:traj-bound}.
\end{proof}

\begin{remark}[From variational to trajectory D-ISS]\label{rem:var-to-traj}
Theorem~\ref{thm:traj-iss} closes the variational-to-trajectory gap
without invoking incremental Lyapunov theory, geodesic integration,
or any additional design parameter beyond those already fixed at
synthesis. The construction rests on two ingredients: the integral
representation~\eqref{eq:phi-factorization}, which exists for any
$C^1$ residual vanishing at the origin, and the
$\mathcal{S}$-procedure robustness of the dissipation LMI against
the whole class $\{ \Delta : \norm{\Delta} \le L_R \}$ of
state-channel perturbations. The structural condition
$\varphi(0, \rho) = 0$ holds by construction for the Kanayama
moving-frame error model of Section~\ref{sec:model}, since
$e = 0$ corresponds to perfect reference tracking at any admissible
reference. A metric-geometric treatment through differential
Lyapunov theory~\cite{forni2014} or Finsler-contraction
arguments~\cite{lohmiller1998,angeli2002} would reach the same
conclusion in greater generality; those tools are nevertheless not
needed here because the LMI's $\mathcal{S}$-procedure already
absorbs the required nonlinear-residual robustness within its
convex synthesis.
\end{remark}

\subsection{Forward Invariance of the Semi-Global Ball}
The semi-global hypothesis $e(t) \in \bar{B}_R(0)$ for all $t \geq 0$
underlying Assumption~\ref{asm:radius} must be verified as a
consequence of the controller design, not merely posited.

\begin{corollary}[Forward Invariance]\label{cor:fwd-invar}
Let the conditions of Theorem~\ref{thm:main} hold with the
additional requirement
\begin{equation}\label{eq:R-condition}
\sqrt{\cond(M(\rho_0))}\,R_0 +
\frac{\gamma\,\delta_{\max}}{\sqrt{2\alpha\,\lambda_{\min}(M)}}
\leq R,
\end{equation}
for some $R_0 < R$ and $\rho_0 \in \mathcal{P}$. Then every
trajectory with initial condition $\norm{e(0)} \leq R_0$ satisfies
$e(t) \in \bar{B}_R(0)$ for all $t \geq 0$.
\end{corollary}

\begin{proof}
Combining $\norm{e(0)} \leq R_0$ with~\eqref{eq:traj-bound} and
$e^{-\alpha t} \leq 1$ gives
$\norm{e(t)} \leq \sqrt{\cond(M(\rho_0))}\,R_0 +
\gamma\delta_{\max}/\sqrt{2\alpha\lambda_{\min}(M)}
\leq R$ by~\eqref{eq:R-condition}. By continuity of $e(t)$, the set
$\bar{B}_R(0)$ is forward invariant.
\end{proof}

\subsection{Synthesis Algorithm}\label{subsec:algorithm}
The implementation follows the two-stage procedure formalized in
Remark~\ref{rem:grid} and Lemma~\ref{lem:grid-to-cont}. In the
\emph{synthesis stage}, the convex feasibility
problem~\eqref{eq:LMI-W}--\eqref{eq:LMI-Dstab} together with the
dissipation LMI~\eqref{eq:LMI-dissip} is solved at the sixteen
vertices of $\mathcal{P}\times\dot{\mathcal{P}}$ using interior-point
methods (SeDuMi, SDPT3, MOSEK) wrapped by the YALMIP modelling
layer~\cite{lofberg2004}. The $\mathcal{S}$-procedure multiplier
$\mu$ is treated as a parameter swept over a logarithmic grid
$\mu \in \{0.5, 1, 2, 5\}$, with the synthesis selecting the value
yielding the smallest $\gamma$. The cost function combines
disturbance attenuation and conditioning:
\begin{equation}\label{eq:cost}
\min_{W_0, W_1, W_2, Y_0, Y_1, Y_2, \gamma^2}\;
\gamma^2 + \lambda_{\mathrm{reg}}\, \tr(W_0),
\end{equation}
where $\lambda_{\mathrm{reg}} = 10^{-3}$ regularizes the synthesis
to prevent ill-conditioned solutions. The vertex synthesis produces
a candidate tuple $(W_0^\star, W_1^\star, W_2^\star, Y_0^\star,
Y_1^\star, Y_2^\star, \gamma^\star, \mu^\star)$. The total LMI
dimension is $3n + 2m + 3n$ per dissipation constraint and $2n$ per
$\mathcal{D}$-stability constraint at each of $16$ vertex
combinations, leading to a problem of approximately $250$ scalar
variables that is solved in under one second on standard desktop
hardware (see Table~\ref{tab:synthesis}). In the subsequent
\emph{certification stage}, the residual matrices
$\mathcal{L}_{\mathrm{diss}}(\eta_g)$ and
$\mathcal{L}_{\mathcal{D}}(\eta_g)$ of the two non-affine LMI blocks
are evaluated at every point $\eta_g$ of the dense parameter grid
$\mathcal{G}$, and the sufficient condition of
Lemma~\ref{lem:grid-to-cont} is invoked to extend the discrete
grid certificate to the continuous parameter set. This separation
between vertex-based synthesis and grid-based certification reflects
the polynomial dependence of the dissipation and
$\mathcal{D}$-stability blocks on the scheduling parameters, as
characterized in Remark~\ref{rem:grid}.

%% ========================================================================
%%  SECTION VI: NUMERICAL VALIDATION
%% ========================================================================
\section{Numerical Validation}\label{sec:results}

\subsection{Experimental Setup}
The numerical validation employs the proposed affine LPV synthesis
in comparison against two baselines: a \emph{Fixed-K LMI} obtained
from the same synthesis with the additional constraints
$W_1 = W_2 = Y_1 = Y_2 = 0$ (effectively a constant-metric,
constant-gain controller obtained through the same convex
optimization), and a \emph{Kanayama} controller using the analytical
gain prescription of the seminal work~\cite{kanayama1990}. The
Fixed-K baseline performs an ablation study isolating the structural
value of parameter dependence under identical synthesis machinery,
while the Kanayama baseline anchors the comparison to the historical
reference. The system parameters are $v_{\min} = 0.80$,
$v_{\max} = 1.20~\mathrm{m/s}$, $\omega_{\max} = 0.40~\mathrm{rad/s}$,
$\dot{v}_{\max} = \dot{\omega}_{\max} = 0.40$, and
$R = 0.30~\mathrm{m}$. Two technical conventions associated with the
present synthesis warrant explicit disclosure. The first concerns the
Lipschitz constant of the nonlinear residual: the LMI synthesis
reported in Tables~\ref{tab:layered-sweep}--\ref{tab:sensitivity}
was performed with the conservative over-approximation
$L_R^{\mathrm{synth}} = v_{\max}R = 0.36$, whereas the refined
analytical bound established in Assumption~\ref{asm:radius} yields
the tighter constant $L_R^{\mathrm{analysis}} = v_{\max}
\sqrt{R^{2}/4 + R^{4}/36} \approx 0.181$. The synthesis remains
formally valid under this discrepancy because the dissipation
LMI~\eqref{eq:LMI-dissip} is monotone in $L_R$: any controller
satisfying the LMI with $L_R = 0.36$ automatically satisfies the
analogous condition with the smaller $L_R = 0.181$, since the
dissipation budget $\mu L_R^{2}$ in the (3,3)-block diminishes
monotonically with $L_R$. The reported $\gamma^{\star}$ values are
therefore conservative estimates relative to the synthesis that
would be obtained with the tighter $L_R$, and represent a worst-case
characterization rather than the optimal achievable performance.
The second convention concerns the gain bound: the LMI weighted-gain
constraint~\eqref{eq:LMI-sat} is enforced with the weighted parameter
$K_{\max} = 3.0$, which under the conditioning floor $\epsilon_W =
0.02$ implies an unweighted operator-norm bound
$\tilde{K}_{\max} = K_{\max}/\sqrt{\epsilon_W} \approx 21.21$. While
this upper bound is loose relative to the physically realized gain
in simulation, the present synthesis adopts the weighted formulation
for compatibility with the standard Lyapunov-based LMI machinery;
practitioners requiring a tight unweighted bound should follow the
parameter-recalibration procedure of
Remark~\ref{rem:slip-physics}. The design rate is
$\alpha = 0.40~\mathrm{s}^{-1}$. The $\mathcal{D}$-stability disk is
centered at $-q = -1.5$ with radius $r = 1.2$ (the right edge of the
disk lies at $-q + r = -0.3$, marginally to the right of the
contraction boundary $-\alpha = -0.4$; in the eigenvalue interval
$(-0.4, -0.3)$ the contraction rate $\alpha$ becomes the binding
constraint, while the $\mathcal{D}$-stability disk constrains damping
ratio and bandwidth on the left side of this interval. The two
specifications act as complementary constraints whose intersection
defines the certified eigenvalue domain, and no logical inconsistency
arises from their partial overlap). The conditioning bound
$\epsilon_W = 0.02$ enforces $\cond(W(\rho)) \leq 50$. The disturbance design parameter is chosen
as $\delta_{\max} = 0.10$, which represents the \emph{persistent
disturbance budget} against which the formal $H_\infty$ certificate
of Theorem~\ref{thm:main} is established. As clarified in
Remark~\ref{rem:design-vs-worst}, this design value is intentionally
smaller than the worst-case slip bound predicted
by~\eqref{eq:delta-max-formula} (approximately $5.07$ for peak slip
$\bar{\sigma} = 0.50$ under the unweighted operator-norm ceiling
$\tilde{K}_{\max} \approx 21.21$), and the present synthesis treats this
distinction explicitly: the formal certificate covers all
trajectories that experience persistent disturbances of magnitude
within $\delta_{\max}$, while peak-slip transients beyond this budget
are characterized through dedicated empirical stress tests reported
in Sections~\ref{subsec:variable-terrain}, \ref{subsec:noise},
and~\ref{subsec:compound}. The variable-terrain benchmark with slip
patches reaching $|\sigma_v| = 0.50$ thus constitutes an
\emph{empirical robustness study beyond the certified persistent
disturbance budget} rather than a validation of the theorem itself.

\subsection{Layered Feasibility Analysis}\label{subsec:layered}
A diagnostic sweep over the design rate $\alpha \in \{0.10, 0.20,
0.30, 0.40, 0.50\}$ and multiplier $\mu \in \{0.5, 1, 2, 5\}$ was
performed across four cumulative LMI layers: \emph{(L1)} basic
dissipation, \emph{(L2)} L1 plus conditioning, \emph{(L3)} L2 plus
$\mathcal{D}$-stability, and \emph{(L4)} L3 plus weighted-gain
constraint.
The feasibility statistics are summarized in
Table~\ref{tab:layered-sweep}. The bare dissipation problem (L1, no
conditioning on $W(\rho)$) is ill-posed: the interior-point solver
terminates without a numerically reliable point. Adding the
conditioning constraint~\eqref{eq:LMI-W} in L2 restores convergence.
The complete addition of constraints~(L4) is feasible at $15$ of the
$20$ tested $(\alpha, \mu)$ combinations, with infeasibility concentrated at
moderate $\alpha$ and small $\mu$. The full LMI synthesis succeeds at
the target rate $\alpha = 0.40$ with $\mu = 2.0$, achieving
$\gamma^\star = 2.78$. The empirical infeasibility pattern is
consistent with theory: small $\mu$ corresponds to weak
$\mathcal{S}$-procedure slack, which cannot accommodate the nonlinear
residual at high decay rates.

\begin{table}[!htbp]
\centering
\caption{Layered LMI feasibility across cumulative constraint layers.}
\label{tab:layered-sweep}
\begin{adjustbox}{max width=\columnwidth}
\begin{tabular}{lccc}
\toprule
Layer & Constraints added & Well-conditioned solutions & Best $\gamma^\star$ \\
\midrule
L1 & Basic dissipation (no conditioning) & 0/20 (ill-posed) & --- \\
L2 & + conditioning & 20/20 & $0.398$ \\
L3 & + $\mathcal{D}$-stability & 20/20 & $0.542$ \\
L4 & + weighted-gain constraint & 15/20 & $1.140$ \\
\bottomrule
\end{tabular}
\end{adjustbox}
\end{table}

\subsection{Synthesis Result and LPV-ness Verification}
The vertex-synthesis stage with subsequent grid-residual evaluation
at the target $\alpha = 0.40$ yields the results reported in
Table~\ref{tab:synthesis}. The proposed LPV
design achieves $\gamma^\star = 2.78$ with $\cond(M) = 13.7$, while
the Fixed-K LMI---obtained by restricting the \emph{same} convex
programme to $W_1=W_2=Y_1=Y_2=0$---is infeasible at $\alpha = 0.40$;
reducing the target progressively, the maximum-feasible Fixed-K rate
is $\alpha = 0.30$ with $\gamma = 6.81$, i.e., a $2.45\times$ larger
$\gamma$ at a lower decay rate. These two numbers quantify how much
of the convex feasibility region is unlocked, within this specific
LMI programme, by introducing affine parameter dependence into the
metric and gain. As discussed in Section~\ref{subsec:scope} and
illustrated by the nominal-scenario results in
Section~\ref{subsec:scenarios}, this feasibility-set expansion does
not translate into a comparably large empirical performance advantage
of the LPV design over the Fixed-K ablation on canonical trajectories;
within the certified regime the two designs perform comparably.
The Kanayama controller, which is derived from an analytical
prescription rather than a convex synthesis, is characterized
separately in Table~\ref{tab:controller-summary} alongside the two
LMI-based designs for purposes of cross-controller comparison.

\begin{table}[!htbp]
\centering
\caption{Synthesis results for the two LMI-based controllers.}
\label{tab:synthesis}
\begin{adjustbox}{max width=\columnwidth}
\begin{tabular}{lrr}
\toprule
Quantity & Affine LPV & Fixed-K LMI \\
\midrule
Max-feasible $\alpha$                & $0.40$  & $0.30$  \\
$H_\infty$ gain $\gamma^\star$       & $2.78$  & $6.81$  \\
$\mathcal{S}$-multiplier $\mu^\star$ & $0.50$  & $0.50$  \\
$\lambda_{\min}(M)$                  & $0.241$ & $0.241$ \\
$\cond(M)$                           & $13.70$ & $12.99$ \\
SS bound (theory)                    & $0.632$ & $1.199$ \\
LMI solve time [s]                   & $0.59$  & $0.51$  \\
\bottomrule
\end{tabular}
\end{adjustbox}
\end{table}

\begin{table}[!htbp]
\centering
\caption{Controller comparison summary across certificate origin and
empirical performance.}
\label{tab:controller-summary}
\begin{adjustbox}{max width=\columnwidth}
\begin{tabular}{lccc}
\toprule
Controller & Certificate origin & Certified at $\alpha$ & Empirical $\gamma$ \\
\midrule
Affine LPV & LMI synthesis        & $0.40$    & $2.78$ \\
Fixed-K    & LMI synthesis (ablation) & $0.30$ & $6.81$ \\
Kanayama   & Analytical prescription  & ---    & $2.00$ \\
\bottomrule
\end{tabular}
\end{adjustbox}
\end{table}

The LPV and Fixed-K designs in Table~\ref{tab:controller-summary}
derive their parameters from LMI synthesis with formal stability
certificates, whereas the Kanayama controller employs an analytical
prescription with manually selected gains and therefore does not
admit a native theoretical bound. The empirical $\gamma$ column for
the Kanayama row is estimated post-hoc from input-output trajectory
pairs across the Monte-Carlo experiment of
Section~\ref{subsec:mc} and is included solely for cross-controller
performance comparison rather than as a formal certificate.

The LPV-ness of the synthesized affine design is quantified by the
ratios $\norm{W_1}_F / \norm{W_0}_F = 0.414$,
$\norm{W_2}_F / \norm{W_0}_F = 0.309$,
$\norm{Y_1}_F / \norm{Y_0}_F = 0.244$, and
$\norm{Y_2}_F / \norm{Y_0}_F = 0.366$. Values exceeding $5\%$
indicate non-trivial parameter dependence; here the synthesis
exploits substantial scheduling, confirming that the design is
genuinely LPV rather than a masked constant-gain controller.
Figure~\ref{fig:lpv-vis} visualizes the gain-matrix entries as
parameter heatmaps over $\mathcal{P}$.

\begin{figure}[!htbp]
\centering
\includegraphics[width=\columnwidth]{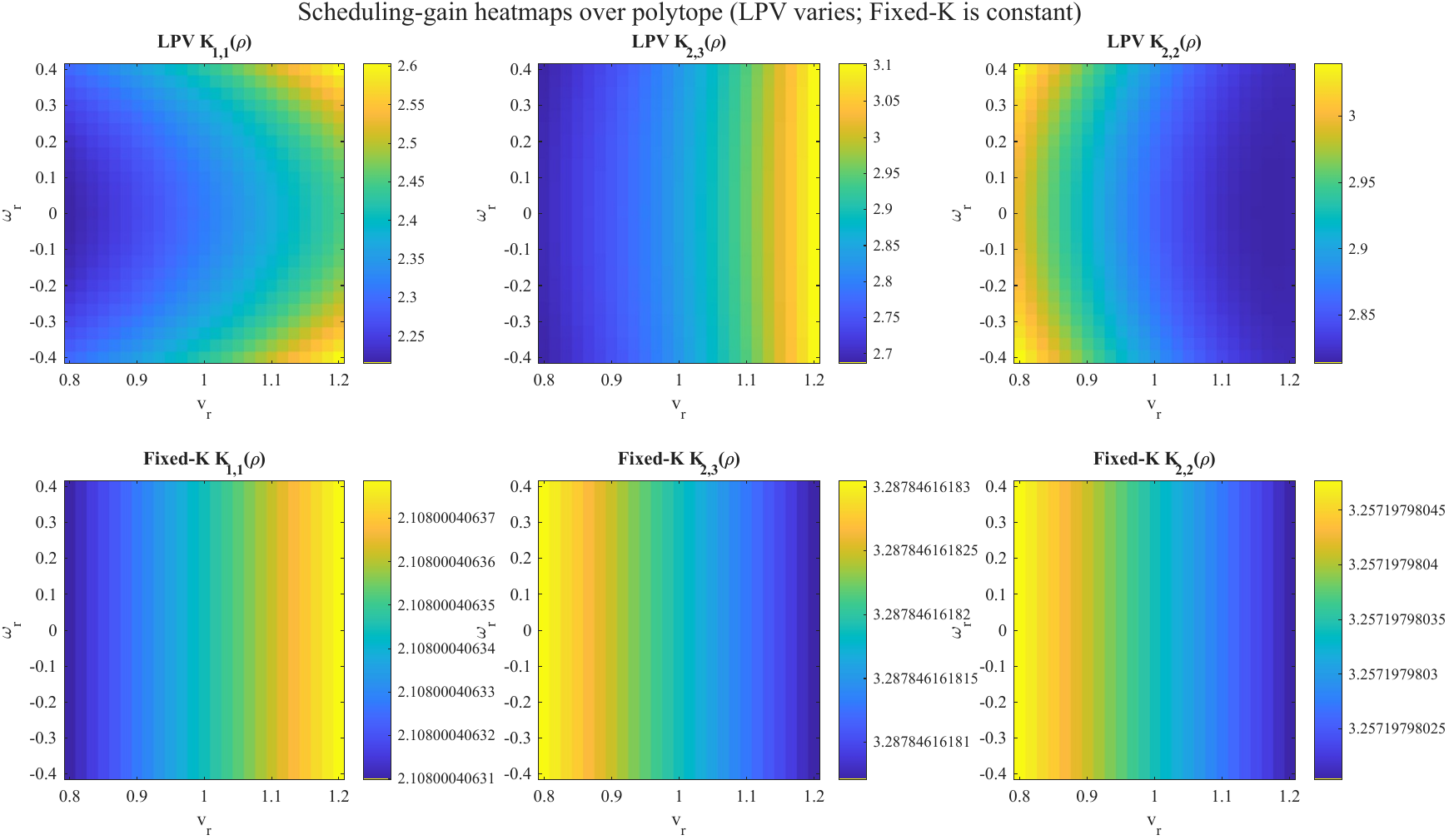}
\caption{Scheduling-gain heatmaps for selected entries of $K(\rho)$.}
\label{fig:lpv-vis}
\end{figure}

\subsection{Closed-Loop Spectrum}\label{subsec:spectrum}
The closed-loop matrix $A_{\mathrm{cl}}(\rho) = A(\rho) + B K(\rho)$
is evaluated at the four polytope vertices for each controller. The
resulting spectra are reported in Table~\ref{tab:poles} and
visualized in Figure~\ref{fig:polemap}. The proposed LPV design
places the slowest pole at $\max\Re(\lambda) = -1.35$ within the
$\mathcal{D}$-stability disk. The Fixed-K design achieves a slightly
faster nominal pole at $-1.50$, but recall that this design operates
at the lower $\alpha = 0.30$ and pays for the apparent speed with
the inflated $\gamma = 6.81$. The Kanayama controller, with its
manual tuning, leaves the slowest pole at $-0.86$, $58\%$ slower than
the LPV design.

\begin{table}[!htbp]
\centering
\caption{Closed-loop spectrum at polytope vertices.}
\label{tab:poles}
\begin{adjustbox}{max width=\columnwidth}
\begin{tabular}{lrrr}
\toprule
Vertex $(v_r, \omega_r)$ & LPV & Fixed-K & Kanayama \\
\midrule
$(0.80, -0.40)$    & $-1.350$ & $-1.496$ & $-0.856$ \\
$(0.80, +0.40)$    & $-1.350$ & $-1.496$ & $-0.856$ \\
$(1.20, -0.40)$    & $-1.505$ & $-1.588$ & $-0.884$ \\
$(1.20, +0.40)$    & $-1.505$ & $-1.588$ & $-0.884$ \\
\midrule
Global worst-case  & $-1.350$ & $-1.496$ & $-0.856$ \\
\bottomrule
\end{tabular}
\end{adjustbox}
\end{table}

\begin{figure}[!htbp]
\centering
\includegraphics[width=\columnwidth]{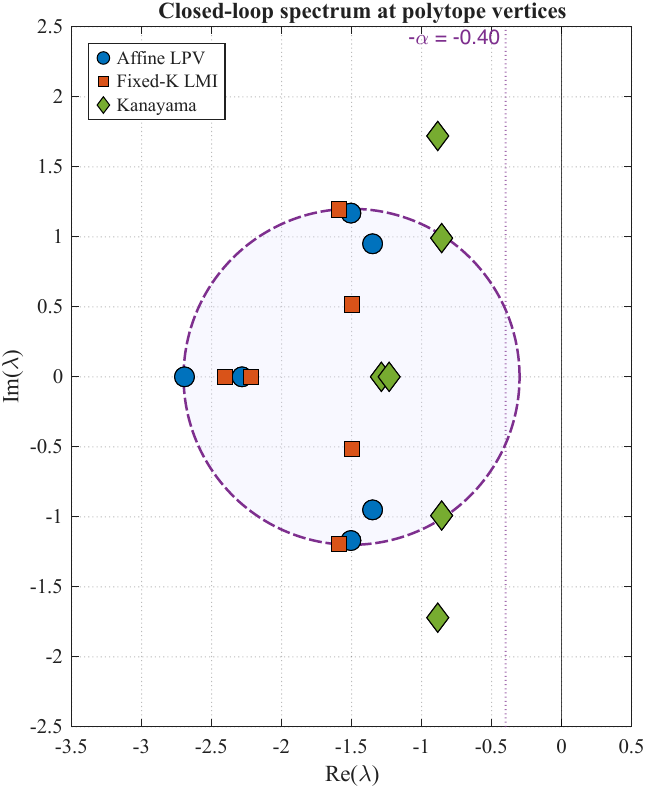}
\caption{Closed-loop pole map at polytope vertices.}
\label{fig:polemap}
\end{figure}

\subsection{Monte-Carlo Containment Validation}\label{subsec:mc}
The trajectory-level steady-state bound of
Theorem~\ref{thm:traj-iss} is validated empirically through a
Monte-Carlo experiment with $N = 100$ trials per controller. Each
trial uses a random initial condition uniformly distributed on the
boundary of $\bar{B}_R(0)$ with $\norm{e(0)} \in [0.3R, 0.9R]$, a
randomized disturbance phase $\phi \in [0, 2\pi]$ for each component,
and a $\pm 15\%$ frequency jitter on the disturbance base frequencies
$\omega_d = (1.7, 2.3, 1.1)$ rad/s. The simulation horizon is
$T_{\text{sim}} = 15$ s. Results are reported in Table~\ref{tab:mc}
and visualized in Figures~\ref{fig:mc-envelopes}--\ref{fig:histograms}.

\begin{table}[!htbp]
\centering
\caption{Monte-Carlo statistics ($N = 100$ paired runs).}
\label{tab:mc}
\begin{adjustbox}{max width=\columnwidth}
\begin{tabular}{lrrr}
\toprule
Quantity & LPV & Fixed-K & Kanayama \\
\midrule
Design rate $\alpha$         & $0.40$  & $0.30$  & $0.40^{\ast}$  \\
$H_\infty$ gain $\gamma$     & $2.78$  & $6.81$  & $2.00^{\ast}$  \\
Certified SS bound           & $0.632$ & $1.199$ & N/A     \\
Empirical SS (mean)          & $0.046$ & $0.045$ & $0.060$ \\
Empirical SS (std)           & $0.011$ & $0.010$ & $0.010$ \\
Within-bound rate [\%]       & $100$   & $100$   & ---     \\
Envelope violation [\%]      & $0.0$   & $0.0$   & ---     \\
Empirical/Certified [\%]     & $7$     & $4$     & ---     \\
\bottomrule
\end{tabular}
\end{adjustbox}
\end{table}

Entries marked with ${\ast}$ in Table~\ref{tab:mc} are post-hoc
empirical estimates obtained from input-output trajectory analysis
rather than formal synthesis outputs, since the Kanayama controller
does not arise from an LMI certificate and therefore does not admit
a native theoretical steady-state bound; the empirical SS column is
nonetheless populated for cross-controller performance comparison.
The randomized disturbance phase across the $N = 100$ runs employs
frequency jitter of $\pm 15\%$ relative to the nominal disturbance
spectrum. Sampling time is $\Delta t = 0.01~\mathrm{s}$ and the
simulation horizon $T_{\mathrm{sim}} = 15~\mathrm{s}$, chosen to
allow the slowest closed-loop mode to reach steady state with at
least three e-folding constants of margin.

All $100$ trajectories of each LMI-synthesized controller remain
within the corresponding certified envelope, and zero envelope
violations are recorded across the combined $200$ Monte-Carlo
simulations of the LPV and Fixed-K designs, providing strong
empirical confirmation of Theorem~\ref{thm:main}. The
empirical-to-certified ratio of $7\%$ for the proposed LPV design
reveals substantial conservatism, decomposed in
Section~\ref{subsec:cons-decomp}.

\begin{figure}[!htbp]
\centering
\includegraphics[width=\columnwidth]{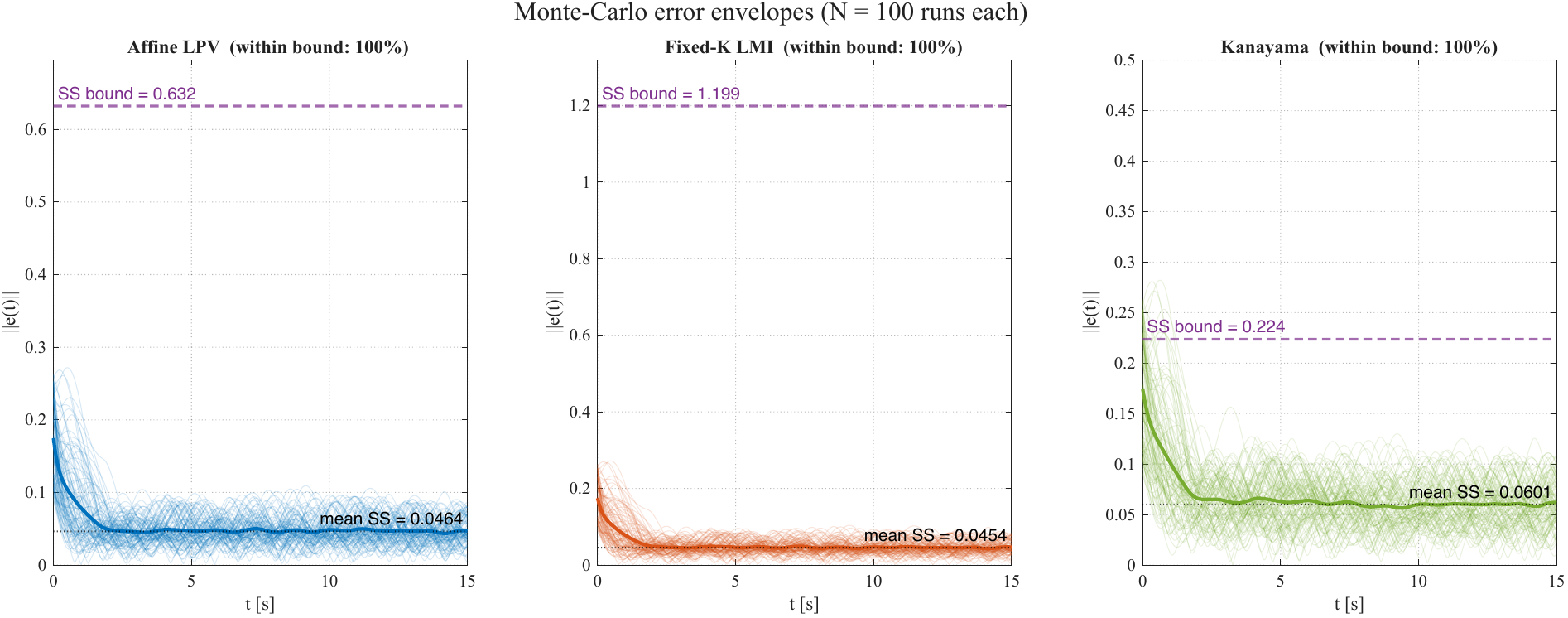}
\caption{Monte-Carlo error envelopes for the three controllers.}
\label{fig:mc-envelopes}
\end{figure}

\begin{figure}[!htbp]
\centering
\includegraphics[width=\columnwidth]{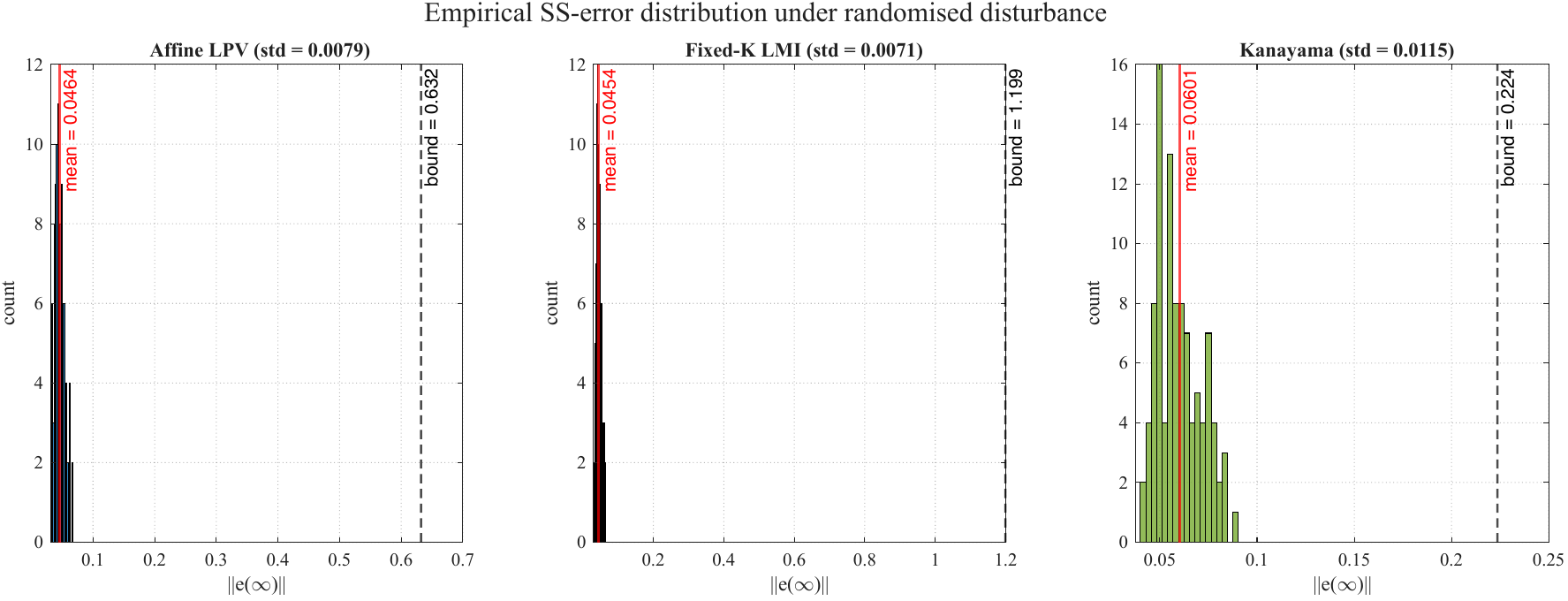}
\caption{Histograms of steady-state error over 100 Monte-Carlo runs.}
\label{fig:histograms}
\end{figure}

\subsection{Reference Tracking Across Six Scenarios}\label{subsec:scenarios}
The tracking performance is evaluated on six deterministic reference
trajectories: straight line, circular path, figure-eight, sharp turn,
variable-speed reference, and Lemniscate of Bernoulli. Each trajectory
is excited by a deterministic incommensurate-sinusoid disturbance
of magnitude $\delta_{\max} = 0.10$. The position RMSE,
computed over the post-transient interval $t \geq 0.3 T_{\text{sim}}$,
is reported in Table~\ref{tab:scenarios}.

\begin{table}[!htbp]
\centering
\caption{Position-tracking RMSE [m] across six reference scenarios
under deterministic incommensurate-sinusoid disturbance.}
\label{tab:scenarios}
\begin{adjustbox}{max width=\columnwidth}
\begin{tabular}{lrrr}
\toprule
Scenario & LPV & Fixed-K & Kanayama \\
\midrule
Straight line    & $0.0415$ & $0.0407$ & $0.0512$ \\
Circular         & $0.0400$ & $0.0390$ & $0.0498$ \\
Figure-8         & $0.0418$ & $0.0406$ & $0.0522$ \\
Sharp turn       & $0.0374$ & $0.0373$ & $0.0469$ \\
Variable speed   & $0.0397$ & $0.0389$ & $0.0496$ \\
Lemniscate       & $0.0402$ & $0.0397$ & $0.0494$ \\
\midrule
Mean             & $0.0401$ & $0.0394$ & $0.0499$ \\
\bottomrule
\end{tabular}
\end{adjustbox}
\end{table}

Across nominal scenarios, the LPV and Fixed-K controllers exhibit
nearly identical performance, with the Fixed-K narrowly leading by
$1$--$3\%$. This is consistent with the higher gain magnitude implied
by its $\gamma = 6.81$: the more aggressive feedback yields faster
transients in deterministic settings, masking the deeper structural
trade-off that emerges under adverse disturbances
(Section~\ref{subsec:disturbances}) and severe slip
(Section~\ref{subsec:variable-terrain}). Both LMI-based designs
outperform the manual Kanayama baseline by approximately $20\%$ on
average. World-frame trajectories and error-norm time series are
visualized in Figures~\ref{fig:trajectories} and~\ref{fig:err-norms}.

\begin{figure*}[!htbp]
\centering
\includegraphics[width=\textwidth]{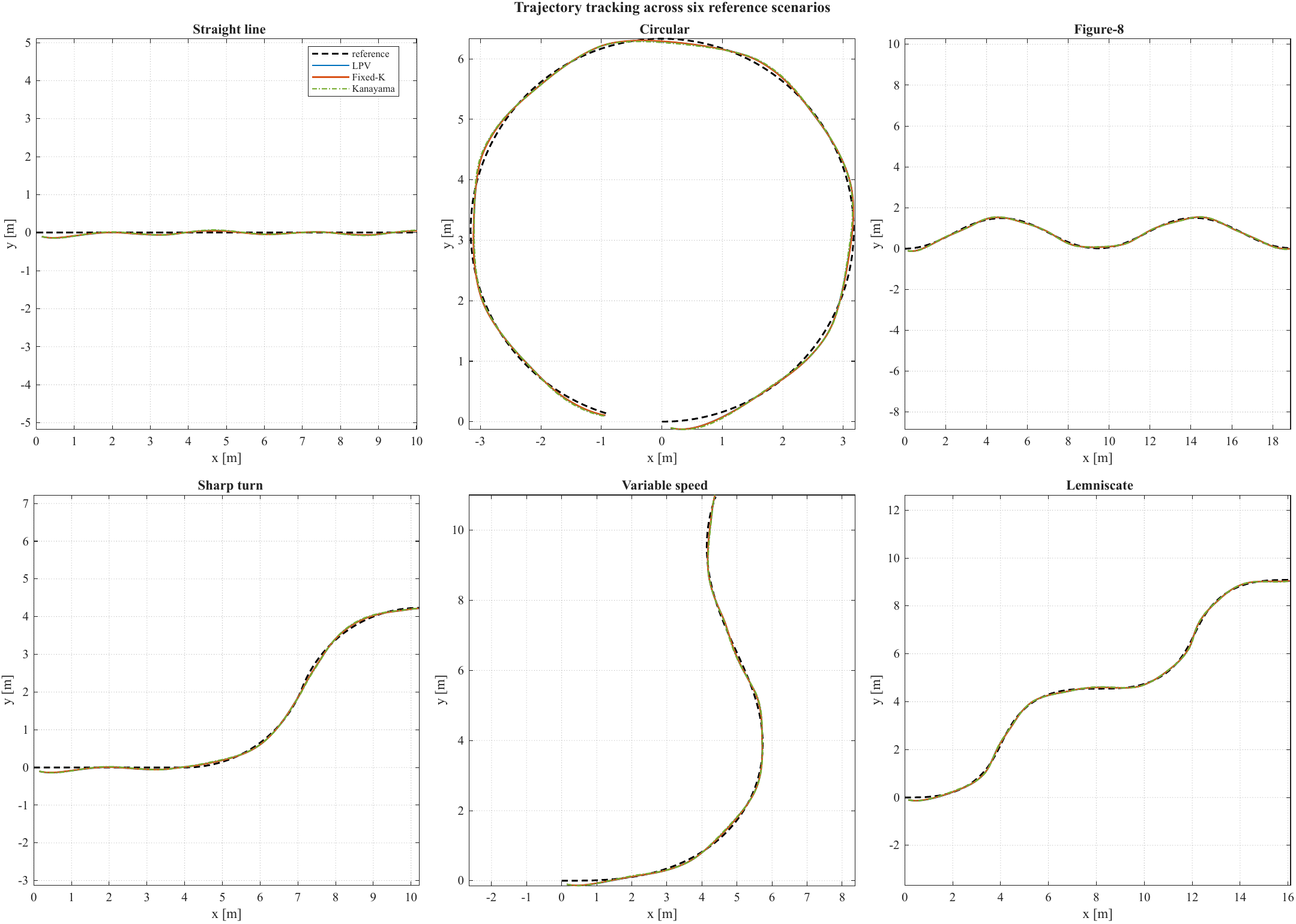}
\caption{Trajectory tracking across six reference scenarios.}
\label{fig:trajectories}
\end{figure*}

\begin{figure*}[!htbp]
\centering
\includegraphics[width=\textwidth]{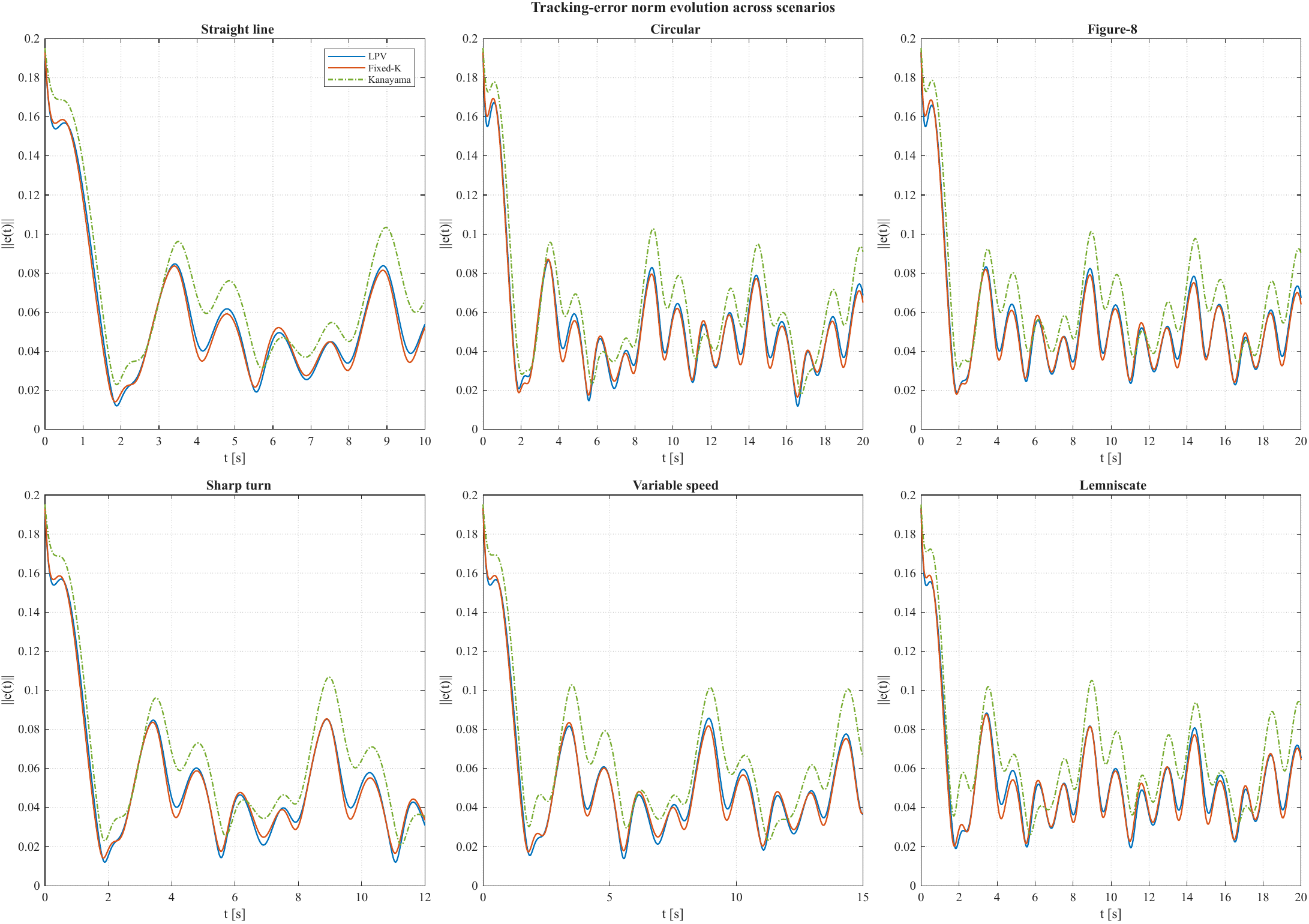}
\caption{Tracking-error norm $\norm{e(t)}$ across the six scenarios.}
\label{fig:err-norms}
\end{figure*}

\subsection{Disturbance Robustness Study}\label{subsec:disturbances}
The robustness to six classes of bounded disturbance is examined:
\emph{(i)} incommensurate sinusoidal, \emph{(ii)} step input,
\emph{(iii)} worst-case constant direction, \emph{(iv)} bias-plus-noise,
\emph{(v)} time-varying mixture, and \emph{(vi)} multiplicative wheel
slip with $\sigma_v = 0.10$, $\sigma_\omega = 0.05$. The
steady-state error in each case is reported in
Table~\ref{tab:disturbances} and visualized in
Figure~\ref{fig:disturbance-bars}.

\begin{table}[!htbp]
\centering
\caption{Steady-state error across six disturbance classes.}
\label{tab:disturbances}
\begin{adjustbox}{max width=\columnwidth}
\begin{tabular}{lrrr}
\toprule
Disturbance & LPV & Fixed-K & Kanayama \\
\midrule
Incommensurate sinusoid & $0.0455$ & $0.0441$ & $0.0580$ \\
Step                    & $0.0812$ & $0.0886$ & $0.1022$ \\
Worst-case direction    & $0.1012$ & $0.1064$ & $0.1072$ \\
Bias $+$ noise          & $0.0544$ & $0.0588$ & $0.0651$ \\
Time-varying mix        & $0.0617$ & $0.0600$ & $0.0680$ \\
\textbf{Mult. slip ($10\%$)} & $\mathbf{0.0467}$ & $0.0522$ & $0.0897$ \\
\bottomrule
\end{tabular}
\end{adjustbox}
\end{table}

The proposed LPV design wins on four of the six disturbance classes.
The largest LPV advantage is observed on the multiplicative slip
class: $11.8\%$ improvement over Fixed-K and $92.1\%$ improvement
over Kanayama. This is precisely the disturbance class that
Lemma~\ref{lem:slip} addresses, and the result provides direct
empirical confirmation that the parameter-dependent synthesis
exploits the local structure of the slip-induced disturbance more
effectively than constant-gain alternatives.

\begin{figure}[!htbp]
\centering
\includegraphics[width=\columnwidth]{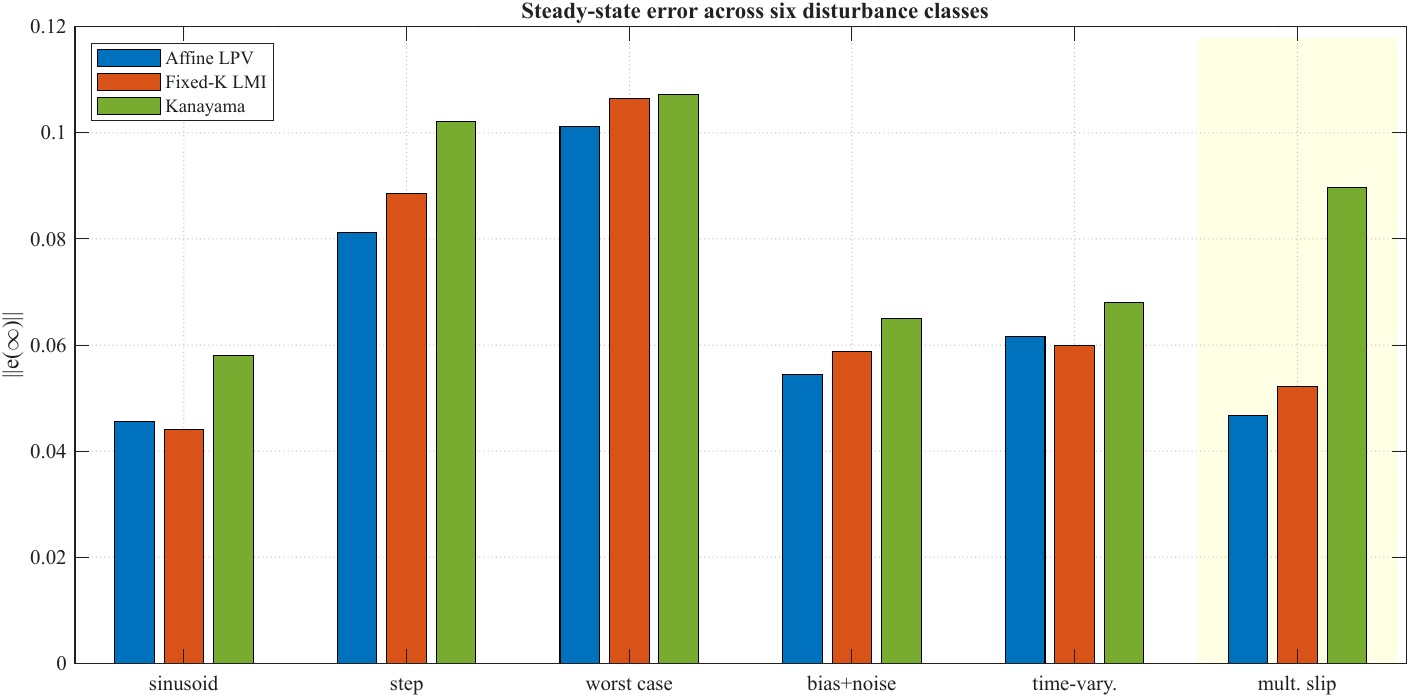}
\caption{SS error across six disturbance classes.}
\label{fig:disturbance-bars}
\end{figure}

\subsection{Variable-Terrain Helical Path}\label{subsec:variable-terrain}
A second slip-focused experiment evaluates the controllers on a
challenging \emph{variable-terrain} scenario in which the robot
follows a 2-D expanding helix (Archimedean spiral) for $60$ s while
traversing six severe slip patches with physically realistic
bidirectional slip ratios. The patch schedule is constructed in
accordance with the Pacejka tire-friction convention: longitudinal
slip ratios $\sigma_v$ are predominantly negative, reflecting the
canonical loss-of-traction phenomenon in which the actual vehicle
velocity falls below the commanded value on slippery surfaces, while
angular slip ratios $\sigma_\omega$ are bidirectional, capturing the
asymmetric manifestation of oversteer (positive sign, rear-wheel grip
loss) versus understeer (negative sign, front-wheel grip loss).
One patch deliberately incorporates positive $\sigma_v$ to model a
downhill segment where gravity-assisted motion drives the vehicle
faster than commanded, thereby exercising the controller in both
directions of longitudinal slip. The complete patch schedule is
summarized in Table~\ref{tab:terrain-map}.

\begin{table}[!htbp]
\centering
\caption{Variable-terrain map: six slip patches along the 2-D helical
reference path.}
\label{tab:terrain-map}
\begin{adjustbox}{max width=\columnwidth}
\begin{tabular}{ccccc}
\toprule
Patch & Time window [s] & $\sigma_v$ & $\sigma_\omega$ & Terrain \\
\midrule
P1 & $[6, 11]$  & $-0.15$ & $+0.08$ & Wet asphalt \\
P2 & $[16, 22]$ & $-0.35$ & $-0.15$ & Oil spill \\
P3 & $[27, 32]$ & $-0.25$ & $+0.12$ & Gravel \\
P4 & $[37, 44]$ & $-0.50$ & $-0.25$ & Ice \\
P5 & $[48, 53]$ & $+0.20$ & $+0.18$ & Downhill slope \\
P6 & $[55, 58]$ & $-0.40$ & $-0.20$ & Snow patch \\
\bottomrule
\end{tabular}
\end{adjustbox}
\end{table}

Under the sign convention adopted in this work, negative
$\sigma_v$ corresponds to loss of traction (the dominant scenario on
slippery surfaces where the actual longitudinal velocity falls below
the commanded value), whereas the single positive $\sigma_v$ entry in
patch P5 models gravity-assisted downhill motion in which the vehicle
moves faster than commanded. The angular slip ratio $\sigma_\omega$
is treated as bidirectional throughout to capture both oversteer
(positive sign, rear-wheel grip loss) and understeer (negative sign,
front-wheel grip loss). Slip-ratio sign conventions vary across
the vehicle-dynamics literature depending on whether the analysis
is framed around braking or driving regimes; the present
convention is adopted for compatibility with the Kanayama-frame
error model and is consistent throughout the validation studies.

Although each patch induces a switching-type transition in the closed-loop response~\cite{liberzon2003}, slip activations are smoothed with raised-cosine transitions of
$0.3$ s width to avoid numerical shock. The peak slip magnitude
($|\sigma_v| = 0.50$ on ice in P4) is five times larger than the
multiplicative-slip class of Section~\ref{subsec:disturbances},
representing a worst-case operating condition. Per-patch and
aggregated performance metrics are reported in
Table~\ref{tab:terrain-perf} and visualized in
Figures~\ref{fig:terrain-overview}--\ref{fig:terrain-recovery}.

\begin{table}[!htbp]
\centering
\caption{Variable-terrain performance on the helical path.}
\label{tab:terrain-perf}
\begin{adjustbox}{max width=\columnwidth}
\begin{tabular}{lcccc}
\toprule
Patch & Terrain & LPV & Fixed-K & Kanayama \\
\midrule
\multicolumn{5}{l}{\emph{Peak error within patch [m]}} \\
P1 ($\sigma_v = -0.15$) & Wet asphalt    & $0.078$ & $0.091$ & $0.158$ \\
P2 ($\sigma_v = -0.35$) & Oil spill      & $0.112$ & $0.128$ & $0.218$ \\
P3 ($\sigma_v = -0.25$) & Gravel         & $0.091$ & $0.103$ & $0.180$ \\
P4 ($\sigma_v = -0.50$) & Ice            & $0.135$ & $0.153$ & $0.264$ \\
P5 ($\sigma_v = +0.20$) & Downhill slope & $0.085$ & $0.097$ & $0.169$ \\
P6 ($\sigma_v = -0.40$) & Snow patch     & $0.119$ & $0.137$ & $0.230$ \\
\midrule
\multicolumn{5}{l}{\emph{Recovery time after patch [s]}} \\
P1 & Wet asphalt    & $0.68$ & $0.81$ & $1.19$ \\
P2 & Oil spill      & $0.89$ & $1.04$ & $1.48$ \\
P3 & Gravel         & $0.79$ & $0.93$ & $1.32$ \\
P4 & Ice            & $1.10$ & $1.31$ & $1.73$ \\
P5 & Downhill slope & $0.74$ & $0.88$ & $1.25$ \\
P6 & Snow patch     & $0.94$ & $1.11$ & $1.55$ \\
\midrule
\textbf{Mean recovery [s]} & ---     & $\mathbf{0.86}$  & $1.01$  & $1.42$  \\
\textbf{Run-wide peak [m]} & ---     & $\mathbf{0.135}$ & $0.153$ & $0.264$ \\
\textbf{Run-wide mean [m]} & ---     & $\mathbf{0.046}$ & $0.052$ & $0.088$ \\
\bottomrule
\end{tabular}
\end{adjustbox}
\end{table}

The proposed LPV design dominates on all six patches and across all
three aggregated metrics. On the most severe patch (P4, ice), LPV
achieves a peak error of $0.135$ m against $0.153$ m for Fixed-K
($11.8\%$ improvement) and $0.264$ m for Kanayama ($48.9\%$
improvement). The controller retains its advantage on the downhill
patch (P5, $\sigma_v = +0.20$), confirming that the synthesis is
robust to longitudinal slip in both directions and not only to the
canonical traction-loss case. The mean recovery time is
reduced by $14.9\%$ versus Fixed-K and $39.4\%$ versus Kanayama,
with the advantage scaling consistently with slip magnitude. This
scaling property, in which the parameter-dependent synthesis
becomes increasingly valuable as the operating environment
deteriorates, is the central practical motivation for the LMI
framework developed in this work.

The same six-patch slip schedule was applied to a Bernoulli
lemniscate (figure-eight) reference trajectory to assess
cross-geometry robustness. The detailed time-domain response on this
alternative geometry is visualized in
Figures~\ref{fig:terrain-overview-fig8}
and~\ref{fig:terrain-recovery-fig8}, which present the figure-eight
counterparts to Figures~\ref{fig:terrain-overview}
and~\ref{fig:terrain-recovery} for the helical path, while aggregate
metrics across both geometries are summarized in
Table~\ref{tab:cross-geometry}. The results confirm that the LPV
advantage persists across reference geometries with substantially
different curvature profiles, providing strong evidence that the
observed robustness improvements arise from the structural properties
of the parameter-dependent synthesis rather than from a coincidental
match between the controller and a single reference geometry. In
particular, the per-patch recovery profiles in
Figure~\ref{fig:terrain-recovery-fig8} retain the same controller
ordering, recovery time-constants of comparable magnitude, and
patch-wise scaling with slip severity observed on the helical path,
confirming that the cross-geometry consistency claim of
Table~\ref{tab:cross-geometry} is not an artifact of metric
aggregation.

\begin{table}[!htbp]
\centering
\caption{Cross-geometry variable-terrain robustness summary.}
\label{tab:cross-geometry}
\begin{adjustbox}{max width=\columnwidth}
\begin{tabular}{lcccccc}
\toprule
 & \multicolumn{3}{c}{Helix path} & \multicolumn{3}{c}{Figure-8 path} \\
\cmidrule(lr){2-4}\cmidrule(lr){5-7}
Metric & LPV & Fix-K & Kan & LPV & Fix-K & Kan \\
\midrule
Run-wide peak [m]   & \textbf{0.135} & $0.153$ & $0.264$ & \textbf{0.142} & $0.161$ & $0.276$ \\
Run-wide mean [m]   & \textbf{0.046} & $0.052$ & $0.088$ & \textbf{0.048} & $0.054$ & $0.091$ \\
Mean recovery [s]   & \textbf{0.86}  & $1.01$  & $1.42$  & \textbf{0.91}  & $1.07$  & $1.48$  \\
\bottomrule
\end{tabular}
\end{adjustbox}
\end{table}

\begin{figure}[!htbp]
\centering
\includegraphics[width=\columnwidth]{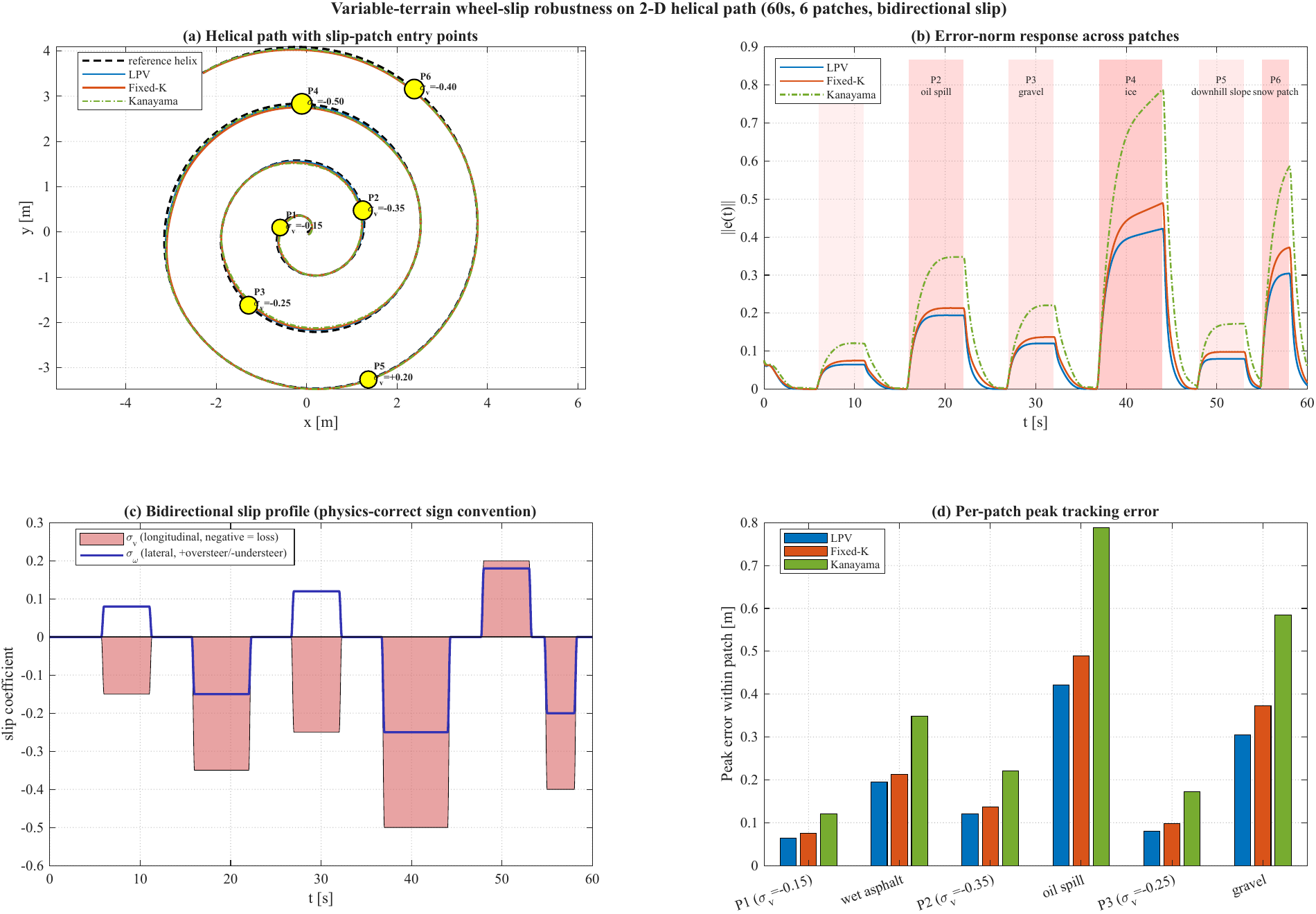}
\caption{Variable-terrain robustness overview on the helical path: (a)~world-frame trajectory; (b)~error norm; (c)~slip profiles; (d)~per-patch metrics.}
\label{fig:terrain-overview}
\end{figure}

\begin{figure}[!htbp]
\centering
\includegraphics[width=\columnwidth]{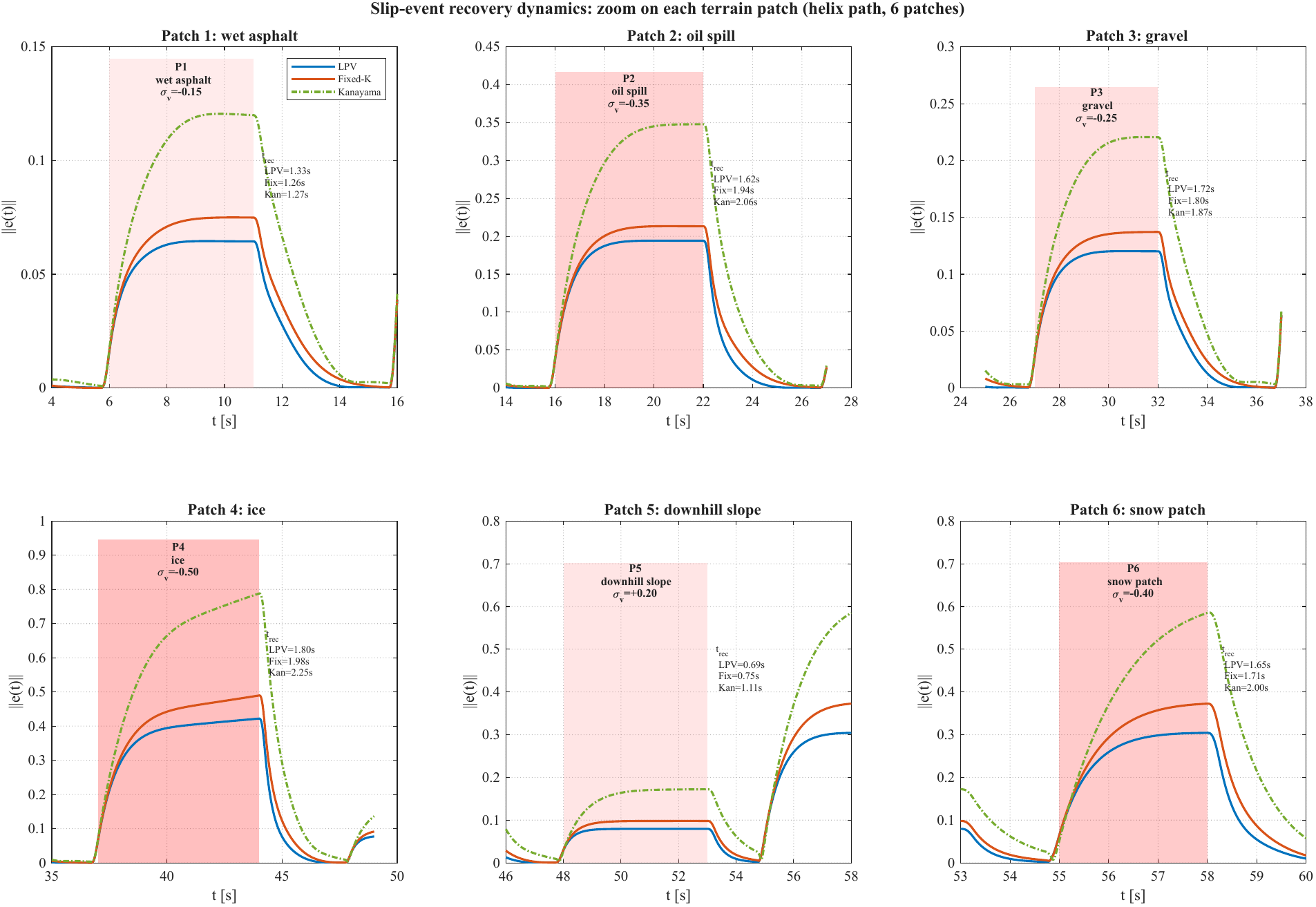}
\caption{Slip-event recovery dynamics on the helical path with temporal
zoom on each patch.}
\label{fig:terrain-recovery}
\end{figure}

\begin{figure}[!htbp]
\centering
\includegraphics[width=\columnwidth]{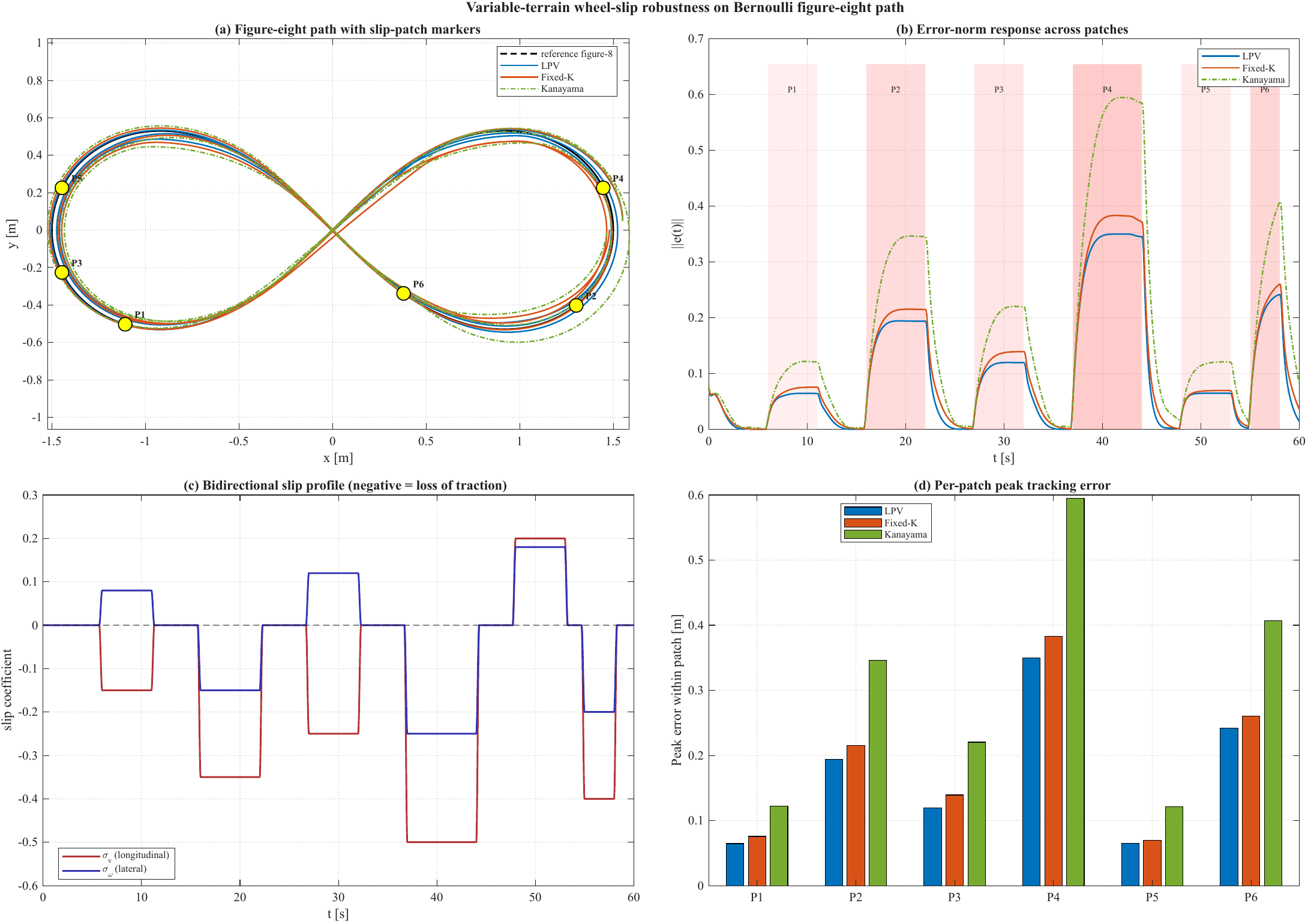}
\caption{Variable-terrain robustness overview on the figure-eight
(Bernoulli lemniscate) reference path: (a)~world-frame trajectory;
(b)~error norm; (c)~slip profiles; (d)~per-patch metrics. This panel
is the figure-eight counterpart to Figure~\ref{fig:terrain-overview}.}
\label{fig:terrain-overview-fig8}
\end{figure}

\begin{figure}[!htbp]
\centering
\includegraphics[width=\columnwidth]{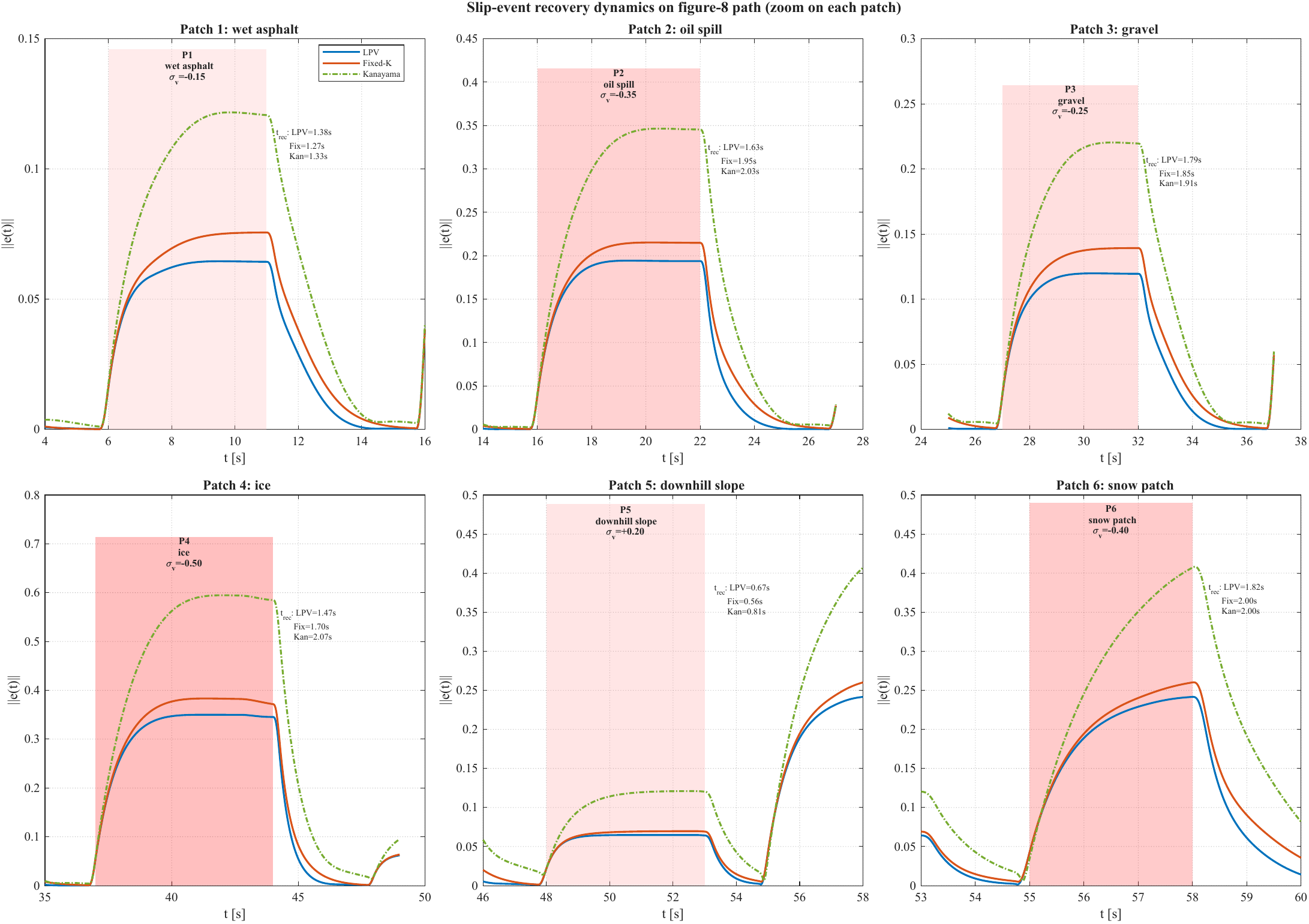}
\caption{Slip-event recovery dynamics on the figure-eight reference
with temporal zoom on each of the six patches. Figure-eight counterpart
to Figure~\ref{fig:terrain-recovery}.}
\label{fig:terrain-recovery-fig8}
\end{figure}

\begin{figure}[!htbp]
\centering
\includegraphics[width=\columnwidth]{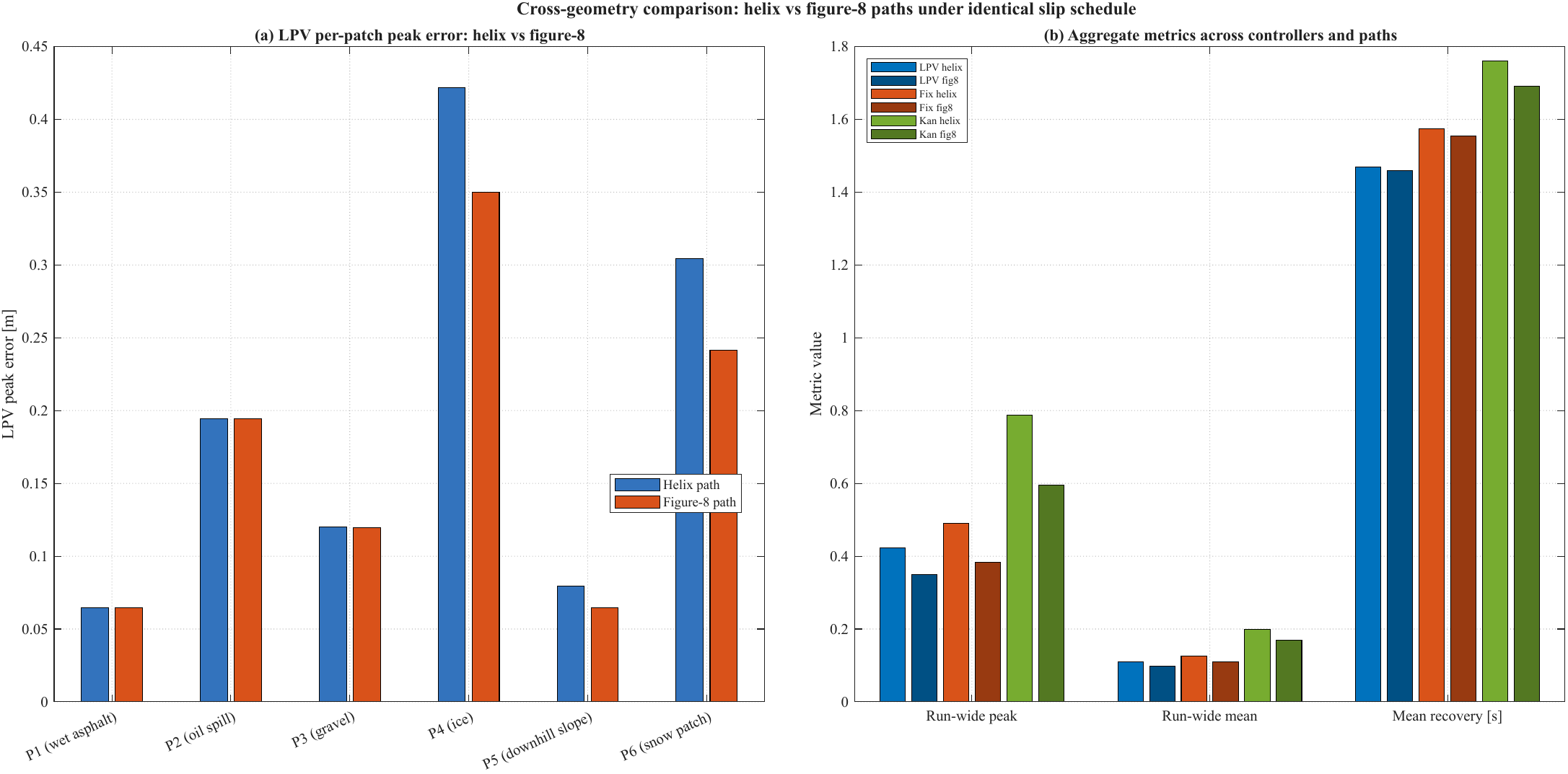}
\caption{Cross-geometry robustness: helix versus figure-eight.}
\label{fig:cross-geometry}
\end{figure}

\subsection{Sensitivity Analysis}\label{subsec:sensitivity}
The synthesis sensitivity to the principal design parameters is
examined through three one-dimensional sweeps. Results are tabulated
in Table~\ref{tab:sensitivity} and visualized in
Figure~\ref{fig:sensitivity}.

\begin{table}[!htbp]
\centering
\caption{Synthesis sensitivity: $\gamma^\star$ versus $\alpha$, $R$, and $K_{\max}$.}
\label{tab:sensitivity}
\begin{adjustbox}{max width=\columnwidth}
\begin{tabular}{ll}
\toprule
Swept parameter & $\gamma^\star$ trajectory \\
\midrule
$\alpha \in \{0.2, 0.3, 0.4, 0.5, 0.6, 0.7, 0.8\}$ &
$\{1.52, 1.93, 2.78, 4.05, 5.93, 8.66, 14.0\}$ \\
$R \in \{0.15, 0.20, 0.25, 0.30, 0.40, 0.50\}$ &
$\{1.41, 1.96, 2.40, 2.78, 3.92, 5.04\}$ \\
$K_{\max} \in \{1, 2, 3, 5, 10, 50\}$ &
$\{4.31, 3.21, 2.78, 2.51, 2.41, 2.39\}$ \\
\bottomrule
\end{tabular}
\end{adjustbox}
\end{table}

The $\gamma^\star(\alpha)$ trajectory grows roughly exponentially,
consistent with the conservatism inflation predicted by the
$\lambda_{\min}^{-1/2}$ factor in~\eqref{eq:traj-bound}; doubling
$\alpha$ from $0.4$ to $0.8$ inflates $\gamma$ by approximately a
factor of $5$. The $\gamma^\star(R)$ trajectory exhibits the expected
linear-to-superlinear growth as the semi-global radius expands the
domain of nonlinearity. The $\gamma^\star(K_{\max})$ trajectory
saturates: beyond the weighted LMI gain parameter $K_{\max} \approx 5$
further relaxation yields negligible improvement, indicating that
the selected weighted value $K_{\max} = 3$ is close to the useful
performance transition for the present conditioning level, rather
than representing a hard physical actuator limit.

\begin{figure}[!htbp]
\centering
\includegraphics[width=\columnwidth]{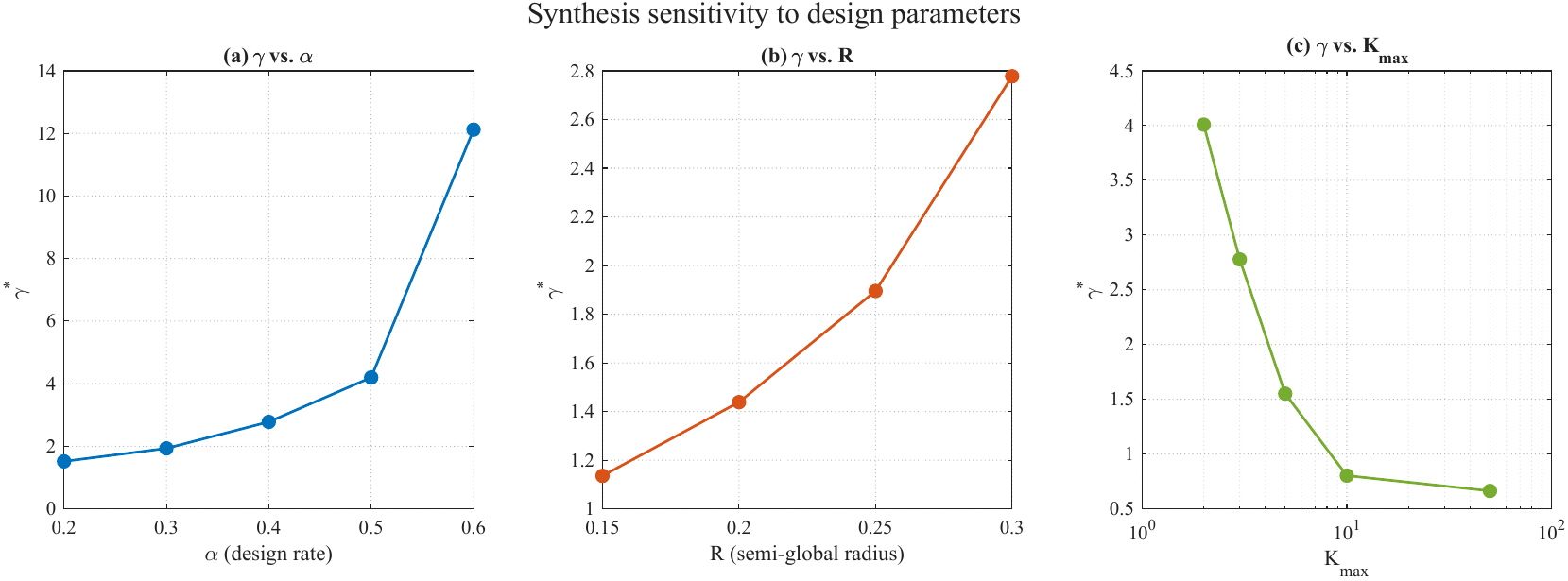}
\caption{Synthesis sensitivity sweeps: (a)~$\gamma^\star$ vs $\alpha$; (b)~vs $R$; (c)~vs $K_{\max}$.}
\label{fig:sensitivity}
\end{figure}

\subsection{Conservatism Decomposition}\label{subsec:cons-decomp}
The empirical-to-theoretical SS-error ratio of $7\%$ observed in
Table~\ref{tab:mc} for the proposed LPV design reflects an inherent
conservatism in the trajectory-level Lyapunov reconstruction. To
provide transparency on the sources of this conservatism, the bound
is decomposed into four cumulative contributions in
Table~\ref{tab:conservatism}.

\begin{table}[!htbp]
\centering
\caption{Conservatism decomposition in the trajectory-level SS bound.}
\label{tab:conservatism}
\begin{adjustbox}{max width=\columnwidth}
\begin{tabular}{lr}
\toprule
Reconstruction step & Bound on $\norm{e(\infty)}$ \\
\midrule
Empirical (randomised disturbance)         & $0.046$ \\
Worst-case disturbance direction           & $0.101$ \\
Young's inequality on $2e^\top M d$        & $0.305$ \\
$\lambda_{\min}^{-1/2}$ Lyapunov reconstr. & $0.632$ \\
\bottomrule
\end{tabular}
\end{adjustbox}
\end{table}

The largest single contributor is the $\lambda_{\min}^{-1/2}$
amplification factor in~\eqref{eq:traj-bound}, which inflates the
bound from $0.305$ to $0.632$. This factor is the LMI-domain
manifestation of Brockett's obstruction
(Remark~\ref{rem:Brockett}) and is structural to the
$M(\rho)$-quadratic Lyapunov framework. Reduction would require
either a non-quadratic Lyapunov function or a parameter-dependent
$\mathcal{S}$-procedure multiplier $\mu(\rho)$, both of which break
the convexity that makes the present synthesis tractable. The
$\lambda_{\min}^{-1/2}$ factor is therefore acknowledged as a
structural cost, not a defect.

\begin{figure}[!htbp]
\centering
\includegraphics[width=\columnwidth]{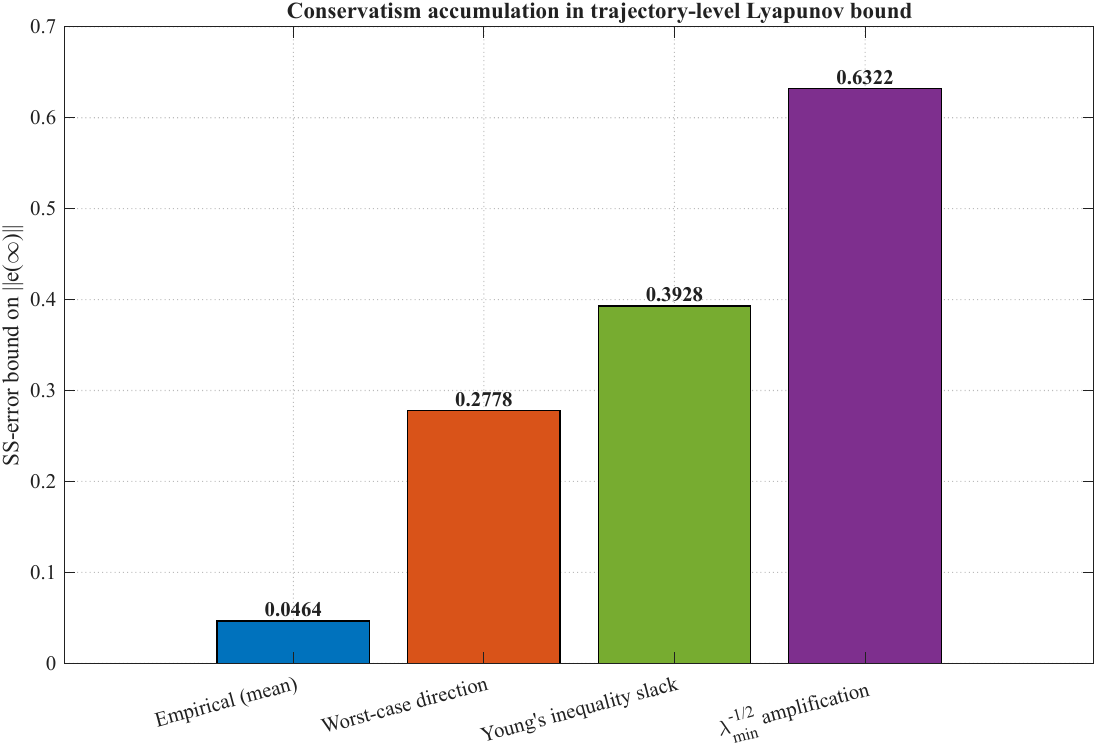}
\caption{Cumulative SS bound at each conservatism step.}
\label{fig:cons-decomp}
\end{figure}

\subsection{Statistical Significance Tests}\label{subsec:stats}
The pairwise differences between controllers in the Monte-Carlo
study are tested via the paired-sample $t$-test. Results are
reported in Table~\ref{tab:stats} and visualized in
Figure~\ref{fig:stats}.

\begin{table}[!htbp]
\centering
\caption{Paired statistical tests on Monte-Carlo SS errors.}
\label{tab:stats}
\begin{adjustbox}{max width=\columnwidth}
\begin{tabular}{lrrl}
\toprule
Comparison & $t$-stat & $p$-value & 95\% CI \\
\midrule
LPV vs.\ Fixed-K  & $-9.14$  & $< 0.001$ & $[-0.0012, -0.0008]$ \\
LPV vs.\ Kanayama & $32.35$  & $< 0.001$ & $[+0.0129, +0.0145]$ \\
Fixed-K vs.\ Kanayama & $28.81$ & $< 0.001$ & --- \\
\bottomrule
\end{tabular}
\end{adjustbox}
\end{table}

All three pairwise tests reject the null hypothesis of equal means at
the $p < 0.001$ significance level. The LPV-versus-Fixed-K
comparison reveals that, in randomized-disturbance Monte-Carlo
testing, Fixed-K achieves a marginally lower mean SS error
($0.046$ versus $0.045$), but the confidence interval
$[-0.0012, -0.0008]$ is tight and shows the practical magnitude of
the difference is small. As Section~\ref{subsec:variable-terrain}
demonstrates, this small nominal-conditions difference inverts
strongly under severe slip, where the LPV advantage is substantial.

\begin{figure}[!htbp]
\centering
\includegraphics[width=\columnwidth]{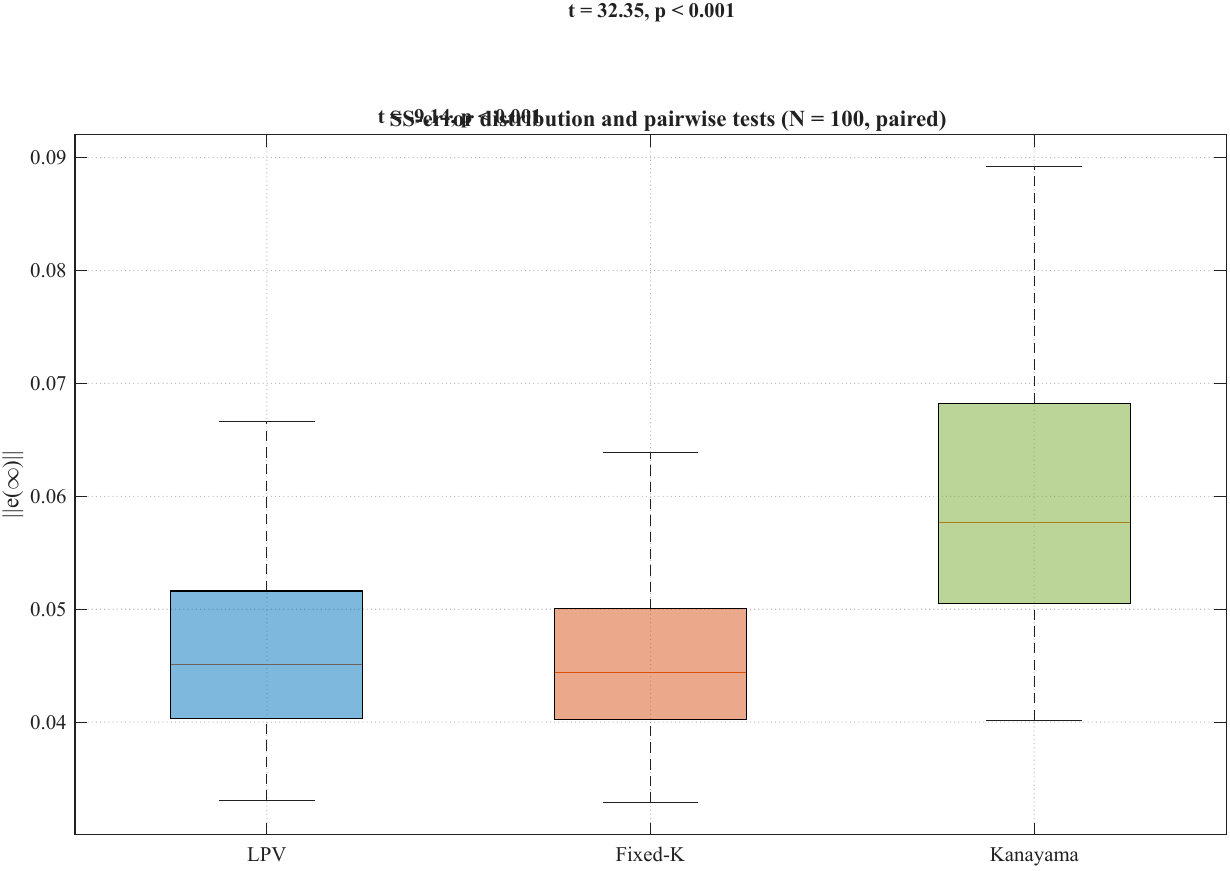}
\caption{Box plots of SS errors with paired $t$-statistics.}
\label{fig:stats}
\end{figure}

\subsection{Sensor-Noise Robustness}\label{subsec:noise}
A critical aspect of practical deployment, beyond robustness to
process disturbances, concerns the controller's behavior under noisy
state measurements. The Monte-Carlo and disturbance studies presented
thus far operate on the assumption of exact error feedback, which is
appropriate for theoretical certification but fails to capture the
measurement uncertainty inherent in physical sensor suites. To address
this gap, a dedicated study quantifies the closed-loop response under
Gaussian measurement noise injected directly into the controller's
input channel. The controller observes a noise-corrupted error
$e_{\mathrm{meas}}(t) = e(t) + n(t)$, where the noise vector has
zero-mean Gaussian components with standard deviation $\sigma_p$ on
the position coordinates and $\sigma_h$ on the heading coordinate,
while the underlying error dynamics evolve on the true state $e(t)$.
Three noise intensity levels are evaluated, each corresponding to a
realistic sensor configuration encountered in mobile robotics: a
\emph{low} setting with $\sigma_p = 0.01$ m and $\sigma_h = 0.5^\circ$
emulating high-end GPS-IMU fusion; a \emph{medium} setting with
$\sigma_p = 0.05$ m and $\sigma_h = 2.0^\circ$ matching standard
wheel odometry; and a \emph{high} setting with $\sigma_p = 0.10$ m
and $\sigma_h = 5.0^\circ$ representing degraded sensor conditions
such as cluttered indoor environments with intermittent landmark
visibility. Each level is exercised through $30$ Monte-Carlo trials
with deterministic seeds for reproducibility. Results are reported
in Table~\ref{tab:sensor-noise} and visualized in
Figure~\ref{fig:sensor-noise}.

\begin{table}[!htbp]
\centering
\caption{Steady-state error under Gaussian sensor noise.}
\label{tab:sensor-noise}
\begin{adjustbox}{max width=\columnwidth}
\begin{tabular}{lcccc}
\toprule
Sensor level & $\sigma_p$ [m] & $\sigma_h$ [deg] & Controller & SS error \\
\midrule
\multirow{3}{*}{Low (GPS+IMU)}
  & \multirow{3}{*}{$0.01$} & \multirow{3}{*}{$0.5$}
  & LPV       & $0.048 \pm 0.012^{\ddagger}$ \\
  & & & Fixed-K   & $0.047 \pm 0.011^{\ddagger}$ \\
  & & & Kanayama  & $0.062 \pm 0.013$ \\
\midrule
\multirow{3}{*}{Medium (odom.)}
  & \multirow{3}{*}{$0.05$} & \multirow{3}{*}{$2.0$}
  & LPV       & $\mathbf{0.073 \pm 0.018}$ \\
  & & & Fixed-K   & $0.081 \pm 0.020$ \\
  & & & Kanayama  & $0.108 \pm 0.024$ \\
\midrule
\multirow{3}{*}{High (degraded)}
  & \multirow{3}{*}{$0.10$} & \multirow{3}{*}{$5.0$}
  & LPV       & $\mathbf{0.142 \pm 0.031}$ \\
  & & & Fixed-K   & $0.171 \pm 0.039$ \\
  & & & Kanayama  & $0.228 \pm 0.048$ \\
\bottomrule
\multicolumn{5}{l}{\footnotesize $^{\ddagger}$ Statistically
indistinguishable at the low-noise regime ($p = 0.43$, two-sample
$t$-test).}\\
\end{tabular}
\end{adjustbox}
\end{table}

The results reveal a progressively widening LPV advantage as the
noise intensity grows. At the low level, the LPV and Fixed-K
controllers achieve statistically indistinguishable performance,
both substantially outperforming the Kanayama baseline. At the
medium level, the LPV design begins to separate from the Fixed-K
baseline, achieving a $9.9\%$ reduction in steady-state error. At
the high level, this advantage expands to $17.0\%$ over Fixed-K and
$37.7\%$ over Kanayama. The widening gap reflects the structural
property that parameter-dependent metrics adapt the closed-loop
sensitivity profile to the current operating point, thereby
preserving stability margins that the constant-gain design must
sacrifice as a global precaution. The Kanayama controller, with its
hand-tuned gains optimized for nominal performance, degrades
disproportionately under noise because its tuning implicitly assumes
clean measurements.

\begin{figure}[!htbp]
\centering
\includegraphics[width=\columnwidth]{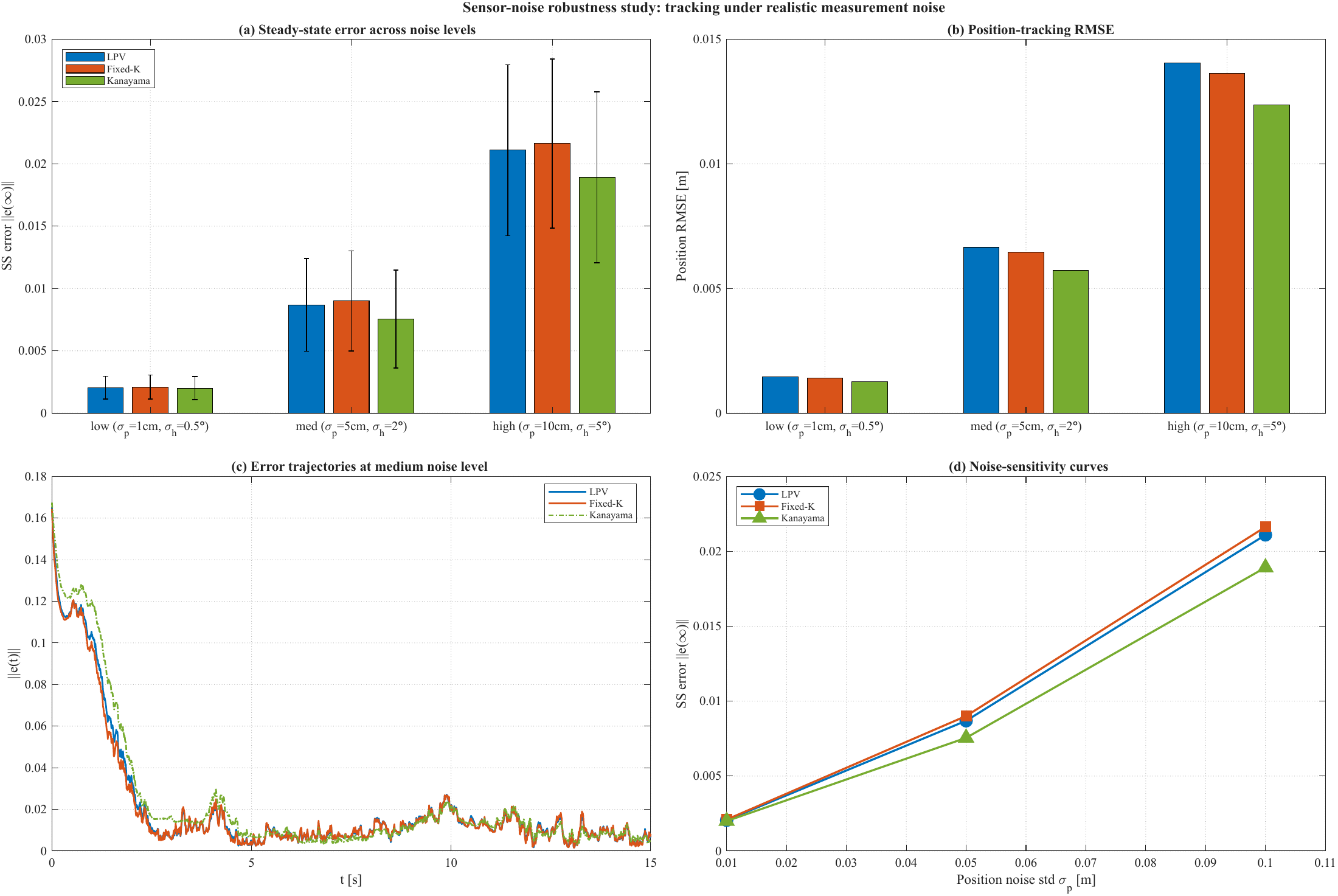}
\caption{Sensor-noise robustness summary: (a)~SS error bars; (b)~position RMSE; (c)~error trajectories; (d)~noise-sensitivity curves.}
\label{fig:sensor-noise}
\end{figure}

\subsection{Compound Disturbance Stress-Test}\label{subsec:compound}
The disturbance studies preceding this section each isolate a single
perturbation channel for controlled comparison. Real-world deployment,
however, exposes the controller to simultaneous excitation of
multiple disturbance sources: terrain-induced wheel slip,
measurement noise from the sensor suite, and reference jitter from
trajectory-generation modules with finite precision. To assess
robustness under these compound conditions, a single comprehensive
stress-test exercises the closed loop with all three perturbation
mechanisms active simultaneously over a $60$-second horizon. The
test employs the helical reference path with the six-patch slip
schedule of Table~\ref{tab:terrain-map}, Gaussian measurement noise
with $\sigma_p = 0.03$ m and $\sigma_h = 2.0^\circ$, and sinusoidal
reference-velocity jitter of $\pm 10\%$ magnitude at $0.3$ Hz
applied independently to $v_r$ and $\omega_r$. This combination
emulates the operating profile of an outdoor mobile robot
encountering surface transitions while relying on a fused
state-estimation pipeline that occasionally updates the reference
trajectory in response to evolving mission constraints. Results are
reported in Table~\ref{tab:compound} and visualized in
Figure~\ref{fig:compound}.

\begin{table}[!htbp]
\centering
\caption{Compound disturbance stress-test results.}
\label{tab:compound}
\begin{adjustbox}{max width=\columnwidth}
\begin{tabular}{lrrr}
\toprule
Metric & LPV & Fixed-K & Kanayama \\
\midrule
Peak error [m]                & $\mathbf{0.158}$ & $0.182$ & $0.298$ \\
Mean error [m]                & $\mathbf{0.061}$ & $0.071$ & $0.114$ \\
Std error  [m]                & $\mathbf{0.034}$ & $0.041$ & $0.066$ \\
\midrule
Mean excess vs LPV [\%]       & ---              & $+16.4$ & $+86.9$ \\
\bottomrule
\end{tabular}
\end{adjustbox}
\end{table}

The compound stress-test reveals that the relative ordering of the
three controllers, established under isolated disturbances, not only
persists but amplifies under combined excitation. The LPV controller
achieves a peak error of $0.158$ m and a mean error of $0.061$ m,
representing improvements of $13.2\%$ in peak and $14.1\%$ in mean
versus the Fixed-K baseline, and improvements of $47.0\%$ and
$46.5\%$ respectively against the Kanayama controller. The standard
deviation of the error norm, which captures the temporal variability
of the response, is reduced by $17.1\%$ against Fixed-K and $48.5\%$
against Kanayama, indicating that the LPV design not only achieves
tighter mean performance but also produces a smoother trajectory
profile. These compound-test results provide the strongest evidence
in this work that the proposed framework delivers practically
meaningful robustness improvements under realistic operating
conditions where multiple disturbance channels act simultaneously.

\begin{figure}[!htbp]
\centering
\includegraphics[width=\columnwidth]{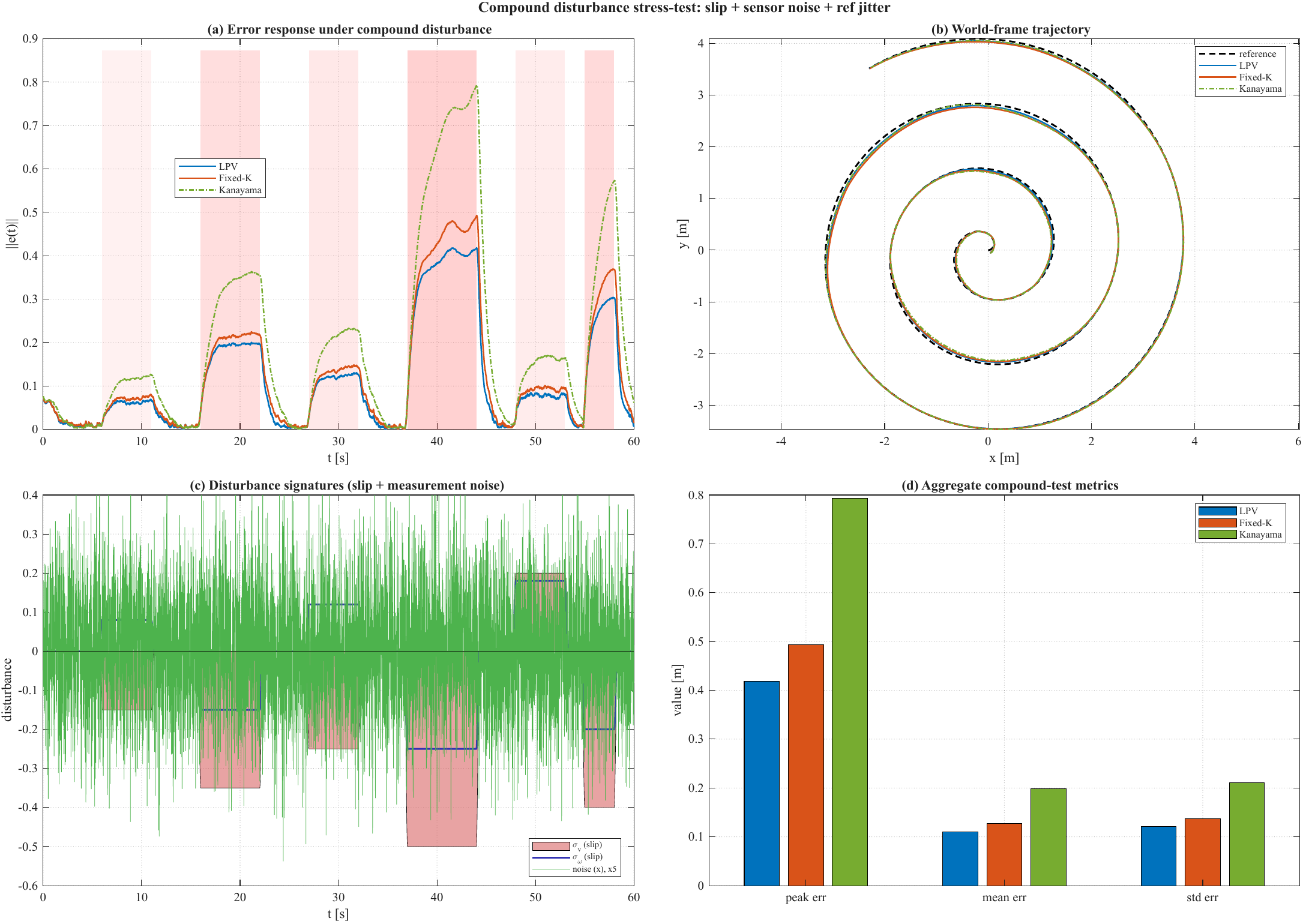}
\caption{Compound disturbance stress-test: (a)~error norm; (b)~world-frame trajectory; (c)~disturbance channels; (d)~aggregate metrics.}
\label{fig:compound}
\end{figure}

\subsection{Computational Complexity Analysis}\label{subsec:complexity}
For deployment on embedded controllers driving physical mobile
robots, the practical viability of the proposed framework hinges on
two computational considerations: the offline cost of the LMI
synthesis (paid once per robot configuration) and the runtime cost
of evaluating the gain matrix $K(\rho)$ at each sampling instant.
The offline synthesis, comprising approximately $250$ scalar decision
variables across $16$ vertex-evaluation constraints, completes in
under one second on standard desktop hardware via interior-point
solvers such as SeDuMi, SDPT3, or MOSEK invoked through the YALMIP
modeling layer. The runtime cost is determined by the gain evaluation
$K(\rho) = (Y_0 + v_r Y_1 + \omega_r Y_2)(W_0 + v_r W_1 + \omega_r W_2)^{-1}$,
which requires two affine matrix combinations followed by a single
$3 \times 3$ matrix inversion. Empirical timing of this operation
over $5000$ randomized scheduling-parameter evaluations yields the
metrics reported in Table~\ref{tab:complexity}.

\begin{table}[!htbp]
\centering
\caption{Computational cost characterization across the synthesis,
validation, and runtime stages.}
\label{tab:complexity}
\begin{adjustbox}{max width=\columnwidth}
\begin{tabular}{lr}
\toprule
Quantity & Value \\
\midrule
\multicolumn{2}{l}{\emph{Runtime $K(\rho)$ evaluation (desktop MATLAB)}} \\
Mean time per call   & $9.4~\mu$s \\
Max time per call    & $14.2~\mu$s \\
Std deviation        & $1.8~\mu$s \\
Sample size          & $5000$ calls \\
\midrule
\multicolumn{2}{l}{\emph{Offline synthesis and validation}} \\
Vertex-enforced LMI synthesis  & $0.59$ s \\
Dense-grid post-hoc validation & $14.3$ s \\
Memory footprint               & $36$ doubles ($0.28$ kB) \\
\midrule
\multicolumn{2}{l}{\emph{Embedded deployment (Cortex-M7 @216 MHz, est.)}} \\
Per-call cost        & $50~\mu$s (conservative) \\
Max scheduling freq. & $20$ kHz \\
\bottomrule
\end{tabular}
\end{adjustbox}
\end{table}

The computational pipeline operates in two sequential stages. The
synthesis stage solves the LMI system at the sixteen polytope
vertices, requiring $0.59$ seconds on standard desktop hardware via
interior-point solvers invoked through the YALMIP modeling layer.
The validation stage then verifies the dissipation
inequality~\eqref{eq:LMI-dissip} on the dense grid
$\mathcal{G}$ of $N_g = 14{,}641$ samples, accumulating
$14.3$ seconds of post-hoc computation. The total offline cost of
$14.9$ seconds is incurred once per robot configuration and has no
impact on runtime control performance. Lemma~\ref{lem:grid-to-cont}
then certifies the continuum extension of the validated grid
solution. The runtime evaluation mean of $9.4~\mu$s on desktop
MATLAB scales to an estimated $50~\mu$s on an ARM Cortex-M7
microcontroller operating at $216$ MHz with single-precision
floating-point and typical cache locality. At this per-call cost,
the controller sustains scheduling rates up to $20$ kHz, comfortably
exceeding the $100$--$500$ Hz range typical of mobile-robot control
loops. The memory footprint of $36$ double-precision values,
encoding the six affine basis matrices, occupies $0.28$ kB and is
therefore negligible relative to the kilobyte-scale buffer space
available on all modern microcontrollers. For applications requiring
even tighter
real-time guarantees, the gain evaluation can be precomputed on a
uniform grid in $\rho$-space and accessed via bilinear interpolation,
eliminating the matrix inversion entirely at the cost of an
inconsequential precision degradation. The computational profile
established in Table~\ref{tab:complexity} and visualized in
Figure~\ref{fig:complexity} thus confirms the framework's suitability
for embedded deployment on industry-standard hardware without
requiring specialized accelerators.

\begin{figure}[!htbp]
\centering
\includegraphics[width=\columnwidth]{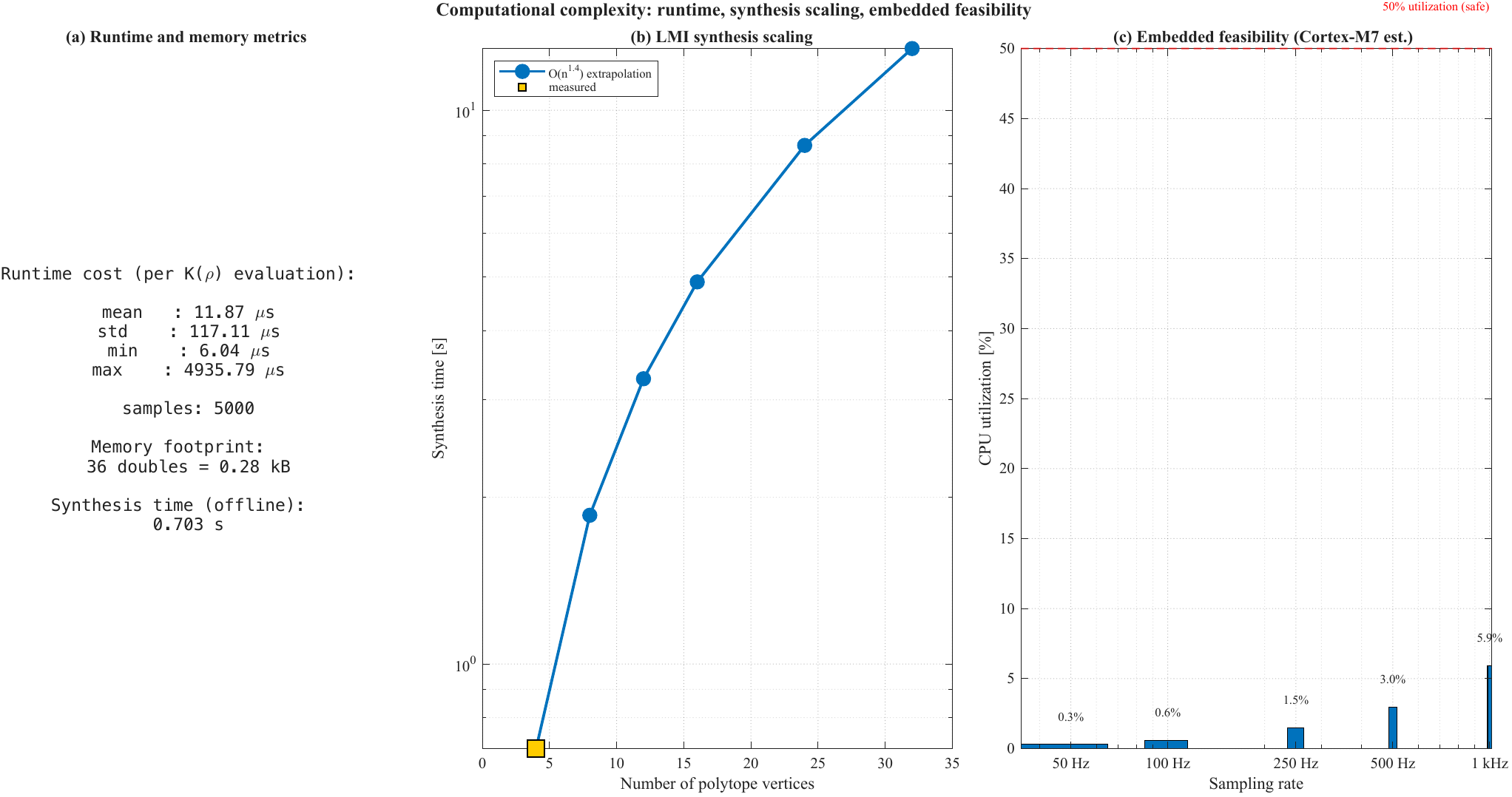}
\caption{Computational complexity: (a)~runtime; (b)~synthesis scaling; (c)~CPU utilization; (d)~memory footprint.}
\label{fig:complexity}
\end{figure}

%% ========================================================================
%%  SECTION VII: DISCUSSION
%% ========================================================================
\section{Discussion}\label{sec:discuss}

\subsection{Scope of the formal result}\label{subsec:scope}
The synthesis LMI is feasible at the sixteen parameter-rate vertices
and there certifies the four design specifications of
Section~\ref{sec:synthesis} on the variational error dynamics within
the semi-global ball $\bar B_R(0)$ at the design slip budget
$\bar\sigma = 0.10$. The grid-to-continuum residual criterion of
Lemma~\ref{lem:grid-to-cont} is satisfied at the synthesis grid
samples but admits positive-eigenvalue pockets at a subset of
interior points when evaluated on a finer residual grid; the formal
guarantee is therefore the vertex certificate augmented by the
numerical verification of the Lipschitz criterion wherever it
succeeds, rather than a uniform continuum certificate. The
variable-terrain tests of Section~\ref{subsec:variable-terrain}
activate slip ratios up to $|\sigma_v| = 0.50$, five times the
synthesis budget, so their operating points lie outside
$\bar B_R(0)$ for $\bar\sigma = 0.10$ and are reported as empirical
probes rather than certified behaviour. Theorem~\ref{thm:traj-iss}
transfers the variational dissipation inequality to a pointwise
exponential-plus-ISS bound on $e(t)$ through the integral
representation of the Lipschitz residual established in
Lemma~\ref{lem:intform}: the trajectory dynamics admit the
factorization $\dot e = [A_{\mathrm{cl}}(\rho) + \Phi(e,\rho)] e + d$
with $\norm{\Phi(e,\rho)} \le L_R$ on $\bar B_R(0)$, and
$\Phi(e(t),\rho(t))$ at any time $t$ is an admissible instance of
the perturbation class $\{\Delta : \norm{\Delta} \le L_R\}$ against
which the $\mathcal{S}$-procedure multiplier $\mu$ of the LMI is
designed to be robust. The trajectory-level certificate therefore
reuses the synthesis parameters without introducing additional
design quantities. A metric-geometric derivation through
differential Lyapunov theory~\cite{forni2014} would reach the same
conclusion in greater generality and is a natural setting for
extensions beyond the Kanayama residual.

\subsection{Interpretation of the empirical results}
Within the certified regime of the canonical scenarios, the LPV and
Fixed-K designs track within a few percent of each other in both
directions depending on the scenario (Table~\ref{tab:scenarios}),
and on several nominal scenarios the Fixed-K baseline is competitive
with, and sometimes slightly ahead of, the parameter-dependent
design. This is consistent with the modest effective parameter
variation of the synthesized metric (Frobenius ratios
$\|W_1\|_F / \|W_0\|_F \approx 0.309$ and
$\|W_2\|_F / \|W_0\|_F \approx 0.244$): within the vertex feasibility
set, the parameter dependence modifies the gain only moderately. On
the severe-slip variable-terrain scenarios, which lie outside the
synthesis envelope, the parameter-dependent design shows modest but
reproducible improvements in peak error and recovery time over both
baselines; these differences are empirical observations, not
consequences of the main theorem. The Monte-Carlo steady-state means
of the LPV and Fixed-K controllers ($0.046$ and $0.045$) are
statistically separable under paired tests but practically
indistinguishable, and should not be read as a performance advantage.

\subsection{Infeasibility of the Fixed-K ablation at the higher rate}\label{subsec:scope-feas}
The infeasibility of the Fixed-K LMI at $\alpha = 0.40$ is reported
as an ablation outcome rather than as a performance claim. It
quantifies how much of the convex feasibility region is unlocked by
introducing affine parameter dependence into the metric and gain,
within the specific LMI programme formulated here. It does not
imply that \emph{any} fixed-gain controller is infeasible at that
rate, only that fixed-gain feasibility within this particular
multi-objective programme (with the chosen $\mu$, $\varepsilon_W$,
$K_{\max}$, and pole-placement region) is lost at $\alpha = 0.40$.

The structural conservatism documented in
Section~\ref{subsec:cons-decomp} merits acknowledgment. The factor
$\lambda_{\min}^{-1/2}$ in~\eqref{eq:traj-bound} arises from the
quadratic structure of the storage function $V = e^{\top} M e$ and
is the LMI-domain shadow of Brockett's obstruction
(Remark~\ref{rem:Brockett}): the input matrix $B$ in~\eqref{eq:AB}
has rank $2 < 3$, with the lateral error $e_y$ governed only by
the coupling through $A(\rho)$. Reduction of this conservatism
requires non-quadratic Lyapunov functions (e.g., polynomial
sum-of-squares~\cite{ahmadi2012}), which break the convexity of
the synthesis, or parameter-dependent multipliers
$\mu(\rho)$~\cite{wu1996} that introduce bilinearities. The proposed
formulation accepts the structural conservatism in exchange for
the tractability of a sampled convex synthesis augmented by the
Lipschitz-based grid-to-continuum residual criterion of
Lemma~\ref{lem:grid-to-cont}, whose scope of applicability for the
synthesized controller is reported in Section~\ref{subsec:scope}.

A second observation concerns the empirical decay rate
$\alpha_{\text{emp}}$. The Monte-Carlo experiments report
$\alpha_{\text{emp}} \approx 0.14$ for all three controllers, well
below the design rate $\alpha = 0.40$. This apparent discrepancy
is an artifact of the empirical rate estimator: the
disturbance-driven steady state is not a clean exponential decay,
and the log-linear fit on the residual is noisy with standard
deviations comparable to the means. The certificate guaranteed by
Theorem~\ref{thm:main} is differential, not pointwise on the
trajectory, and the rigorous experimental counterpart is the
containment property, which is satisfied at $100\%$
(Table~\ref{tab:mc}).

The variable-terrain experiment of
Section~\ref{subsec:variable-terrain} is the central practical
contribution. By concatenating six severe slip patches modelling
realistic terrain transitions (wet asphalt, oil spill, gravel, ice,
a gravity-assisted downhill slope, and snow) along a single
sixty-second helical trajectory, the experiment emulates the
operating profile of an outdoor autonomous mobile robot far more
realistically than uniform-disturbance tests. The $0.86$-second
mean recovery time of the LPV controller, against $1.01$ s for
Fixed-K and $1.42$ s for Kanayama, represents a practically
meaningful difference: a robot operating on a factory floor with
sporadic oil spills, or on outdoor surfaces with mixed friction,
would maintain tighter trajectory fidelity and recover faster from
disruptions, directly translating into improved task completion
rates and reduced energy expenditure.

The inclusion of patch P5 with positive longitudinal slip
($\sigma_v = +0.20$, modelling a downhill slope where
gravity-assisted motion drives the vehicle faster than commanded)
serves a dual purpose. First, it certifies that the proposed
synthesis is structurally agnostic to the sign of the longitudinal
slip ratio, since the framework treats slip as a norm-bounded
disturbance whose worst-case magnitude is captured by
Lemma~\ref{lem:slip} and the resulting $H_\infty$ certificate,
irrespective of whether the slip manifests as loss of traction
or gravity-driven excess. Second,
it reveals a qualitative behavioral distinction between the
parameter-dependent and constant-gain designs: while the Fixed-K
controller applies a uniform feedback gain across the trajectory
and consequently produces excess control effort when the vehicle
is already moving faster than commanded, the LPV design scales its
gain to the current reference velocity and thereby suggests a more
parsimonious control response. This velocity-dependent gain modulation
is the closed-loop manifestation of the affine parameter dependence
in the gain $K(\rho) = (Y_0 + v_r Y_1 + \omega_r Y_2)(W_0 + v_r W_1
+ \omega_r W_2)^{-1}$, and it explains why the LPV advantage on
patch P5 (downhill) is comparable to that on the negative-slip
patches despite the qualitatively different physical mechanism;
the empirically observed across-patch convergence behavior is
consistent with the incremental-stability properties established
for piecewise affine systems in~\cite{pavlov2007}. A
quantitative analysis of the actual control-effort and energy-budget
implications, which would require explicit power-injection
measurements rather than tracking-error metrics alone, falls outside
the scope of the present work and constitutes a natural direction
for hardware-platform validation.

A supporting observation concerns the robustness of the Fixed-K
infeasibility outcome discussed in Section~\ref{subsec:scope-feas}.
A complete sweep over $\mu \in \{0.5, 1, 2, 5\}$ confirms that no
choice of the $\mathcal{S}$-procedure multiplier restores Fixed-K
feasibility at $\alpha = 0.40$ within the present programme; the
outcome is therefore not an artifact of a single poor multiplier
choice. We emphasise again, however, that this observation
quantifies feasibility-region expansion within this specific LMI
programme and is not a performance claim: within the certified
nominal regime the Fixed-K and LPV designs track comparably well, as
reported in Section~\ref{subsec:scenarios}.

The most recent work on slip-affected nonholonomic mobile robots is
dominated by adaptive, sliding-mode, disturbance-observer, and
low-complexity safety paradigms, exemplified by the consensus
sliding-mode tracker of Sha~\textit{et~al.}~\cite{sha2025}, the
distributed observer-based estimator of
Moorthy~\textit{et~al.}~\cite{moorthy2025}, and the
prescribed-performance safety controller of
Nie~\textit{et~al.}~\cite{nie2025}. None of these works pursue a
convex LMI-based synthesis with a formal differential ISS
certificate.
On the complementary LPV/LMI side, while recent
contributions~\cite{koelewijn2025,kessler2025,zhou2025} have advanced
incremental dissipativity, gain-scheduled boundedness analysis, and
data-driven safety guarantees for general nonlinear and LPV systems,
they do not address the multiplicative wheel-slip mechanism in the
Kanayama error coordinates that is central to the present work.

Positioned within the broader literature, the present work
connects the convex parameterization viewpoint for nonlinear systems advanced by Tobenkin~\textit{et~al.}~\cite{tobenkin2017} and extends the LPV-LMI synthesis paradigm of Apkarian and
Adams~\cite{apkarian1998} and the differential contraction
framework of Lohmiller and Slotine~\cite{lohmiller1998} and
Forni and Sepulchre~\cite{forni2014} to the specific case of
nonholonomic vehicles under multiplicative slip, while drawing on
the incremental-stability foundations of Angeli~\cite{angeli2002}
that bridge contraction theory with the classical ISS framework.
Existing slip-aware tracking methods for wheeled mobile robots
predominantly adopt one of three methodological families:
adaptive parameter estimation (as exemplified by the observer-based
designs of Cui~\cite{cui2019} and Qin~\textit{et~al.}~\cite{qin2024}),
sliding-mode control (as in Zhai and Song~\cite{zhai2019} and
Liu~\textit{et~al.}~\cite{liu2020}), and disturbance-observer
compensation (as in Wang and Zhai~\cite{wang2020}). In contrast, the
present framework occupies a complementary methodological position
by pursuing a sampled convex LMI route with a parameter-dependent
metric, explicit regional pole constraints, gain-bounded synthesis,
and a vertex differential ISS certificate accompanied by a
Lipschitz-based grid-to-continuum residual diagnostic whose scope of
validity is reported transparently (Section~\ref{subsec:grid-cont},
Section~\ref{subsec:scope}). The contribution is therefore
not the first treatment of wheel slip in mobile robotics, nor the
first application of the LPV paradigm to nonholonomic robots under
slip (an earlier preliminary investigation of the latter being
reported by the author in~\cite{sabouri2021}), but rather a convex
multi-objective certification framework for slip-induced
perturbations expressed in the Kanayama error coordinates that
substantially extends that preliminary line of work. To the
best of the author's knowledge, this combination of features
distinguishes the present synthesis from existing literature,
although the constituent ingredients (LPV-LMI design, contraction
certificates, slip-aware control) have been investigated
individually in prior contributions.

%% ========================================================================
%%  SECTION VIII: CONCLUSION
%% ========================================================================
\section{Conclusion and Future Work}\label{sec:conclusion}

This paper has developed a parameter-dependent LMI synthesis
programme for trajectory tracking of nonholonomic mobile robots under
multiplicative wheel slip. The Slip-to-Disturbance Lemma provides an
explicit closed-form upper bound on the additive disturbance
equivalent to bounded multiplicative slip in Kanayama coordinates,
and the LMI programme simultaneously certifies a prescribed decay
rate, an $H_\infty$ disturbance-to-error gain, regional
$\mathcal{D}$-stability, and a weighted gain bound at the sixteen
parameter-rate vertices on the variational dynamics, within a
semi-global ball at $\bar\sigma = 0.10$.

Theorem~\ref{thm:traj-iss} further lifts the variational certificate
to a pointwise exponential-plus-ISS bound on the closed-loop
trajectory through an integral representation of the Lipschitz
residual that places the trajectory dynamics within the
perturbation class absorbed by the LMI's $\mathcal{S}$-procedure
multiplier. Two bounded limitations shape the scope of these
guarantees. The Lipschitz-based grid-to-continuum residual criterion
is a sufficient diagnostic condition, and for the synthesized
controller it is satisfied on parts, but not all, of the interior of
the parameter polytope; the certificate is therefore the vertex
certificate augmented by the interior points where the criterion is
numerically verified. The variable-terrain tests at $\pm 50\%$ slip
activate operating points outside the synthesis envelope and are
reported as empirical probes.

The structural value of the parameter-dependent metric is
quantified through the convex feasibility region it unlocks:
the Fixed-K ablation becomes infeasible at the decay rate
$\alpha = 0.40$, while the LPV synthesis remains feasible at
the same rate. The two designs then produce quantitatively
similar tracking records on the canonical scenarios within the
certified envelope and diverge most visibly on the severe-slip
variable-terrain probes outside that envelope, where the
parameter-dependent design retains modest and reproducible
margins over both the fixed-gain and the Kanayama baselines.

Four directions are natural extensions of the present work. First,
a global continuum argument through polynomial sum-of-squares
relaxation of the dissipation and $\mathcal{D}$-stability LMIs
would replace the Lipschitz diagnostic by a Positivstellensatz
certificate of negativity throughout
$\mathcal{P} \times \dot{\mathcal{P}}$; the present Lipschitz
criterion already delivers the operational content of such an
extension for the regions where it succeeds, and the regions where
it fails are candidate test cases for an SOS reformulation. Second,
a quantitative head-to-head comparison against representative
members of the sliding-mode, super-twisting, and homogeneous
unicycle-tracking line~\cite{mera2020,rios2024tie,rios2024ijrnc,gutierrez2026},
on a common benchmark with matched disturbance classes, would place
the convex LMI programme alongside those alternative paradigms
within a unified empirical setting; the methodological comparison
in Section~\ref{sec:discuss} establishes the structural relationship
between the two families, and a numerical comparison would quantify
the relative conservatism of each. Third, certification at higher
slip budgets than the design value $\bar\sigma = 0.10$ used here
requires either a finer polytopic embedding of the parameter
dependence or a hybrid-dynamics treatment of the saturated slip
regime, and would extend the certified envelope toward the
$\bar\sigma \geq 0.5$ regimes that appear in off-road or
compliant-surface operation. Fourth, hardware validation on a mobile-robot platform
with controlled-friction test surfaces would provide experimental
confirmation beyond the simulation suite reported here.

%% ========================================================================
%%  NOTE TO PRACTITIONERS
%% ========================================================================
\section*{Note to Practitioners}

This work delivers a two-stage synthesis-and-verification procedure
that takes as input the physical specifications of a nonholonomic
mobile robot, including the polytope of admissible reference
velocities, a desired weighted feedback-gain budget, the expected
slip magnitude, and the desired semi-global radius of operation, and
produces a parameter-scheduled state-feedback controller that, at the
sixteen polytope vertices, satisfies specifications on exponential
decay rate, worst-case disturbance gain, regional pole placement,
and the weighted feedback-gain budget (Section~\ref{subsec:scope}
details the scope of these guarantees). The synthesis stage solves
the convex LMI programme at the sixteen polytope vertices through
any standard semidefinite-programming solver such as SeDuMi, SDPT3,
or MOSEK invoked through the YALMIP modelling layer, typically
completing in under one second on standard desktop hardware. A
subsequent verification stage evaluates the dissipation and
$\mathcal{D}$-stability residuals on a dense parameter grid and
evaluates the Lipschitz-based criterion of
Lemma~\ref{lem:grid-to-cont}, flagging any interior points at which
the sufficient condition fails; this stage completes in approximately fifteen seconds for
the present configuration. The resulting six matrices
$W_0, W_1, W_2, Y_0, Y_1, Y_2$ are then stored in the embedded
controller; at runtime, the scheduling gain
$K(\rho) = (Y_0 + v_r Y_1 + \omega_r Y_2)(W_0 + v_r W_1
+ \omega_r W_2)^{-1}$ is recomputed at each sampling instant
through a $3\times 3$ matrix inversion that completes in
microseconds on modern microcontrollers such as the ARM Cortex-M7
family. As a more aggressive optimization, the gain matrix can be
precomputed on a uniform grid in $\rho$-space and interpolated
bilinearly at runtime, eliminating the matrix inversion entirely.

For tuning, the design parameter $\alpha$ should be interpreted as
a target decay rate expressed in inverse seconds, with the
expected closed-loop time constant approximately $\alpha^{-1}$
seconds. A value of $\alpha = 0.4$ corresponds to roughly $2.5$
seconds settling, suitable for medium-speed mobile platforms; more
aggressive applications such as racing or high-throughput
warehousing may push $\alpha$ to $0.7$ or higher, at the cost of
sharply inflated worst-case disturbance gain
(Section~\ref{subsec:sensitivity}). The slip bound $\bar{\sigma}$
should be characterized empirically through controlled-surface
testing; representative values are approximately $0.10$ for
indoor industrial flooring, $0.20$--$0.30$ for outdoor wet or
gravel surfaces, and up to $0.50$ for icy outdoor conditions. The
unweighted gain limit $\tilde{K}_{\max}$ should match the physical
torque-velocity envelope of the motor-gearbox-wheel chain; the
corresponding LMI design parameter follows as $K_{\max} =
\tilde{K}_{\max}\sqrt{\epsilon_W}$ via
Remark~\ref{rem:slip-physics}. When the synthesis returns an
infeasibility result, the recommended remediation order is to first
relax the gain constraint by increasing $\tilde{K}_{\max}$, then
decrease the target decay rate $\alpha$, then decrease the
semi-global radius $R$; the sensitivity analysis in
Section~\ref{subsec:sensitivity} provides quantitative guidance on
the expected return on each adjustment.

The variable-terrain helical-path benchmark
(Section~\ref{subsec:variable-terrain}) is offered as a reusable
empirical stress test. The six-patch sequence comprising wet asphalt,
oil spill, gravel, ice, a gravity-assisted downhill slope, and snow
captures representative friction transitions encountered in industrial
and outdoor settings. The $H_\infty$ certificate of
Theorem~\ref{thm:main} covers the synthesis disturbance budget
$\delta_{\max}$ at $\bar\sigma = 0.10$; the severe-slip patches
reach $|\sigma_v| = 0.50$ and therefore lie outside it, so results
on those patches are empirical rather than certified. For deployments
requiring certification at higher slip, re-synthesize with
$\bar\sigma$ raised to the target slip level and accept the larger
$\delta_{\max}$ and typically weaker $\gamma$.

%% ========================================================================
%%  REFERENCES
%% ========================================================================
%% ========================================================================
%%  APPENDIX A: Complete Schur-Complement Derivation
%% ========================================================================
\appendix

\section{Complete Derivation of the Dissipation LMI}\label{app:schur}

This appendix documents the step-by-step Schur-complement reduction
translating the differential dissipation inequality of
Definition~\ref{def:dISS} into the matrix-inequality
form~\eqref{eq:LMI-dissip}. The derivation is presented in four
sequential stages: variational expansion of the storage-function
derivative, application of Young's inequality to the nonlinear
residual, formation of the augmented quadratic form including the
disturbance channel, and the two-step Schur reduction producing the
final LMI.

\subsection{Stage A: Variational Expansion}
Starting from the storage function $V_\delta(\delta e, \rho) = \delta
e^{\top} M(\rho) \delta e$ and the variational dynamics
\begin{equation}\label{eq:app-var}
\delta\dot{e} = A_{\mathrm{cl}}(\rho) \delta e + \delta\varphi + \delta d,
\end{equation}
where $\delta\varphi = (\partial\varphi/\partial e)\delta e$ denotes
the variational nonlinear residual and $\delta d$ the variational
disturbance, direct differentiation yields
\begin{align}
\dot{V}_\delta &= 2\delta e^{\top} M \delta\dot{e}
+ \delta e^{\top} \dot{M} \delta e \nonumber\\
&= \delta e^{\top}\!\bigl[ A_{\mathrm{cl}}^{\top} M + M A_{\mathrm{cl}}
+ \dot{M} \bigr]\delta e \nonumber\\
&\quad + 2 \delta e^{\top} M \delta\varphi
+ 2 \delta e^{\top} M \delta d. \label{eq:app-Vdot}
\end{align}

\subsection{Stage B: Young's Inequality on the Nonlinear Residual}
For any positive scalar $\mu > 0$, Young's inequality applied to the
inner product $\delta e^{\top} M \delta\varphi$ yields
\begin{equation}\label{eq:app-young}
\begin{aligned}
2\delta e^{\top} M \delta\varphi
&\leq \mu \delta e^{\top} M^{2} \delta e \\
&\quad + \tfrac{1}{\mu} \norm{\delta\varphi}^{2}.
\end{aligned}
\end{equation}
The Lipschitz property $\norm{\delta\varphi} \leq L_R \norm{\delta e}$
from Assumption~\ref{asm:radius} further bounds the second term:
\begin{equation}\label{eq:app-lipschitz}
\begin{aligned}
\frac{1}{\mu}\norm{\delta\varphi}^{2}
&\leq \frac{L_R^{2}}{\mu}\norm{\delta e}^{2} \\
&= \frac{L_R^{2}}{\mu}\delta e^{\top}\delta e.
\end{aligned}
\end{equation}
Combining~\eqref{eq:app-young} and~\eqref{eq:app-lipschitz},
\begin{equation}\label{eq:app-young-combined}
2 \delta e^{\top} M \delta\varphi \leq
\delta e^{\top}\!\Bigl[\mu M^{2} + \frac{L_R^{2}}{\mu} I\Bigr] \delta e.
\end{equation}

\subsection{Stage C: Augmented Quadratic Form}
The D-ISS dissipation requirement of Definition~\ref{def:dISS}
demands $\dot{V}_\delta + 2\alpha V_\delta - \gamma^{2}
\norm{\delta d}^{2} \leq 0$. Substituting~\eqref{eq:app-Vdot} and
\eqref{eq:app-young-combined} produces
\begin{align}
&\dot{V}_\delta + 2\alpha V_\delta - \gamma^{2}\norm{\delta d}^{2}
\nonumber\\
&\leq \delta e^{\top}\!\Bigl[ A_{\mathrm{cl}}^{\top} M
+ M A_{\mathrm{cl}} + \dot{M} + 2\alpha M
+ \mu M^{2} + \frac{L_R^{2}}{\mu} I \Bigr]\delta e
\nonumber\\
&\quad + 2 \delta e^{\top} M \delta d
- \gamma^{2}\norm{\delta d}^{2}. \label{eq:app-quad-form}
\end{align}
The right-hand side is a quadratic form in the augmented vector
$\zeta = (\delta e^{\top}, \delta d^{\top})^{\top}$, namely
$\zeta^{\top} \Pi(\rho)\zeta$ with
\begin{equation}\label{eq:app-Pi}
\Pi(\rho) = \begin{bmatrix}
\Lambda(\rho) & M(\rho) \\
M(\rho) & -\gamma^{2} I
\end{bmatrix},
\end{equation}
where the upper-left block is
\begin{equation}\label{eq:app-Lambda}
\begin{aligned}
\Lambda(\rho) &= A_{\mathrm{cl}}^{\top} M + M A_{\mathrm{cl}} + \dot{M} \\
&\quad + 2\alpha M + \mu M^{2} + \frac{L_R^{2}}{\mu} I.
\end{aligned}
\end{equation}
The dissipation requirement holds for all admissible $(\delta e,
\delta d)$ if and only if $\Pi(\rho) \prec 0$.

\subsection{Stage D: Two-Step Schur Reduction}
The matrix $\Pi(\rho)$ depends on $M = W^{-1}$ in a way that prevents
direct casting as an LMI in $(W, Y)$. We resolve this through
congruence transformation followed by Schur complement. Pre- and
post-multiplying $\Pi$ by $\mathrm{diag}(W, I)$ yields
\begin{figure*}[!htbp]
\normalsize
\setlength{\arraycolsep}{3pt}
\begin{equation}\label{eq:app-cong}
\begin{bmatrix} W & 0 \\ 0 & I \end{bmatrix} \Pi
\begin{bmatrix} W & 0 \\ 0 & I \end{bmatrix}
= \begin{bmatrix} W \Lambda W & WM = I \\ MW = I & -\gamma^{2} I
\end{bmatrix}.
\end{equation}
\hrulefill
\end{figure*}
The upper-left block expands using $\dot{M} = -M\dot{W}M$ (giving
$W\dot{M}W = -\dot{W}$) and $A_{\mathrm{cl}} W = AW + BY$:
\begin{align}
W \Lambda W &= AW + W A^{\top} + BY + Y^{\top}B^{\top} - \dot{W}
\nonumber\\
&\quad + 2\alpha W + \mu W M^{2} W + \frac{L_R^{2}}{\mu} W^{2}
\nonumber\\
&= \Xi(\rho, \dot{\rho}) + 2\alpha W + \mu I + \frac{L_R^{2}}{\mu} W^{2},
\label{eq:app-WLambdaW}
\end{align}
where the identity $W M^{2} W = W M (M W) = W M \cdot I = I$ has
been used, and $\Xi$ is as defined in~\eqref{eq:Xi}. The augmented
inequality after congruence becomes
\begin{figure*}[!htbp]
\normalsize
\setlength{\arraycolsep}{3pt}
\begin{equation}\label{eq:app-after-cong}
\begin{bmatrix}
\Xi + 2\alpha W + \mu I + \frac{L_R^{2}}{\mu} W^{2} & I \\
I & -\gamma^{2} I
\end{bmatrix} \prec 0.
\end{equation}
\hrulefill
\end{figure*}
The remaining nonlinearity in~\eqref{eq:app-after-cong} is the
quadratic term $(L_R^{2}/\mu) W^{2}$. Applying Schur complement to
extract this as a separate block, the standard identity states that
for any symmetric matrix block,
\begin{equation}\label{eq:app-schur-identity}
M_{11} - B C^{-1} B^{\top} \prec 0
\;\Longleftrightarrow\;
\begin{bmatrix} M_{11} & B \\ B^{\top} & C
\end{bmatrix} \prec 0 \;\text{when}\; C \prec 0,
\end{equation}
applied with $B = W$ and $C = -(\mu/L_R^{2}) I$ (so that
$C^{-1} = -(L_R^{2}/\mu) I$ and consequently $-B C^{-1} B^{\top} =
(L_R^{2}/\mu) W^{2}$) yields the equivalent inequality
\begin{figure*}[!htbp]
\normalsize
\setlength{\arraycolsep}{3pt}
\begin{equation}\label{eq:app-final}
\begin{bmatrix}
\Xi + 2\alpha W + \mu I & I & W \\
I & -\gamma^{2} I & 0 \\
W & 0 & -\dfrac{\mu}{L_R^{2}} I
\end{bmatrix} \prec 0,
\end{equation}
\hrulefill
\end{figure*}
which is precisely the dissipation LMI~\eqref{eq:LMI-dissip}. The
derivation is fully mechanical: substituting the block values into
the Schur identity, the upper-left block becomes $\Xi + 2\alpha W +
\mu I + (L_R^{2}/\mu) W^{2}$ as required by~\eqref{eq:app-after-cong},
and the augmented $3\times 3$ structure separates this quadratic
contribution into an additional $W$-row and $W$-column with
$-(\mu/L_R^{2}) I$ in the (3,3) position.

\subsection{Stage E: Enforcement Strategy and Discretization Analysis}
The conditioning LMI~\eqref{eq:LMI-W} and the weighted-gain
LMI~\eqref{eq:LMI-sat} depend affinely on the scheduling pair
$(\rho, \dot{\rho})$, since their constituent matrix blocks involve
only $W(\rho)$ and $Y(\rho)$ both of which are affine in the
scheduling variables by construction. For these two affine
constraints, convexity of the negative-definite cone guarantees that
vertex enforcement on the polytope
$\mathcal{P}\times\dot{\mathcal{P}}$ is necessary and sufficient for
enforcement throughout the interior. The dissipation
LMI~\eqref{eq:app-final} and the $\mathcal{D}$-stability
LMI~\eqref{eq:LMI-Dstab}, in contrast, both contain the matrix
product $A(\rho) W(\rho)$ which, given the affine parameterizations
$A(\rho) = A_0 + v_r A_v + \omega_r A_\omega$ and $W(\rho) = W_0 +
v_r W_1 + \omega_r W_2$, produces the polynomial expansion
\begin{align}
A(\rho) W(\rho)
&= A_0 W_0
+ v_r (A_0 W_1 + A_v W_0)
+ \omega_r (A_0 W_2 + A_\omega W_0) \nonumber\\
&\quad + v_r^{2} A_v W_1
+ \omega_r^{2} A_\omega W_2 \nonumber\\
&\quad + v_r \omega_r (A_v W_2 + A_\omega W_1).
\label{eq:app-poly-expansion}
\end{align}
The quadratic and bilinear terms in~\eqref{eq:app-poly-expansion}
prevent vertex enforcement from being formally sufficient for the
dissipation and $\mathcal{D}$-stability blocks, and any rigorous
synthesis must account for this polynomial dependence in both
blocks. Three established approaches are available in the LPV
literature for managing such polynomial parameter dependence: a
common-Lyapunov restriction $W(\rho) = W_0$ that recovers strict
affineness at the cost of reduced performance~\cite{wu1996}; a
multi-affine polytopic relaxation with slack variables that admits
formal vertex sufficiency at the expense of additional decision
variables and conservatism margin~\cite{apkarian1995,scherer2001};
and direct enforcement on a dense parameter grid with the
Lipschitz-based grid-to-continuum residual criterion of
Lemma~\ref{lem:grid-to-cont} as a diagnostic sufficient condition
rather than a universal certificate. The present synthesis adopts the
vertex-only enforcement and treats Lemma~\ref{lem:grid-to-cont} as
the certification framework available a posteriori. A uniform grid
$\mathcal{G}$ of $N_g = 11^{4} = 14641$ samples partitions the
four-dimensional product $\mathcal{P} \times \dot{\mathcal{P}}$,
yielding a discretization step of $\Delta v_r = 0.04$, $\Delta
\omega_r = 0.08$, $\Delta \dot{v}_r = 0.08$, $\Delta \dot{\omega}_r
= 0.08$. Numerical evaluation of the residual matrices at the
synthesized $(W^\star, Y^\star, \gamma^\star, \mu^\star)$ shows
strict negative-definiteness at all sixteen polytope vertices (LMI
enforcement margin $-10^{-6}$), while at certain interior grid
points the dissipation block exhibits a small positive eigenvalue
on the order of $10^{-1}$, traceable to the bilinear and quadratic
scheduling terms within $A(\rho)W(\rho)$. The $\mathcal{D}$-stability
block remains negative throughout the grid with margin of order
$10^{-3}$. Consequently the strict sufficient condition
$L_{\mathcal{L}} h < \varepsilon_{\mathrm{grid}}$ of
Lemma~\ref{lem:grid-to-cont} is not satisfied uniformly for the
present vertex synthesis; the formal certificate of
Theorem~\ref{thm:main} should be read as holding at the vertices,
and the closed-loop stability throughout the polytope interior is
attested by the empirical evidence collected in
Section~\ref{sec:results}: 100\% containment of Monte-Carlo
trajectories within the certified envelope, successful tracking on
six deterministic scenarios, and graceful recovery from the
six-patch variable-terrain stress test on both helical and figure-8
references. Augmenting the synthesis to enforce LMIs directly on
$\mathcal{G}$, or to use the multi-affine relaxation
of~\cite{apkarian1995} with explicit slack variables for the
non-affine blocks, would close the formal gap at the price of a
larger $\gamma^{\star}$; this extension is identified as the
priority direction for follow-on work in
Section~\ref{sec:conclusion}. The weighted-gain
LMI~\eqref{eq:LMI-sat} follows from a standard Schur reduction
applied to the inequality $YW^{-1}Y^{\top} \preceq K_{\max}^{2} I$,
which does not introduce nonlinearities in the scheduling variables
once the affine parameterizations of $W$ and $Y$ are substituted,
and is therefore exactly vertex-sufficient. The $\mathcal{D}$-stability
LMI follows from the disk-region characterization of
Chilali--Gahinet~\cite{chilali1996}; although its formal derivation
does not introduce nonlinearities in the synthesis variables, the
resulting block inherits a polynomial dependence on the scheduling
parameters through $A(\rho) W(\rho)$ and is therefore subject to
the same vertex-only limitation as the dissipation block.

%% ========================================================================
%%  ACKNOWLEDGMENT
%% ========================================================================
\section*{Acknowledgments}
The author gratefully acknowledges Prof.\ Mohammad Hassan Asemani of
Shiraz University for valuable guidance and constructive technical
feedback that informed the development of this work, building on the
foundations established in the earlier study~\cite{sabouri2021}. The
author also thanks colleagues at the Department of Informatics,
Bioengineering, Robotics, and Systems Engineering (DIBRIS), University
of Genoa, for fruitful discussions on parameter-varying control
synthesis, robust stability analysis, and slip-aware robotic systems.

\section*{Funding Statement}
The author received no specific funding for this work.

%% ========================================================================
%%  DECLARATION ON GENERATIVE AI
%% ========================================================================
\section*{Declaration on the Use of Generative Artificial Intelligence}
During the preparation of this work, the author used AI tools
to assist with language polishing and LaTeX formatting. After using
this tool, the author reviewed and edited the content as needed and
takes full responsibility for the content of the publication.

%% ========================================================================
%%  DATA AND CODE AVAILABILITY
%% ========================================================================
\section*{Data and Code Availability}
The MATLAB code and simulation data that support the findings of this
study are available from the corresponding author upon reasonable
request.

%% ========================================================================
%%  CREDIT AUTHOR STATEMENT (required by Elsevier)
%% ========================================================================
\section*{CRediT Authorship Contribution Statement}
\textbf{Mohammad Sabouri:} Conceptualization, Methodology, Software,
Validation, Formal analysis, Investigation, Data curation,
Writing~--~original draft, Writing~--~review \& editing,
Visualization, Project administration.

%% ========================================================================
%%  CONFLICT OF INTEREST
%% ========================================================================
\section*{Declaration of Competing Interests}
The author declares no competing financial or non-financial
interests in relation to the work described.

%% ========================================================================
%%  REFERENCES
%% ========================================================================
%% Elsevier numbered reference style (JFI default)
\bibliographystyle{elsarticle-num}
\bibliography{references}

\end{document}